\documentclass[12pt, draftclsnofoot, onecolumn]{IEEEtran}
\usepackage{graphicx}
\usepackage{amsthm}
\usepackage{epsfig}
\usepackage{latexsym}
\usepackage{amsfonts}
\usepackage{here}
\usepackage{rawfonts}
\usepackage[latin1]{inputenc}
\usepackage[T1]{fontenc}
\usepackage{calc}
\usepackage{capitalgreekitalic}
\usepackage{url}
\usepackage{enumerate}
\usepackage{color}
\usepackage[tbtags]{amsmath}
\usepackage{amssymb}
\usepackage{upref}
\usepackage{epic,eepic}
\usepackage{times}
\usepackage{dsfont}
\usepackage{comment}
\usepackage{cite}
\usepackage{bbm} 
\usepackage{optidef}
\usepackage{empheq}
\newtheorem{proposition}{\bf Proposition}
\newtheorem{lemma}{\bf Lemma}

\newtheorem{definition}{\bf Definition}

\usepackage{dsfont}
\newcounter{step}
\newlength{\totlinewidth}
\newenvironment{algorithm}{%
  \rule{\linewidth}{1pt}
  \begin{list}{}%
    {\usecounter{step}%
      \settowidth{\labelwidth}{\textbf{Step 2:}}%
      \setlength{\leftmargin}{\labelwidth}%
      \setlength{\topsep}{-2pt}%
      \addtolength{\leftmargin}{\labelsep}%
      \addtolength{\leftmargin}{2mm}%
      \setlength{\rightmargin}{2mm}%
      \setlength{\totlinewidth}{\linewidth}%
      \addtolength{\totlinewidth}{\leftmargin}%
      \addtolength{\totlinewidth}{\rightmargin}%
      \setlength{\parsep}{0mm}%
      \raggedright}}%
  {\end{list}%
  \rule{\linewidth}{1pt}}
\newcounter{substep}

  {\end{list}}

\newlength{\aligntop}
\newlength{\alignbot}
 \makeatletter
\renewenvironment{align}{%
  \vspace{\aligntop}
  \start@align\@ne\st@rredfalse\m@ne
}{%
  \math@cr \black@\totwidth@
  \egroup
  \ifingather@
    \restorealignstate@
    \egroup
    \nonumber
    \ifnum0=`{\fi\iffalse}\fi
  \else
    $$%
  \fi
  \ignorespacesafterend%
  \vspace{\alignbot}\par\noindent
} \makeatother

\IEEEoverridecommandlockouts

\usepackage{algorithm}%,algpseudocode}
\usepackage{algorithmic}
\usepackage{subfigure}
\begin{document}
%\clearpage
\title{\huge Joint Access and Backhaul Resource Management in Satellite-Drone Networks: A Competitive Market Approach}%Competitive Market for Joint Access and Backhaul Resource Allocation in Satellite-Drone Networks\vspace{-0.26cm}}

\author{ Ye Hu,  {Mingzhe Chen}, Walid Saad, \emph{Fellow, IEEE} \vspace*{-2em}\\ 
%M\'erouane Debbah\IEEEauthorrefmark{3}, and Choong-Seon Hong\IEEEauthorrefmark{4}\vspace*{0em}\\
%\authorblockA{\small \IEEEauthorrefmark{1}Beijing Laboratory of Advanced Information Network, Beijing University of Posts and Telecommunications, Beijing, China 100876,\\ Emails: \protect\url{chenmingzhe@bupt.edu.cn}, \protect\url{ccyin@ieee.org.} \\
%\IEEEauthorrefmark{2}Centre for Telecommunications Research, Department of Informatics, Kings College
%London, WC2B 4BG, UK,\\ Email: \protect\url{yang.zhaohui@kcl.ac.uk}\\
%\IEEEauthorrefmark{3}Wireless@VT, Bradley Department of Electrical and Computer Engineering, Virginia Tech, Blacksburg, VA, USA, Email: \protect\url{walids@vt.edu.}\\
%\IEEEauthorrefmark{4}SRIBD and The Chinese University of Hong Kong, Shenzhen, China, Email: \protect\url{shuguangcui@cuhk.edu.cn}.\\
%\IEEEauthorrefmark{5}Department of Electrical and Computer Engineering University of California, Davis, CA, USA, Email: \protect\url{sgcui@ucdavis.edu}.\\
%\IEEEauthorrefmark{6} Department of Electrical Engineering, Princeton University, Princeton, NJ, USA, Email: \protect\url{poor@princeton.edu}.\\
%\IEEEauthorrefmark{3}\small Mathematical and Algorithmic Sciences Lab, Huawei France R \& D, Paris, France, \\Email: merouane.debbah@huawei.com.\\
%\IEEEauthorrefmark{4}\small Department of Computer Science and Engineering, Kyung Hee University, Yongin, South Korea, \\Email: \protect\url{cshong@khu.ac.kr.}\\
%}\vspace*{-2em}
\thanks{A preliminary version of this work was presented at IFIP NTMS \cite{hu2019competitive}.}
\thanks{Y. Hu, W. Saad are with the Wireless@VT, Bradley Department of Electrical and Computer Engineering, Virginia Tech, Blacksburg, VA, USA, Emails: \protect{yeh17@vt.edu}, \protect{walids@vt.edu}.}
\thanks{M. Chen is with the Department of Electrical Engineering, Princeton University, Princeton, NJ, USA,  and the Future Network of Intelligence Institute, Chinese University of Hong Kong, Shenzhen, China, Email: \protect{mingzhec@princeton.edu}.}}

\maketitle
 \vspace{-0.55cm} 
\begin{abstract}
\boldmath
In this paper, the problem of user association and resource allocation is studied for an integrated satellite-drone network (ISDN). In the considered model, drone base stations (DBSs) provide downlink connectivity, supplementally,  to ground users whose demand cannot be satisfied by terrestrial small cell base stations (SBSs). Meanwhile, a satellite system and a set of terrestrial macrocell base stations (MBSs) are used to provide resources for backhaul connectivity for both DBSs and SBSs. For this scenario, one must jointly consider resource management over satellite-DBS/SBS backhaul links, MBS-DBS/SBS terrestrial backhaul links, and DBS/SBS-user radio access links as well as user association with DBSs and SBSs. This joint user association and resource allocation problem is modeled using a competitive market setting in which the transmission data is considered as a good that is being exchanged between users, DBSs, and SBSs that act as ``buyers'', and DBSs, SBSs, MBSs, and the satellite that act as ``sellers''. In this market, the quality-of-service (QoS) is used to capture the quality of the data transmission (defined as good), while the energy consumption the buyers use for data transmission is the cost of exchanging a good. 
According to the quality of goods, sellers in the market propose quotations to the buyers to sell their goods, while the buyers purchase the goods based on the quotation.
%users in the market propose quotations to DBSs and SBSs to buy their goods, while the DBSs and SBSs propose quotations to  MBSs, and the satellite to buy their goods. %??????
%Meanwhile, sellers will sell goods to the buyers who propose the highest quotation, with a minimum energy cost. 
The buyers profit from the difference between the earned QoS and the charged price, while the sellers profit from the difference between earned price and the energy spent for data transmission. The buyers and sellers in the market seek to reach a Walrasian equilibrium, at which all the goods are sold, and each of the devices' profit is maximized.  
%MBSs, and the satellite seek to sent out requested data in energy efficient ways. %The average data rates are defined as the users, DBSs and SBSs' profits, the energy efficiency is defined as the DBSs, SBSs, MBSs and the satellite' profits.
%the users, DBSs or SBSs seeks to request such goods with optimal average data rates. On the other hand, each of the DBS, SBS, MBS or satellite seeks to find an optimal data allocation scheme for a maximal profit, in terms of their power consumptions. 
%The problem of social welfare maximization is divided into multiple subproblems with dual decomposition, and then, solved in a distributed way at each device. 
A heavy ball based iterative algorithm is proposed to compute the Walrasian equilibrium of the formulated market. 
Analytical results show that, with well-defined update step sizes, the proposed algorithm is guaranteed to reach one Walrasian equilibrium. Simulation results show that, at the achieved Walrasian equilibrium solution, the proposed algorithm can yield a two-fold gain in terms of  the number of radio access links with a data rate of over $40$ Mbps, and a three-fold gain in terms of the number of backhaul links  with a data rate greater than $1.6$ Gbps. 
 
 \end{abstract}
%\newpage

% IEEEtran.cls defaults to using nonbold math in the Abstract.
% This preserves the distinction between vectors and scalars. However,
% if the conference you are submitting to favors bold math in the abstract,
% then you can use LaTeX's standard command \boldmath at the very start
% of the abstract to achieve this. Many IEEE journals/conferences frown on
% math in the abstract anyway.

\renewcommand{\thefootnote}{\fnsymbol{footnote}}

%\footnotetext[1]{This research was supported in part by the National Research Foundation for the Doctoral Program of Higher Education of China (No.20120005110007), the NSFC (No. 61271257), and by the U.S. National Science Foundation under Grant 1513697.}
%\thanks{This research was supported by the U.S. National Science Foundation under Grant 1513697.}

% keywords

%\begin{keywords}
%uplink-downlink decoupling, LTE-unlicensed, user association, learning, echo state networks;        
%\end{keywords}
%\renewcommand{\thefootnote}{\fnsymbol{footnote}}

%\footnotetext[1]{}
%This work is supported by the National Natural Science Foundation of China (61271177)and the Fundamental Research Funds for the Central Universities.
% For peer review papers, you can put extra information on the cover
% page as needed:
% \ifCLASSOPTIONpeerreview
% \begin{center} \bfseries EDICS Category: 3-BBND \end{center}
% \fi
%
% For peerreview papers, this IEEEtran command inserts a page break and
% creates the second title. It will be ignored for other modes.
\IEEEpeerreviewmaketitle
 \vspace{-0.35cm}  

\section{Introduction}

Unmanned aerial vehicles (UAVs), popularly known as drones, provide an effective approach to complement the connectivity of terrestrial wireless networks by serving as flying drone base stations (DBSs) that provide  coverage to hotspots, disaster-affected, or rural areas \cite{chen2017caching,8766136, saad2019vision,mozaffari2018beyond,mozaffari2017mobile}. However, the limited-capacity of terrestrial wireless backhaul links can significantly limit the quality-of-service (QoS) provided by DBSs. To overcome this backhaul challenge, satellite communication systems can be leveraged to provide high data rate backhaul services for DBSs with seamless wide-area coverage as discussed in \cite{jia2016broadband} and \cite{zhao2018beam}. %Meanwhile, unmanned aerial vehicle (UAV) technologies provide new ways for managing and operating the integrated space-air-ground network.
%This integration of the satellite and drone base stations (DBSs) within existing networks is referred to as the \emph{Integrated satellite-drone networks (ISDNs)}. 
 However, to integrate satellite systems with DBSs and terrestrial networks, one must overcome many challenges such as spectrum resource management, network modeling, and cross-layer power management. In particular, spectrum sharing and resource management for terrestrial communication systems and satellite backhaul communication systems is a major challenge for integrated satellite/terrestrial networks,  given the nonlinear coupling between terrestrial links and satellite backhaul links \cite{wei2014key, jaber20165g}. 
\vspace{-0.35cm}
\subsection{Related Works}
 
The existing literature such as in \cite{gapeyenko2018flexible, cicek2018backhaul, nguyen2018novel, 8614433, galkin2017stochastic, alzenad2018fso,kalantari2017user,kalantari2017backhaul} has studied a number of problems related to resource allocation for drone-based systems with backhaul considerations. The work in \cite{gapeyenko2018flexible} investigates the problem of dynamic link rerouting in a UAV-based relay network. In\cite{cicek2018backhaul}, the authors optimize the 3D locations of UAVs as well as the bandwidth allocation so as to maximize data rates. 
 In\cite{nguyen2018novel}, the authors investigate the problem of radio resource allocation in a UAV-assisted network, while employing a non-orthogonal multiple access (NOMA) scheme on the wireless backhaul transmission. In \cite{8614433}, a joint caching and resource allocation is investigated for a network of cache-enabled UAVs.  The authors in \cite{galkin2017stochastic} study the existence of an optimal UAV height that meets backhaul requirements at the UAVs while maximizing the coverage probability of the ground user. The feasibility of a novel vertical backhaul/fronthaul framework in which the unmanned flying platforms transport the backhaul/fronthaul traffic between the access and core networks via point-to-point free space optics links is investigated in \cite{alzenad2018fso}.
 The work in \cite{kalantari2017user} develops a novel algorithm to find efficient 3D locations for the DBSs and optimize the bandwidth allocation and user association so as to maximize the sum-rate of the users. 
 In \cite{kalantari2017backhaul}, the authors study how network-centric and user-centric wireless backhaul solutions can affect the number of served users.
 Despite their promising results, the works in\cite{gapeyenko2018flexible, cicek2018backhaul, nguyen2018novel, 8614433,galkin2017stochastic, alzenad2018fso,kalantari2017user,kalantari2017backhaul} do not consider the implementation of an \emph{integrated satellite-drone network (ISDN)}, in which satellite systems provide backhaul connectivity for UAVs. \iffalse{\color{blue}Using satellite systems to provide the backhaul connectivity for UAVs can ......... }\fi%To enable the satellites to provide the backhaul services for UAVs, one needs to consider the resource allocation between existing ground backhaul communication systems and satellite backhaul communication system.
 
Meanwhile, the works in \cite{zhao2018beam} and \cite{li2018investigation,palazzi2004satellite, palazzi2005enhancing} have studied a number of problems related to ISDNs.  The authors in \cite{zhao2018beam} propose a blind beam tracking approach for a Ka-band UAV-satellite communication system. %, with the UAV navigation resulted unstable beam pointing being considered. 
 The work in \cite{li2018investigation} studies the Doppler effect, the pointing error effect, and the atmospheric turbulence effect on the communication performance of a UAV-to-satellite optical communication system. The work in \cite{palazzi2004satellite} introduces an innovative architecture in which high altitude platform drones and UAVs are deployed to improve the visibility of satellites.
The authors in \cite{palazzi2005enhancing} develop an architecture in which high altitude platform drones are connected to a satellite to enhance telecommunication capabilities. 
%The authors in \cite{zhao2018integrating} explore the opportunities and challenges of combining communications and control in UAV systems. %systems, including beam tracking with joint mechanical and electrical adjustment for UAV satellite communications.
%In \cite{palazzi2005enhancing}, the authors approach an innovative and more challenging architecture foreseeing HAPS/UAV connected to the satellite in order to enlarge coverage and to allow interconnection with very remote locations.
 However, the works in \cite{zhao2018beam} and \cite{li2018investigation, palazzi2004satellite, palazzi2005enhancing} do not consider any spectrum sharing problem in an ISDN system. Indeed, sharing the stringent spectrum resources among satellite backhaul links, terrestrial backhaul links, and the radio access links has a direct impact on the achievable ISDN data rates.
 Hence, the resource allocation problem over all communication links must be jointly studied within the context of an ISDN. % \cite{hu2019competitive}, a Hungarian-based solution is proposed to solve the problem of spectrum allocation and user association in the ISDN. However, a centralized method could yield significant delay and overhead to the system and could introduce privacy concerns.
\vspace{-0.35cm}
\subsection{Contributions}
The main contribution of this paper is a novel framework for jointly managing resources across radio access links, satellite backhaul links, and terrestrial backhaul links, while maximizing the data rates in ISDNs, in a distributed way. %(Satellite-SBS and Satellite-DBS links) to terrestrial backhaul connections (MBS-SBS and MBS-DBS links), and user association between DBSs, SBSs, and terrestrial users.  %the soldier's mobility along with the psychological factors pertaining to both the soldier and the attacker. 
This joint backhaul and access resource management problem is formulated as a competitive market in which the data transmission service, including both radio access data and backhaul data, is viewed as good that must be exchanged among the wireless users who seek to maximize their profits. 
%The work in \cite{duan2017distributed} uses a Walrasian equilibrium characterized joint pricing and task allocation scheme to maximize social welfare of an unified mobile crowd-sensing system.
While prior works such as \cite{mochaourab2011walrasian, mochaourab2012walrasian, duan2017distributed} used market models to study resource allocation in different scenarios, those works have not jointly analyzed drone-assisted resource allocation in wireless networks with backhaul considerations. In contrast, here, we need a new joint user association and resource allocation scheme tailored to the ISDN system and whose goal is to optimize the data rates at each communication link and the sum rate in the system with the consideration of the data demand at each user, the resource budget of each base station (BS), and the fairness among each users, BSs, and the satellite's benefits. 
 Unlike previous market-based solutions such as \cite{mochaourab2011walrasian} and \cite{duan2017distributed} that only consider resource allocation among only two communication links which do not include backhaul links, we propose a total payoff maximizing solution that optimizes the sum-rate of the ISDN network, considering the resource allocation for satellite backhaul links, macrocell base station (MBS) backhaul links, and communication links between terrestrial users and their serving small cell base stations (SBSs) or DBSs. Our key contributions include: %$\psi$
\begin{itemize}
\item We develop a novel framework to jointly consider resource allocation over satellite-DBS/SBS backhaul links, terrestrial backhaul links, and DBS/SBS-user radio access links as well as user association with DBSs and SBSs within an ISDN. In this system, the SBSs and DBSs provide downlink data service to terrestrial users, while the MBSs, supplemented with a low earth orbit (LEO) satellite system, provide backhaul service to SBSs and DBSs.

\item We formulate this joint ISDN user access and resource allocation problem as a competitive market, in which the users request downlink data service with high data rates, while the SBSs and DBSs seek to satisfy these data service requests with minimal power consumption. Meanwhile, the SBSs and DBSs request a backhaul service with high data rates, while the MBSs and the LEO satellite seek to provide them with a backhaul service with minimal power consumption. A Walrasian equilibrium is applied to solve the formulated market. In this regard, we prove the existence of the Walrasian equilibrium, at which the communication performance in the system are optimized.

%\item We perform fundamental analysis on the soldier's and attacker's psychology in the battlefield using the framework of psychological game theory\cite{battigalli2009dynamic}. In the formulated psychological game, the psychology of the players (i.e. the soldier and the attacker) is modeled as their intention to \emph{frustrate} each other. The frustration of the players is quantified as the gap, if positive, between their expected payoff and actual payoff. A psychological equilibrium (PE) is used to solve the psychological IoBT game. In this regard, we prove the uniqueness of the PE for our proposed psychological game, under the same set of conditions at which the NE is unique. In addition, our analytical results show that, in an attempt to frustrate the soldier, at the PE, the attacker is more prone to attack the IoBT device with the best channel conditions.%In addition, our analytical results show that, when attempting to frustrate the attacker, the soldier can get a higher material payoff in PE, compared to the one at the NE.  

\item We propose a distributed, iterative algorithm to optimize the data rates at each user, SBS, DBS, MBS and the satellite, based on dual decomposition.  Our analytical results show that, with well defined update step sizes, the proposed algorithm is guaranteed to reach a Walrasian equilibrium at which the total payoff gained by all the ISDN devices is maximized.
%and with well defined updating step sizes, the proposed algorithm with converge with a linear rate.% such that the optimal solution for not only the sum-rates and the data rates at each communication links in the ISDN system is reached.

%establish the players' belief system, so as to solve the proposed psychological IoBT game. In this regard, the algorithm characterizes what is known as an $\epsilon$-like psychological self-confirming equilibrium (PSCE) of our proposed psychological game.

\item The results also show that, the proposed heavy ball based algorithm yields the same performance as a centralized optimization algorithm and a $10\%$ improvement on convergence speed, compared to a sub-gradient algorithm. Simulation results also show that, by considering the Walrasian equilibrium in the market, the proposed algorithm can yield over two-fold gain in terms of  the number of radio access links with over $40$ Mbps rates, and over three-fold gain in terms of the number of backhaul links with over $1.6$ Gbps rates. %The heavy ball algorithm also yields $10\%$ improvement on convergence speed, compared to the sub-gradient algorithm.. 
\end{itemize}

The rest of this paper is organized as follows. The system model are described in Section \uppercase\expandafter{\romannumeral2}. In Section \uppercase\expandafter{\romannumeral3}, the problem formulation and proposed framework for joint backhaul and access resource management is developed and discussed. In Section \uppercase\expandafter{\romannumeral4}, numerical simulation results are presented. Finally, conclusions are drawn in Section \uppercase\expandafter{\romannumeral5}.
 
\begin{figure}
\setlength{\belowcaptionskip}{-2pt}
\setlength{\abovecaptionskip}{-2pt} 
  \centering
  \includegraphics[width=10 cm]{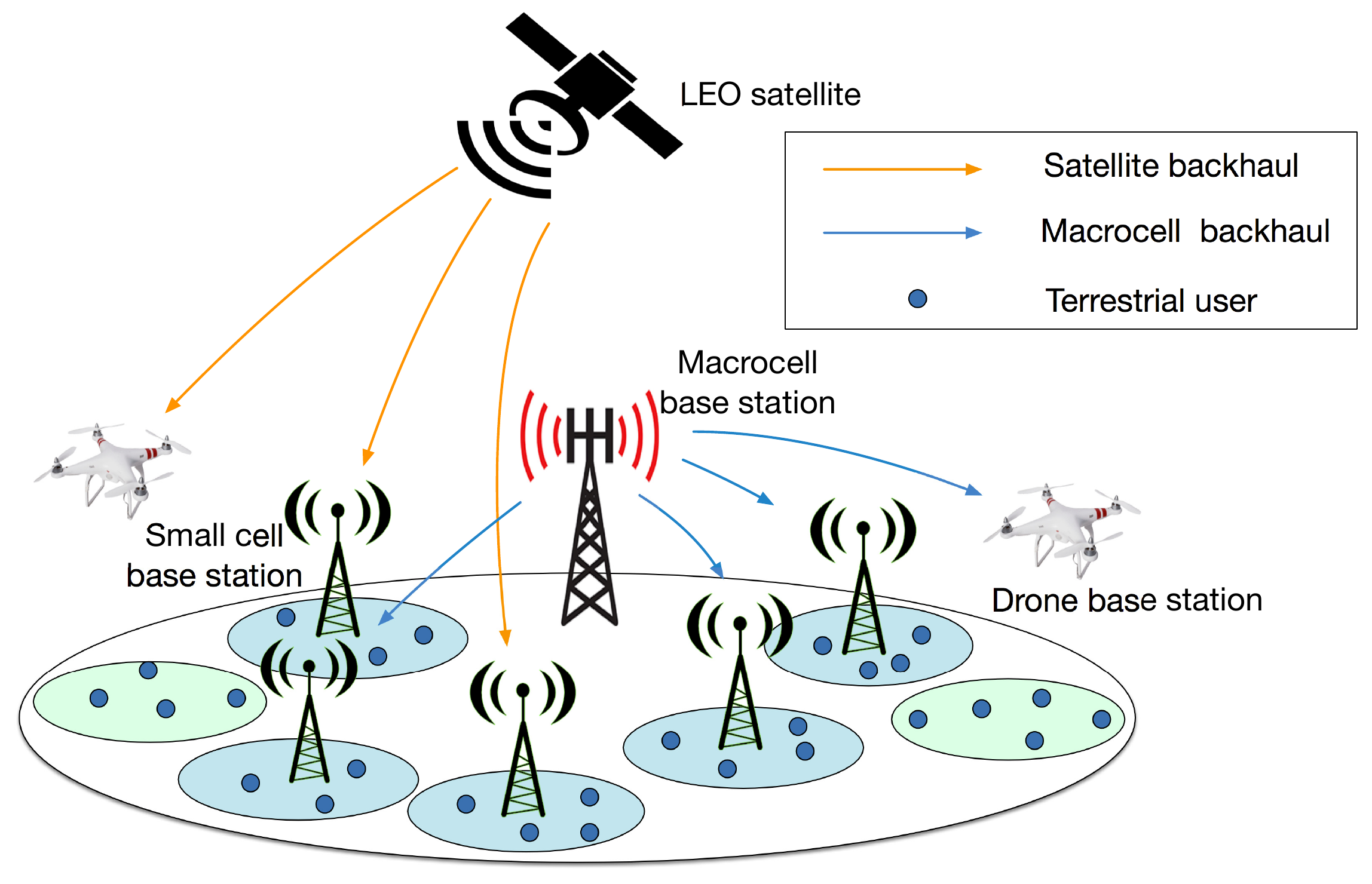}
  \caption{\footnotesize{Illustration of the studied integrated satellite-drone network topology.}}
  \label{Fig. 1}
  \centering
  \vspace{-0.5cm}
\end{figure}

\vspace{-0.12cm}
\section{System Model}
\vspace{-0.12cm}
\begin{table*}\scriptsize
  \newcommand{\tabincell}[2]{\begin{tabular}{@{}#1@{}}#2\end{tabular}}
\renewcommand\arraystretch{1}
 \caption{
    \vspace*{-0.3cm}List of notations}\vspace*{-1em}
\centering  
\begin{tabular}{|c||c|c||c|}% ±íÊ¾¸÷ÁÐÔªËØ¶ÔÆë·½Ê½£¬left-l,right-r,center-c
%\hline
%\textbf{Type} & \textbf{Power} & \textbf{Cache} & \textbf{Moving}&\textbf{Path loss} & \textbf{Frequency bands}&\textbf{Interference} \\
\hline
Notation & Description & Notation & Description\\
\hline
$U$ & Number of users & $h\left(i,j\right)$ & Channel gain between transmitter $i$ and receiver $j$ \\
\hline
$S$ & Number of SBSs & $d_{i,j}$& Distance between transmitter $i$ and receiver $j$\\
\hline
$E$ & Number of DBSs & $L_1$& Rician fading channel coefficient\\
\hline
$M$ & Number of MBSs & $L_2\left(d_{i,j}\right)$& Large-scale channel effects between $i$ and $j$\\
\hline
$N$ & Number of BSs & $\kappa_{i,j}$& Offset angle of $i$ in direction of $j$\\
\hline
$H_n$ & Hover time of DBS $n$ &$G_M\left(\kappa_{i,j}\right)$ & Transmit gain at offset angle $\kappa_{i,j}$  \\
\hline
$\Delta$& Studied time duration  &  $G_R\left(\kappa_{i,j}\right)$ & Receive gain at offset angle $\kappa_{i,j}$ \\
\hline
 $C_u$ & Data service requested by user $u$&$\Omega_{i,j,t}$ & Interference at link between $i$ and $j$ at slot $t$\\
\hline
$T$ & Number of time slots within $\Delta$ &$\gamma_{i,j,t}$ & SINR at link between $i$ and $j$ at slot $t$ \\
\hline
$T_n$ & Number of time slots within $H_n$ &$c_{i,j,t}$ & Data rate at link between $i$ and $j$ at slot $t$  \\
\hline
$\tau$ & Duration of one time slot & $\Pi_i$ & QoS requirement at device $i$\\
\hline
$\mathcal{I}^B$ & Set of links between BS $n$ and user $u$ & $q_{i,j,t}$ & Unit price of time slots at link between $i$ and $j$ \\
\hline
$\mathcal{I}^M$ & Set of links between MBS $m$ and BS $n$ & $V_i$ & Payoff of seller $i$ \\
\hline
$\mathcal{I}^S$ & Set of links between the satellite and BS $n$  & $W_j$& Payoff of buyer $j$ \\
\hline
 $\boldsymbol{\rho}$ & Vector of allocation scheme at $\mathcal{I}^B$& $L_i$ & Lagrangian function at device $i$\\
\hline
$\boldsymbol{\delta}$& Vector of allocation scheme at $\mathcal{I}^M$& $D_i$ & Dual function at device $i$\\
\hline
  $\boldsymbol{\beta}$& Vector of allocation scheme at $\mathcal{I}^S$& $\boldsymbol{\lambda}$ & Vector of dual variables at $\mathcal{I}^B$\\
\hline
$\boldsymbol{\theta}$ & Vector of request scheme at $\mathcal{I}^B$& $\boldsymbol{\varsigma}$ & Vector of dual variables at $\mathcal{I}^M$\\
\hline
$\boldsymbol{\varphi}$& Vector of request scheme at $\mathcal{I}^M$& $\boldsymbol{\xi}$ & Vector of dual variables at $\mathcal{I}^S$\\
\hline
$\boldsymbol{\varepsilon}$& Vector of allocation scheme at $\mathcal{I}^S$& $\pi^{\left(k\right)}$ & Updating step size at iteration $k$\\
\hline
%$x_{t,i},y_{t,i}$ & Coordinates of users   & $P_B$ & Transmit power of the BBUs\\
%\hline
\end{tabular}
\vspace{-0.8cm}
\end{table*}
Consider a downlink ISND that consists of a set $\mathcal{S}$ of $S$ SBSs, a set $\mathcal{E}$ of ${E}$ DBSs supplementing the SBSs, a set $\mathcal{M}$ of $M$ MBSs providing terrestrial backhaul links, a LEO satellite constellation providing a satellite backhaul system, and a set $\mathcal{U}$ of $U$ users as shown in Fig. \ref{Fig. 1}.
%Consider a geographical area within which $U$ terrestrial users in a set $\mathcal{U}$ are randomly deployed and request downlink data services. An ISDN that consists of a set $\mathcal{S}$ of $S$ SBSs, a set $\mathcal{E}$ of ${E}$ DBSs supplementing the SBSs, a set $\mathcal{M}$ of $M$ MBSs providing terrestrial backhaul links, and a LEO satellite constellation providing a satellite backhaul system, is deployed to serve the users in this area as shown in Fig. \ref{Fig. 1}. %A DBS is also deployed to supplement the ground network by providing service to the terrestrial users that are not satisfied with the service provided by SBSs, as shown in Fig. \ref{Fig. 1}. 
%A set $\mathcal{U}$ of $U$ terrestrial users are randomly distributed in an area, requesting downlink service from the ISDN system. 
%The location of each user $u\in\mathcal{U}$ is assumed to be fixed and noted as $\left(\alpha^U_u, \beta^U_u, h^U_u\right)$. The location of each SBS $n\in\mathcal{N}$ is $\left(\alpha^N_n, \beta^U_n, h^N_n\right)$, the location of each MBS $m\in \mathcal{M}$ is $\left(\alpha^M_m, \beta^M_m, h^M_m\right)$. 
In this system, each user requests data service from either an SBS or a DBS whose hover time is limited by its propulsion energy consumption, while the backhaul connections of the SBSs and DBSs are provided by either the LEO satellite or ground MBSs. Let $H_n$ be the hover time of DBS $n$, defined as the time duration that a DBS can use to hover over a certain area to serve the ground users\cite{mozaffari2017wireless}.
Hereinafter, we refer to the DBSs and SBSs in the network as BSs and we define $\mathcal{N}$ as the set of $N={S}+{E}$ BSs. The three-dimensional (3D) coordinates of the users, BSs, and MBSs are defined as $\boldsymbol{S}_{1}=\left[\boldsymbol{s}_{1,1},\cdots,\boldsymbol{s}_{1,U}\right]$, $\boldsymbol{S}_{2}=\left[\boldsymbol{s}_{2,1},\cdots,\boldsymbol{s}_{2,N}\right]$, and $\boldsymbol{S}_{3}=\left[\boldsymbol{s}_{3,1},\cdots,\boldsymbol{s}_{3,M}\right]$, respectively.  Each element $\boldsymbol{s}_{i,j}$ captures the 3D coordinates of a specific location, i.e.,  $\boldsymbol{s}_{i,j}=\left(x_{ij},y_{ij},f_{ij}\right)$. {The satellite will start at an initial location with 3D coordinate $\boldsymbol{s}_{0}=\left(x_{0},y_{0},f_{0}\right)$, and move in the direction of the x-axis with speed $\upsilon$}.
%In this section, the problem of resource allocation and user associated in the ISDN system, and the design of the DBS trajectory within a certain time duration is studied. %In such a case, the DBS flies around in the ISDN system to serve the the users that can not be satisfied by SBSs' downlink service. %That is, after serving the terrestrial users in a certain duration, the DBSs will have to fly away for charging.  %This service duration of the DBSs, defined as the DBSs' service time limitation, is denoted as $T_s$.
 %due to the limitation in the maximum effective isotropically radiated power (EIRP) density of the current Ka band satellite systems \cite{panagopoulos2004satellite}.

We assume that the amount of data (in bits) that a user $u$ requests is $C_u$, which has to be provided to the user within a time duration ${\Delta}$. During ${\Delta}$, the BSs also request backhaul data service to download the requested data from the core network. Here, we refer the amount of data that each user requests as \emph{data service} and the amount of backhaul data that each BS requests as \emph{backhaul service}. The time required to satisfy each user's data service request and each BS's backhaul service request depends on several factors, including the effective service time and the data rate at the corresponding user and BS. In the considered scenario, all the communication links work over the Ka band ($26.5$ -- $40$ GHz), which is a well established millimeter wave (mmW) range suitable for satellite communication and future 5G links, as discussed in \cite{rappaport2013millimeter}. 
%In the considered downlink scenario, the SBSs and MBSs serve their associated user or BSs using the sub-6 GHz frequency band and the LEO satellite serves its associated BSs using the Ka band ($26.5$ - $40$ GHz), which is a well established millimeter wave range suitable for satellite communication and future 5G links, as done in \cite{rappaport2013millimeter}. Moreover, the DBSs in the ISDN also use of the Ka band to serve their associated users, to avoid high interference at sub-6 GHz frequency band from terrestrial cellular network. Note that, the DBSs and the SBSs in the system are assumed to be dual-mode receivers that can jointly receive millimeter wave (mmW) and sub-6 GHz backhaul resources from the satellite and MBSs.
%In the considered downlink scenario, the BSs, MBSs and LEO satellite serve their associated users or BSs using the Ka band ($17.7$-$19.7$ GHz), which is a well established mm-wave rangesuitable for satellite communication and future 5G links, as done in \cite{rappaport2013millimeter}.
%The transmission over radio access links from DBSs to users and backhaul links from the satellite to BSs is operated in a time division mode, such that each user or BS is served in consecutive time slots. 
A time division multiple access (TDMA) scheme is adopted to support directional transmissions at the significantly high path loss mmW band, while maintaining low complex designs for transceivers\cite{semiari2017joint}. 
%The time division multiple access scheme is also deployed to sub-6 GHz transmission on radio access links from SBSs to users and backhaul links from MBSs to BSs to designate the time occasions when sub-6 GHz transmission is performed in the ISDN system. 
Note that, in this model, as done in \cite{panagopoulos2004satellite}, the downlink interference caused by the satellite on the terrestrial links is considered to be negligible due to the satellite's limited transmit power and long-range transmission. %Also, omni-directional antennas are used at the DBSs to maintain low overhead and complexity.
%The transmission on radio access links (links from BSs to users) and backhaul links (links from MBSs or the satellite to BSs) is operated in a time division mode which is widely used in mmWave standards such as 802.11ad, such that each user or BS is served in consecutive time slots. 

Moreover, the considered time duration ${\Delta}$ is further divided, by the BSs, MBSs, and satellite into a set $\mathcal{T}$ of $T$ time slots, each of which has a time duration of $\tau$. Within duration ${\Delta}$, each user (for radio transmission) or BS (for backhauling) will be served at least once in at least one time slot which is the \emph{effective service time} of this user or BS. Note that, among the $T$ time slots in the studied duration ${\Delta}$, each DBS $n\in \mathcal{E}$ is only operational for $T_n=\left\lfloor {\frac{H_n}{\tau}} \right\rfloor$ time slots, as limited by its hover time. We also assume that each DBS is operational from the start of the studied duration, and will hover at its location for $T_n$ time slots.  We also assume that the SBSs, MBSs, and satellite will be operational during the whole process. In other words, $T_n=T$, $T_m=T$ and $T_S=T$ for all $n\in\mathcal{S}$ and $m\in\mathcal{M}$. The \emph{data rates} at the communication links in the ISDN system are modeled next.

\vspace{-0.45cm}
\subsection{Channel Model}

In the studied ISDN, the mmWave links are affected by the surrounding obstacles and, thus, these communication links can be either line-of-sight (LoS) or non-line-of-sight (NLoS). 
%Here, as done in \cite{}, a probabilistic path loss model is considered.
The channel gain between transmitter $i$, which could be a BS, MBS, or the satellite, and its associated user (or BS for backhaul) $j$ is defined as\cite{semiari2017joint}:
\begin{equation}\label{eq:channel}
\begin{split}
h\left(i, j\right) = \left\{ {\begin{array}{*{20}{c}}
{{L_{1}10^{\frac{-L_2\left(d_{i,j}\right)}{10}}G_{\textrm{M}}\left(\kappa_{i,j}\right)G_{\textrm{R}}\left(\kappa_{j,i}\right)}\;\;\;\;\;\;\varpi_{i,j}=1},\;\;\\
{{0\;\;\;\;\;\;\;\;\;\;\;\;\;\;\;\;\;\;\;\;\;\;\;\;\;\;\;\;\;\;\;\;\;\;\;\;\;\;\;\;\;\;\;\;\;\;\;\;\varpi_{i,j}=0}}, 
\end{array}} \right.
%h\left(n, u\right)=10^{-0.1\left(L_{FS}\left(d_0\right)+10\mu\log\left(d^{l_i}_{u}\right)+\chi_{\delta}\right)}.
\end{split}
\end{equation} 
%$h\left(l_i, u\right)=10^{-0.1\left(L_{FS}\left(d_0\right)+10\mu\log\left(d^{l_i}_{u}\right)+\chi_{\delta}\right)}$ is the path loss from transmitter $l_i$ to user $u$, considering shadow fading. 
where $L_{1}$ is the Rician fading channel coefficient, and $L_{2}\left(d_{i,j}\right)=\alpha'+\alpha\log_{10}d_{i,j}+\chi$ captures the large-scale channel effects over the mmW link between $i$ and $j$\cite{ghosh2014millimeter}. Here, $\alpha$ is the slope of the fit and $\alpha'$, the intercept parameter, is the path loss (dB) for 1 meter of distance.  In addition, $\chi$ models the deviation
in fitting (dB) which is a Gaussian random variable with zero mean and variance $\xi ^2$  for $1$ meter of distance. $d_{i,j}$ is the distance between $i$ and $j$. For example, if $i\in\mathcal{U}$ and $j\in\mathcal{N}$, $d_{i,j}=\sqrt {\left(x_{1,i}-x_{2,j}\right)^2+\left(y_{1,i}-y_{2,j}\right)^2+\left(f_{1,i}-f_{2,j}\right)^2} $. 
$G_{\textrm{M}}\left(\kappa_{i,j}\right)$ is the transmit gain of an antenna at an offset angle $\kappa_{i,j}$, which is defined as:
\begin{equation}\label{eq:txgain}
\begin{split}
G_{\textrm{M}}\left(\kappa_{i,j}\right)  = \left\{ {\begin{array}{*{20}{c}}
{{Q_{\textrm{MM}}},\;\;\;~\textrm{if~$\kappa_{i,j}\le\frac{\kappa_{\textrm{T}}}{2}$}},\\
{{\;\;\;\;Q_{\textrm{MS}},\;\;\;\;\;~\textrm{otherwise}}}, \;\;
\end{array}} \right.
\end{split}
\end{equation} 
where $Q_{\textrm{MM}}$ and $Q_{\textrm{MS}}$ are the transmit antenna gains of the main lobe and the side lobe, respectively. $\frac{\kappa_{\textrm{T}}}{2}$ is the half power beam width.  $\kappa_{i,j}$ is the offset angle (from the boresight direction) of $i$'s transmit antenna in the direction of $j$'s receiver antenna. $G_R\left(\kappa_{j,i}\right)$ is the receive gain of an antenna at an offset angle $\kappa_{j,i}$, which is defined as:
\begin{equation}\label{eq:txgain}
\begin{split}
G_{\textrm{R}}\left(\kappa_{j,i}\right)  = \left\{ {\begin{array}{*{20}{c}}
{{Q_{\textrm{RM}}},\;\;\;~\textrm{if~$\kappa_{j,i}\le\frac{\kappa_{\textrm{R}}}{2}$}},\\
{{\;\;\;\;Q_{\textrm{RS}},\;\;\;~\textrm{otherwise}}}, 
\end{array}} \right.
\end{split}
\end{equation} 
where $Q_{\textrm{RM}}$ and $Q_{\textrm{RS}}$ are the receive antenna gains of the main lobe and the side lobe, respectively. $\frac{\kappa_{\textrm{R}}}{2}$ is the half power beam width.  $\kappa_{j,i}$ is the offset angle (from the boresight direction) of the link $j$'s receive antenna in the direction of $i$'s transmitter antenna. %$h\left(i, u\right)$ is the path loss from transmitter $i$ to user $u$.
$\varpi_{i,j}=1$ indicates that the link between $i$ and $j$ is LoS, otherwise, it is NLoS.  In fact, $\varpi_{i,j}$ is a Bernoulli random variable with probability of success $\iota  _{i,j}=e^{\phi_{\textrm{los}} d_{i,j}}$, where $\phi_{\textrm{los}}$ is a parameter determined by the density and the average size of the blockage obstacles in the surrounding area. 
Note that, as done in \cite{zeng2016wireless}, we assume that all the users have LoS paths to DBSs as the DBSs are assumed to hover at relatively high altitudes and the probability of having scatterers around DBS-user link is very small. That is, the probability of  a LoS path $\iota _{i,j}=1$ when $i\in\mathcal{E}$. Similarly, the probability of LoS path at the satellite backhaul link $\iota _{i,j}=1$, as $j\in\mathcal{N}$ and $i$ being the satellite.

In our model, the users and BSs will estimate the interference caused by the surrounding communication links in the system to find the the optimal association and allocation scheme that maximizes their data rates. In particular, the path loss between the interfering BSs or MBSs to the victim users or BSs is estimated based on (1), in which the bore sight angles of interfering transmitters shall be estimated.
 However, estimating interference in an ISDN is not trivial, as it can yield significant overhead and delay for the users and BSs, so that they can gather information on all resources allocated to neighboring links. 
Thus, instead of gathering such information, the bore sight angles of interfering transmitters are assumed, at each victim BS or user, to be independently and uniformly distributed in $\left(0, 2\pi\right)$. In such a case, random transmit antenna gain $G_{\textrm{M}}\left(\kappa_{i,j}\right)$ of interfering transmitters can yield an average value of $\frac{\kappa_{\textrm{T}}}{2\pi}Q_{\textrm{MM}}+\left(1-\frac{\kappa_{\textrm{T}}}{2\pi}\right)Q_{\textrm{MS}}$. The random receive antenna gain $G_{\textrm{R}}\left(\kappa_{i,j}\right)$ of interfering transmitters can yield an average value of $\frac{\kappa_{\textrm{R}}}{2\pi}Q_{\textrm{RM}}+\left(1-\frac{\kappa_{\textrm{R}}}{2\pi}\right)Q_{\textrm{RS}}$. 
Note that, the average value of transmit antenna gain and receive antenna gain is adopt, respectively, as the value of interfering transmit antenna gain and victim receive antenna gain in all the following interference analysis.

\vspace{-0.45cm}
\subsection{Terrestrial Link Data Rate Analysis}

Next, we model the transmission links from BSs to users and the backhaul links from MBSs to BSs, which we call \emph{terrestrial links} hereinafter. 
%In this network, BSs serve their associated users using the Ka band, as done in \cite{rappaport2013millimeter}. We assume that, the transmission on terrestrial links is operated in a time division mode over the Ka band, such that each user or BS is served in consecutive time slots with time duration $t^T$. 
%at each time slot in the ISDN system, the Ka band is divided into $K$ equally-sized frequency sub-channels with bandwidth $B^N$. 
The link between a BS $n\in\mathcal{N}$ and a user $u$ is denoted by $I^{B}_{n,u}$. The link between BS $n$ and MBS $m$ is denoted by $I^{M}_{n,m}$. 
The set of \emph{transmission links} between BSs and users is defined as $\mathcal{I}^{B}$, while the set of \emph{backhaul links} between MBSs and BSs is denoted by $\mathcal{I}^{M}$. %During the whole studied time duration, a set of $T$ time slots is available for each SBS is $\mathcal{T}$, where $t \in \mathcal{T}$ indicates one time slot. Note that, the set of $T_F$ time slots in each of BSs' frame is noted as $\mathcal{T}_F$.
%We assume that the SBSs can always provide services to the terrestrial users, while DBS $n$ can only serve the terrestrial users at most $\left\lfloor {\frac{\tau_n}{t^T}} \right\rfloor$ time slots. Here, as done in \cite{}, the service duration of DBS $n$, $\tau_n$, known as the DBS's service time limitation, is given as:
% \begin{equation}\label{eq:0}
%\begin{split}
%%N_{u}=\sum\limits_{i\in\mathcal{I}\left(b_k\right),i \ne I_{n,u}}O^{i}_{n,u}P_{i} h\left(i\right),
%\tau_n=\frac{E_mV}{c_1V^4+c_2},
%%\sum\limits_{\substack{n\in\mathcal{N},\\k\in\mathcal{K}}}\rho\left(n,u,k\right)=1, \sum\limits_{\substack{m\in\mathcal{M},\\k\in\mathcal{K}}}\rho\left(n,m,k\right)=1, k\in\mathcal{K}.
%\end{split}
%\end{equation} 
%where $V$ is the flight speed of the DBSs, $c_1$ and $c_2$ are two parameters related to the DBSs' weight, wing area, air density, etc. 
%Let $\theta^T_{n,t}\in\left\{0, 1\right\}$ be the availability vector of BS $n$ at time slot $t$, with $\theta^T_{n,t}=1$ representing that BS $n$ is working at time slot $t$. Note that, for each SBSs $n$, $\theta^T_{n,t}=1$ for every $t$, for each DBSs $n$, $\theta^T_{n,t}=1$ only for the time slots falling into its service duration. 

Let $\boldsymbol{a}={\left[{a}_{111},\cdots, {a}_{NUT}\right]}$ be an allocation vector for the transmission links from BSs to users and let $\boldsymbol{w}={\left[{w}_{111},\cdots, {w}_{ NMT}\right]}$ be an allocation vector for the transmission links from MBSs to BSs. In vectors $\boldsymbol{a}$ and $\boldsymbol{w}$, ${a}_{nut}\in\left\{0, 1\right\}$ and ${w}_{nmt}\in\left\{0, 1\right\}$ represent the time slot allocation state at link $I^{B}_{n,u}$ and $I^{M}_{n,m}$, with ${a}_{nut}=1$ and ${w}_{nmt}=1$ indicating that time slot $t$ is assigned to links $I^{B}_{n,u}$ and $I^{M}_{n,m}$, respectively. Let $\boldsymbol{b}={\left[{b}_{11},\cdots, {b}_{NU}\right]}$ and $\boldsymbol{z}=\tiny{\left[{z}_{11},\cdots, {z}_{NM}\right]}$ be, respectively, the association vectors for the transmission links in $\mathcal{I}^{B}$ and backhaul links in $\mathcal{I}^{M}$. Here, ${b}_{nu}\in\left\{0, 1\right\}$ represents the association state between user $u$  and BS $n$ and ${z}_{nm}\in\left\{0, 1\right\}$ represents the association state between terrestrial MBS $m$ and BS $n$, with ${b}_{nu}=1$ and ${z}_{nm}=1$ indicating that user $u$ associates with BS $n$ and BS $n$ associates with  MBS $m$, respectively. 
%Let ${X}={\left[\boldsymbol{x}_{1}, \boldsymbol{x}_{2}\right]}$ be the matrix of terrestrial time allocation schemes, in which $\boldsymbol{x}_{1}={\left[{x}_{1,111},\cdots, {x}_{1,NUT}\right]^T}$ and $\boldsymbol{x}_{2}={\left[{x}_{2,111},\cdots, {x}_{2, NMT}\right]^T}$ represent, respectively, the allocation vectors at transmission links in $\mathcal{I}^{B}$ and backhaul links in $\mathcal{I}^{M}$. In the vectors $\boldsymbol{x}_{1}$ and $\boldsymbol{x}_{2}$, ${x}_{1,nut}\in\left\{0, 1\right\}$ and ${x}_{2,nmt}\in\left\{0, 1\right\}$ represent time slot allocation state at link $I^{B}_{n,u}$ and $I^{M}_{n,m}$, with ${x}_{1,nut}=1$ and ${x}_{2,nmt}=1$ indicating that time slot $t$ is assigned to the link $I^{B}_{l,u}$ and $I^{M}_{l,m}$, respectively. Let $\boldsymbol{y}={\left[\boldsymbol{y}_{1},\boldsymbol{y}_{2} \right]^T}$ be the set of user association vectors, with $\boldsymbol{y}_{1}={\left[{y}_{1,11},\cdots, {y}_{1,NU}\right]^T}$, $\boldsymbol{y}_{2}=\tiny{\left[{y}_{2,11},\cdots, {y}_{2,NM}\right]^T}$ representing, respectively, the association vectors at transmission links in $\mathcal{I}^{B}$ and backhaul links in $\mathcal{I}^{M}$. Here, ${y}_{1,nu}\in\left\{0, 1\right\}$ represents the association state between user $u$  and BS $n$, ${y}_{2,nm}\in\left\{0, 1\right\}$ represents the association state between terrestrial MBS $m$ and BS $n$, with ${y}_{1,nu}=1$ and ${y}_{2,nm}=1$ indicating that user $u$ associates with BS $n$, and BS $n$ associates with  MBS $m$, respectively. \[\sigma \]
We use  $\boldsymbol{\rho}={\left[{\rho}_{111},\cdots, {\rho}_{NUT}\right]}$ and $\boldsymbol{\delta}={\left[{\delta}_{111},\cdots, {\delta}_{NMT}\right]}$ to represent, respectively, the allocation schemes adopt by transmission links in $\mathcal{I}^B$ and backhaul links in $\mathcal{I}^M$, with $\rho_{nut}=a_{nut}b_{nu}=1$ indicating that user $u$ is served by BS $n$ at time slot $t$, $\delta_{nmt}=w_{nmt}z_{nm}=1$ indicating that BS $n$ is served by MBS $m$ at time slot $t$. Recall that each DBS $n\in\mathcal{E}$ is only operational for $T_n$ time slots from the start of the studied period, such that ${\rho}_{nut}=0$, ${\delta}_{nmt}=0$, for all $u\in\mathcal{U}$, $m\in\mathcal{M}$, $n\in\mathcal{E}$, and $t>T_n$.

In this case, the interference at link $I^{B}_{n,u}$ that uses time slot $t$ can be expressed as:
 \begin{equation}\label{eq:interference1}
\begin{split}
%\Omega_{u}=\sum\limits_{i\in\mathcal{I}\left(b_k\right),i \ne I_{n,u}}O^{i}_{n,u}P_{i} h\left(i\right),
&\Omega_{n,u,t}\left(\boldsymbol{\rho}, \boldsymbol{\delta}\right)  = \sum\limits_{\substack{i\in\mathcal{I}^B,i\neq I^{B}_{n,u}}}\!\!\!\!\!\!\!\rho_{n_iu_it}P_{n_i}h\left(n_i, u\right)+\!\!\!\sum\limits_{\substack{i\in\mathcal{I}^M\\}}\delta_{n_im_it}P_{m_i}h\left(m_i, u\right),
%\left\{ {\begin{array}{*{20}{c}}
%{{\sum\limits_{\substack{i\in\mathcal{I}^B,i\neq I^{B}_{n,u}\\n_i\in\mathcal{S}}}\!\!\!\!\!\!\!\rho_{n_iu_it}P_{n_i}h\left(n_i, u\right)+\!\!\!\sum\limits_{\substack{i\in\mathcal{I}^M\\}}\delta_{n_im_it}P_{m_i}h\left(m_i, u\right)},~\textrm{if~$n\in\mathcal{S}$}},\\
%{{\;\;\;\;\;\;\;\;0,\;\;\;\;\;\;\;\;\;\;\;\;\;\;\;\;\;\;\;\;\;\;\;\;\;\;\;\;\;\;\;\;\;\;\;\;\;\;\;\;\;\;\;\;\;\;\;\;\;\;\;\;\;\;\;\;\;\;\;\;\;\;\;\;\;\;\;\;\;~\textrm{if~$n\in\mathcal{E}$},\;\;}}
%\end{array}} \right.
%\sum\limits_{\substack{i\in\mathcal{I}^B/I^{B}_{n,u}\\n_i\in\mathcal{E}}}\!\!\!\!\!\!\!\rho_{1, n_iu_it}P_{n_i}h\left(n_i, u\right)
%\Omega_{n,u,t}\left(\boldsymbol{\rho}, \boldsymbol{\delta}\right)= 
%&\!\!\!\!\!\!\sum\limits_{\substack{i\in\mathcal{I}^B,i\neq I^{B}_{n,u}\\}}\!\!\!\!\!\!\!\rho_{1, n_iu_it}P_{n_i}G\left(\theta^{i}_{n,u}\right)h\left(n_i, u\right)\\
%&+\sum\limits_{\substack{i\in\mathcal{I}^M\\}}\rho_{2, n_im_it}P_{m_i}G\left(\theta^{i}_{n,u}\right)h\left(m_i, u\right),%G^r_t\left(\theta^{n,u}_{i}\right),
\end{split}
\end{equation} 
where $u_i$, $n_i$, and $m_i$ represent, respectively, the terrestrial user, BS, and MBS in link $i$.
%user and BS in link $i$, when link $i$ is a transmission link in $\mathcal{I}^B$. $m_i$ represents the MBS in link $i$, when link $i$ is a backhaul link in $\mathcal{I}^M$. %Here, if link $i$ is a BS-user link, $\rho_{l_i,u_i,t}=\rho^{B}_{l_i,u_i,t}$, if link $i$ is an MBS-BS backhaul link, $\rho_{l_i,u_i,t}=\rho^{M}_{l_i,u_i,t}$. 
%$O^{i}_{n,u}={b^k_{k\left(i\right)}}/{B^N}$ is the compensation factor defined as the portion of the interfering sub-channel $k\left(i\right)$ that overlaps with the victim sub-channel bandwidth $k$.  %for interference from terrestrial link $i$. 
%Here, $k\left(i\right)$ represents the sub-channel allocated to terrestrial link $i$, in other words, ${x}\left(n, u, k\left(I^{B}_{n,u}\right)\right)=1$. $b^k_{k\left(i\right)}$ is the bandwidth of interfering sub-channel $k\left(i\right)$ \iffalse spectral density\fi that overlaps with the bandwidth of victim sub-channel $k$. 
$P_{n_i}$ is the transmit power of the BS at link $i$ and $P_{m_i}$ is the transmit power of the MBS at link $i$.
%Note that, if $l_i\in\mathcal{N}$, $P_{l_i}=P^{\textrm{S}}_{l_i}$, if $l_i\in\mathcal{M}$, $P_{l_i}=P^{\textrm{M}}_{l_i}$, where $P^{\textrm{S}}_{l_i}$ is the transmit power of SBS $l_i$, $P^{\textrm{M}}_{l_i}$ is the transmit power of MBS $l_i$. Also, if $l_i=N+1$, that is, BS $l_i$ is the DBS, $P_{l_i}=P^{\textrm{E}}$, where $P^{\textrm{E}}$ is the transmit power of the DBS. 

%$h\left(n_i, u\right)=10^{-0.1\left(L_{FS}\left(d_0\right)+10\mu\log\left(d^{n_i}_{u}\right)+\chi_{\delta}\right)}$ is the path loss from transmitter $n_i$ to user $u$, where $L_{FS}\left(d_0\right)=20\log\left(\frac{4\pi d_0f_{k}}{\upsilon}\right)$ is the free space path loss, with $\upsilon$ being the light speed, $d_0$ being the free-space reference distance, and $f_{k}$ being the the carrier frequency. $d^{n_i}_u$ represents the distance between the transmitter $n_i$ to user $u$. $\mu$ represents the path loss exponent, ${\chi_{\delta}}$ is Gaussian random variable in dB with zero mean and variance $\delta$. Note that, as done in \cite{zeng2016wireless}, we assume that all the users have line-of-sight (LoS) paths to DBSs as the DBSs are assumed to hover at relatively high altitudes and the probability of having scatterers around DBS-user link is very small. 
Thus, the signal-to-interference-plus-noise ratio (SINR) of link $I^{B}_{n, u}$ is computed as: 
 \begin{equation}\label{eq:2}
\setlength{\abovedisplayskip}{3 pt}
\setlength{\belowdisplayskip}{3 pt}
\begin{split}
&\gamma_{n,u,t}\left(\boldsymbol{\rho}, \boldsymbol{\delta}\right)=\frac{P_nh\left(n,u\right)}{\Omega_{n,u,t}\left(\boldsymbol{\rho}, \boldsymbol{\delta}\right)+\sigma^2},
\end{split}
\end{equation} 
where \iffalse$h\left(n,u\right)$ represents the free space path loss at link $I^{B}_{n,u}$. \fi$\sigma^2$ is the noise power. %Also, $P^B_l=P^D$ represents the transmit power of the DBS, if $l=N+1$, otherwise, $P^B_l=P^S_l$ represents the transmit power of the SBS $l$. 
%As done in \cite{zeng2016wireless}, we only consider the LoS transmission links over DBS-user links. 
The data rate at this link $I^{B}_{n,u}$ will be given by: %$c_{n,u}\left(\boldsymbol{\rho}, \boldsymbol{\delta}, \boldsymbol{\delta}_N\right)=\sum\limits_{k\in\mathcal{K}} \rho^{B}_{n,u,k}B^N\log \left( {1 + \gamma_{n,u}\left(\boldsymbol{\rho}, \boldsymbol{\delta}, \boldsymbol{\delta}_N\right)} \right) $.
 \begin{equation}\label{eq:2}
\setlength{\abovedisplayskip}{3 pt}
\setlength{\belowdisplayskip}{3 pt}
\begin{split}
c_{n,u,t}\left(\boldsymbol{\rho}, \boldsymbol{\delta}\right)=B_N \log \left( {1 + \rho_{nut}\gamma_{n,u,t}\left(\boldsymbol{\rho}, \boldsymbol{\delta}\right)} \right),
\end{split}
\end{equation} 
where $B_N$ is the bandwidth allocated to terrestrial link $I_{n,u}^B$, which is assumed to be equal for all terrestrial links. %of the utilized Ka band. 

The interference at backhaul link $I^{M}_{n,m}$ that uses time slot $t$ is:
% \begin{equation}\label{eq:interference2}
%% \setlength{\abovedisplayskip}{-2 pt}
%%\setlength{\belowdisplayskip}{-2 pt}
%\begin{split}
%%\Omega_{u}=\sum\limits_{i\in\mathcal{I}\left(b_k\right),i \ne I_{n,u}}O^{i}_{n,u}P_{i} h\left(i\right),
%&\Omega_{n,m,t}\left(\boldsymbol{\rho}, \boldsymbol{\delta}\right)  =\\
%& \left\{ {\begin{array}{*{20}{c}}
%{{\!\!\!\!\sum\limits_{\substack{i\in\mathcal{I}^B/ I^{B}_{n,u}\\n_i\in\mathcal{S}}}\!\!\!\!\!\!\!\rho_{1, n_iu_it}P_{n_i}h\left(n_i, u\right)+\!\!\!\sum\limits_{\substack{i\in\mathcal{I}^M\\}}\rho_{2, n_im_it}P_{m_i}h\left(m_i, u\right)},\!\!\!~\textrm{if~$n\in\mathcal{E}$}}\\
%{{\!\!\!\!\!\!\!\!\!\!\!\!\!\!\!\!\!\!\!\!\!\!\!\!\!\!\!\!\!\!\!\!\!\!\!\!\!\!\!\!\!\!\!\!\!\!\!\!\!\!\!\!\!\!\!\!\!\!\sum\limits_{\substack{i\in\mathcal{I}^B,i\neq I^{B}_{n,u}\\n_i\in\mathcal{E}}}\!\!\!\!\!\!\!\rho_{1, n_iu_it}P_{n_i}h\left(n_i, u\right),\;\;\;~\textrm{otherwise}\;\;}}, 
%\end{array}} \right.
%%\Omega_{n,u,t}\left(\boldsymbol{\rho}, \boldsymbol{\delta}, \boldsymbol{\delta}\right)= 
%%&\!\!\!\!\!\!\sum\limits_{\substack{i\in\mathcal{I}^B,i\neq I^{B}_{n,u}\\}}\!\!\!\!\!\!\!\rho_{1, n_iu_it}P_{n_i}G\left(\theta^{i}_{n,u}\right)h\left(n_i, u\right)\\
%%&+\sum\limits_{\substack{i\in\mathcal{I}^M\\}}\rho_{2, n_im_it}P_{m_i}G\left(\theta^{i}_{n,u}\right)h\left(m_i, u\right),%G^r_t\left(\theta^{n,u}_{i}\right),
%\end{split}
%\end{equation} 
 \begin{equation}\label{eq:interference2}\small
\setlength{\abovedisplayskip}{3 pt}
\setlength{\belowdisplayskip}{-3 pt}
\begin{split}
\Omega_{n,m,t}\left(\boldsymbol{\rho}, \boldsymbol{\delta}\right)=&\sum\limits_{\substack{i\in\mathcal{I}^B}}\rho_{n_iu_it}P_{n_i}h\left(n_i, n\right)+\!\!\!\!\!\!\sum\limits_{\substack{i\in\mathcal{I}^M, i\neq I^{M}_{n,m}\\}}\!\!\!\!\!\!\delta_{n_im_it}P_{m_i}h\left(m_i, n\right),
%&\Omega_{u}=\sum\limits_{i\in\mathcal{I}\left(x\left(n,u\right)\right),i \ne I_{n,u}}P_nG^t_t\left(\theta^{i}_{n,u}\right)h\left(n,u\right)G^r_t\left(\theta^{n,u}_{i}\right),
\end{split}
\end{equation} 
%where $O^{i}_{n,m}$  represents the compensation factor for interference between link $i$ and link $I^{M}_{n,m}$.
Thus, the SINR at this link $I^{M}_{n,m}$ is given as: %$\gamma_{n,m,t,f}\left(\boldsymbol{\rho}, \boldsymbol{\delta}_f\right)=\frac{P_{m}h\left(n,m\right)G\left(\theta^{n,m}_{n,m}\right)}{\Omega_{n,m,t,f}\left(\boldsymbol{\rho}, \boldsymbol{\delta}_f\right)+\sigma^2}$, 
 \begin{equation}\label{eq:2}
\setlength{\abovedisplayskip}{2 pt}
\setlength{\belowdisplayskip}{3 pt}
\begin{split}
&\gamma_{n,m,t}\left(\boldsymbol{\rho}, \boldsymbol{\delta}\right)=\frac{P_{m}h\left(m,n\right)}{\Omega_{n,m,t}\left(\boldsymbol{\rho}, \boldsymbol{\delta}\right)+\sigma^2}.
\end{split}
\end{equation} 
%where $P_m$ represents the transmit power of MBS $m$, which is assumed to be equal for all MBSs. 
The data rate at link $I^{M}_{n,m}$ is $c_{n,m,t}\left(\boldsymbol{\rho}, \boldsymbol{\delta}\right)=B_N\log\left(1+\delta_{nmt}\gamma_{n,m,t}\left(\boldsymbol{\rho}, \boldsymbol{\delta}\right)\right)$.% $c_{n,m}\left(\boldsymbol{\rho}, \boldsymbol{\delta}_N\right)=\sum\limits_{k\in\mathcal{K}}\rho^{M}_{n,m,k}B^N\log\left(1+\gamma_{n,m}\left(\boldsymbol{\rho}, \boldsymbol{\delta}_N\right)\right)$.
% \begin{equation}\label{eq:2}
%%\setlength{\abovedisplayskip}{2 pt}
%%\setlength{\belowdisplayskip}{3 pt}
%\begin{split}
%c_{n,m,t}\left(\boldsymbol{\rho}, \boldsymbol{\delta}\right)=\delta_{nmt}B_N\log\left(1+\gamma_{n,m,t}\left(\boldsymbol{\rho}, \boldsymbol{\delta}\right)\right).
%\end{split}
%\end{equation} 
%The average data rate at link $I^{M}_{n,m}$ is given as:
%
%\begin{equation}\label{eq:2}
%%\setlength{\abovedisplayskip}{3 pt}
%%\setlength{\belowdisplayskip}{3 pt}
%\begin{split}
%\overline c_{n,m}\left(\boldsymbol{\rho}, \boldsymbol{\delta}\right)=\frac{1}{B_{n,m}\left(\boldsymbol{\rho}, \boldsymbol{\delta}\right)}\sum\limits_{t\in\mathcal{T}}B_N\log\left(1+\gamma_{n,m,t}\left(\boldsymbol{\rho}, \boldsymbol{\delta}\right)\right),
%\end{split}
%\end{equation} 
%where $B_{n,m}\left(\boldsymbol{\rho}, \boldsymbol{\delta}\right)=\sum\limits_{t\in\mathcal{T}}\rho_{nmt}$ represents the number of time slots that MBS $m$ allocates to BS $n$.
%As such the data rate at user $u$ is given as $c_u=B\gamma_{u}$.

\vspace{-0.35cm}
\subsection{Satellite Communication Link Data Rate Analysis}
In the Ka band, a LEO satellite provides a supplemental backhaul connection to the BSs.  %We assume that the LEO satellite works on a time division mode over the Ka band and serves the BSs with $R$ equally sized consecutive time slots with time duration $t^S$. The set of $R$ time slots available for satellite backhaul links during the studied time duration is $\mathcal{R}$, the the set of $R_F$ time slots in each satellite backhaul frame is noted as $\mathcal{R}_F$.
The link between the satellite and BS $n$ is represented by $I^{S}_{n}$. 
%Let $\theta^S_{n,l}\in\left\{0, 1\right\}$ be the availability vector of BS $n$ at time slot $l\in\mathcal{L}$, with $\theta^S_{n,l}=1$ representing that BS $n$ is working at satellite time slot $l$.
Let $\boldsymbol{u}={\left[{u}_{11},\cdots, {u}_{NT}\right]}$ be the time allocation vector for the satellite backhaul links, where ${u}_{nt}\in\left\{0, 1\right\}$ represents the time allocation state at link $I^{S}_{n}$ at slot $t$, with ${u}_{nt}=1$ indicating that time slot $t\in\mathcal{T}$ is allocated to satellite backhaul link $I^{S}_{n}$, otherwise, we have $u_{nt}=0$.
The satellite backhaul association vector is given by $\boldsymbol{v}=\left[{v}_{1},\cdots,{v}_{N}\right]$, with ${v}_{n}\in\left\{0, 1\right\}$ indicating the association state between the satellite and BS $n$, where ${v}_{n}=1$ indicates that BS $n$ is associated with the satellite.  
Here, $\boldsymbol{\beta}=\left[\beta_{11},\cdots, \beta_{NT}\right]$ is defined as a satellite backhaul allocation vector, with $\beta_{nt}={u}_{nt}{v}_{n}$ being a single allocation scheme of link $I^{S}_{n}$, and  $\beta_{nt}=1$ indicating that BS $n$ is served by the satellite at time slot $t$. %In summary, the allocation scheme in satellite links is noted in a matrix $\boldsymbol{\beta}=\left[\boldsymbol{\beta}_1,\cdots,\boldsymbol{\beta}_F\right]$.

The interference over the backhaul link between BS $n$ and the satellite at slot $t\in\mathcal{T}$ is: %  working on time slot $l$, at time instant $\tau$, can be expressed as: 
% \begin{equation}\label{eq:interference3}
%  %\setlength{\abovedisplayskip}{3 pt}
%%\setlength{\belowdisplayskip}{3 pt}
%\begin{split}
%\Omega_{n,S,t}\left(\boldsymbol{\rho}, \boldsymbol{\delta}\right)=&\!\!\!\sum_{\substack{i\in\mathcal{I}^B, i\in\mathcal{E}}}\rho_{n_i u_i t}P_{m_i}h\left(m_i,n\right).
%%\Omega_{s,n,t}\left(\boldsymbol{\rho}, \boldsymbol{\delta}_s\right)=\!\!\!\sum_{\substack{i\in\mathcal{I}}}\sum\limits_{t\in\mathcal{T}_F}\rho_{l_i, u_i, t}O^{t}_{r}P_{l_i}h\left(l_i,l\right)G\left(\theta^{s,l}_{i}\right).
%%\Omega_{n}=&\sum\limits_{m\in\mathcal{M}}\sum_{\substack{i\in\mathcal{I}\left(b_n\right) \\ i \ne{I}_{n,m} \cup {I}_{s,n}}} {y}\left(n,m\right)O^{i}_{n,m}P_{i} h\left(i\right)\\
%%&+\sum_{\substack{i\in\mathcal{I}\left(b_n\right) \\ i \ne{I}_{n,m} \cup {I}_{s,n}}} {y}\left(s,n\right)O^{i}_{s,n}P_{i} h\left(i\right),
%\end{split}
%\end{equation}

 \begin{equation}\label{eq:interference3}
  \setlength{\abovedisplayskip}{-16 pt}
\setlength{\belowdisplayskip}{-1 pt}
\begin{split}
\Omega_{n,S,t}\left(\boldsymbol{\rho}, \boldsymbol{\delta}\right)=&\!\!\!\sum_{\substack{i\in\mathcal{I}^B}}\rho_{n_i u_i t}P_{n_i}h\left(n_i,n\right)+\sum_{\substack{i\in\mathcal{I}^M}}{\delta}_{ n_i m_i t}P_{m_i}h\left(m_i,n\right).
%\Omega_{s,n,t}\left(\boldsymbol{\rho}_s\right)=\!\!\!\sum_{\substack{i\in\mathcal{I}}}\sum\limits_{t\in\mathcal{T}_F}\rho_{l_i, u_i, t}O^{t}_{r}P_{l_i}h\left(l_i,l\right)G\left(\theta^{s,l}_{i}\right).
%\Omega_{n}=&\sum\limits_{m\in\mathcal{M}}\sum_{\substack{i\in\mathcal{I}\left(b_n\right) \\ i \ne{I}_{n,m} \cup {I}_{s,n}}} {y}\left(n,m\right)O^{i}_{n,m}P_{i} h\left(i\right)\\
%&+\sum_{\substack{i\in\mathcal{I}\left(b_n\right) \\ i \ne{I}_{n,m} \cup {I}_{s,n}}} {y}\left(s,n\right)O^{i}_{s,n}P_{i} h\left(i\right),
\end{split}
\end{equation}
%where $O^{i}_{s,n}={b^i_{s,n}}/{B_S}$ is the compensation factor defined as the portion of the interfering frequency band adopted by link $i$ that overlaps with the victim satellite frequency band adopted by link $I^{S}_{n}$.  $b^i_{s,n}$ is the width of interfering band that overlaps with the band of victim satellite backhaul link $I^S_{n}$, $B_S$ is the bandwidth allocated to each satellite backhaul links.
%where $\mathcal{K}_l$ is the set of terrestrial sub-channels that fall within, or partially overlap with, the bandwidth of victim satellite sub-channel $l$.
%where $\mathcal{I}\left(l\right)$ represent the set of terrestrial links that work on the (or portion of) frequency band of sub-channel $l$. $O^{i}_{s,n}$ is the compensation factor for interference from terrestrial link $i$ to BS $n$'s satellite backhaul link $I^{S}_{s,n}$, respectively. 
The SINR of the backhaul links between BS $n$ and a satellite at time interval $t\in\mathcal{T}$ will be: %$\gamma_{s,n,t,f}\left( \boldsymbol{\rho}_f\right)=\frac{P_S h\left(s,n\right)G\left(\theta^{s,n}_{s,n}\right)}{\Omega_{s,n,t,f}\left(\boldsymbol{\rho}_f\right)+\Omega_c+\sigma^2}$, 
 \begin{equation}\label{eq:2}
\setlength{\abovedisplayskip}{2 pt}
\setlength{\belowdisplayskip}{2 pt}
\begin{split}
\gamma_{n,S,t}\left( \boldsymbol{\rho}, \boldsymbol{\delta}\right)=\frac{P_S h\left(s,n\right)}{\Omega_{n,S,t}\left(\boldsymbol{\rho}, \boldsymbol{\delta}\right)+\Omega_c+\sigma^2},\\
%&\gamma_{s,n,l}\left(\boldsymbol{\rho}_N\right)=\frac{P_S h\left(s,n\right)G\left(0\right)}{N^{S}_{n,l}\left(\boldsymbol{\rho}_N\right)+\Omega_c+\sigma^2},
\end{split}
\end{equation} 
where $P_S$ is the transmit power of the LEO satellite, which is assumed to be equal for all satellites in constellation.  \iffalse$G\left(\theta^{s,n}_{s,n}\right)$ is the transmit gain of satellite antenna at offset angle $\theta^{s,n}_{s,n}$. $\Omega_c$ is the co-channel interference due to the use of a multi-beam satellite. $h\left(s,n\right)=\left(\frac{s}{4\pi d_{s,n} b_{k\left(I^{S}_{s,n}\right)}}\right)^2$ represent the propagation loss at satellite backhaul link $I^{S}_{s,n}$. \fi Thus, the data rate of link $I^{S}_{n}$ is: %$c_{s,n}\left(\boldsymbol{\rho}_N, \boldsymbol{\rho}_S\right)=\sum\limits_{l\in\mathcal{L}}\rho^{S}_{n,l}B^S\log\left(1+\gamma_{s,n}\left(\boldsymbol{\rho}_N\right)\right)$.
 \begin{equation}\label{eq:cnst}
\setlength{\abovedisplayskip}{-5 pt}
\setlength{\belowdisplayskip}{-2 pt}
\begin{split}
c_{n,S,t}\left(\boldsymbol{\rho}, \boldsymbol{\delta},\boldsymbol{\beta} \right)=B_S\log\left(1+\beta_{nt}\gamma_{s,n,t}\left(\boldsymbol{\rho}, \boldsymbol{\delta}\right)\right),
\end{split}
\end{equation} 
where $B_S$ is the bandwidth allocated to link $I_{n}^S$, which is assumed to be equal for all satellite backhaul links. 
  %for interference from terrestrial link $i$. 
% \begin{equation}\label{eq:2}
%%\setlength{\abovedisplayskip}{2 pt}
%%\setlength{\belowdisplayskip}{3 pt}
%\begin{split}
%&c_{n}=\sum\limits_{m\in\mathcal{M}}y\left(n,m\right)B^t\gamma_{n,m}+B^{s,n}\gamma_{s,n}.
%\end{split}
%\end{equation} 
%Thus, the average data rate at satellite link $I^{S}_{n}$ in the studied duration is:
% \begin{equation}\label{eq:2}
%\setlength{\abovedisplayskip}{2 pt}
%%\setlength{\belowdisplayskip}{3 pt}
%\begin{split}
%\overline c_{n,S}\left(\boldsymbol{\rho}, \boldsymbol{\delta},\boldsymbol{\beta} \right)=\frac{1}{B_{n,S}\left(\boldsymbol{\beta}\right)}\sum\limits_{t\in\mathcal{T}}B_S\log\left(1+\gamma_{n,S,t}\left(\boldsymbol{\rho}, \boldsymbol{\delta},\boldsymbol{\beta}\right)\right).
%\end{split}
%\end{equation} 
%where $B_{n,S}\left(\boldsymbol{\beta}\right)=\sum\limits_{t\in\mathcal{T}}\beta_{nt}$ is the number of time slots the satellite allocates to BS $n$.

Here, we assume that within the studied resource stringent ISDN, each BS or MBS will continuously serve users or BSs during the studied duration $\Delta$. In such a case, the interference in (\ref{eq:interference1}) terrestrial link $I^B_{n,u}$ can be rewritten as:
 \begin{equation}\label{eq:reinterference1}
 \setlength{\abovedisplayskip}{-2 pt}
\setlength{\belowdisplayskip}{-2 pt}
\begin{split}
%\Omega_{u}=\sum\limits_{i\in\mathcal{I}\left(b_k\right),i \ne I_{n,u}}O^{i}_{n,u}P_{i} h\left(i\right),
\Omega_{n,u,t}=\sum\limits_{\substack{i\in\mathcal{N}/n\\}}P_{i}h\left(i, u\right)+\sum\limits_{\substack{j\in\mathcal{M}\\}}P_{j}h\left(j, u\right),
%& \left\{ {\begin{array}{*{20}{c}}
%{{\sum\limits_{\substack{i\in\mathcal{S}/n\\}}P_{i}h\left(i, u\right)+\sum\limits_{\substack{j\in\mathcal{M}\\}}P_{j}h\left(j, u\right)},~\textrm{if~$n\in\mathcal{S}$}}\\
%{{\sum\limits_{\substack{i\in\mathcal{E}/n}}P_{i}h\left(i, u\right),\;\;\;\;\;\;\;\;\;\;\;\;\;\;\;\;\;\;\;\;\;\;\;\;\;\;\;\;~\textrm{if~$n\in\mathcal{E}$}\;\;}}
%\end{array}} \right.,
%&\sum\limits_{\substack{i\in\mathcal{S}/n\\}}P_{i}h\left(i, u\right)+\sum\limits_{\substack{j\in\mathcal{M}\\}}P_{j}h\left(j, u\right),%G^r_t\left(\theta^{n,u}_{i}\right),
\end{split}
\end{equation} 
Similarly, the interference in (\ref{eq:interference2}) at backhaul link $I^M_{nu}$ is rewritten as:

\begin{equation}\label{eq:reinterference2}\small
 \setlength{\abovedisplayskip}{-12 pt}
\setlength{\belowdisplayskip}{-2 pt}
\begin{split}
\Omega_{n,m,t}=&\!\!\sum\limits_{\substack{i\in\mathcal{N}\\i\neq n}}P_{i}h\left(i, n\right)+\!\!\!\!\sum\limits_{\substack{j\in\mathcal{M}/m\\}}\!\!\!\!P_{j}h\left(j, n\right).
%&\Omega_{u}=\sum\limits_{i\in\mathcal{I}\left(x\left(n,u\right)\right),i \ne I_{n,u}}P_nG^t_t\left(\theta^{i}_{n,u}\right)h\left(n,u\right)G^r_t\left(\theta^{n,u}_{i}\right),
\end{split}
\end{equation} 
Moreover, the interference in (\ref{eq:interference3}) at backhaul link $I^S_{n}$ is rewritten as:
 \begin{equation}\label{eq:reinterference3}
  \setlength{\abovedisplayskip}{3 pt}
\setlength{\belowdisplayskip}{3 pt}
\begin{split}
\Omega_{n,S,t}=&\sum\limits_{\substack{i\in\mathcal{N}\\i\neq n}}P_{i}h\left(i, n\right)+\!\!\!\!\sum\limits_{\substack{j\in\mathcal{M}\\}} P_{j}h\left(j, n\right).
%\Omega_{s,n,t}\left(\boldsymbol{\rho}_s\right)=\!\!\!\sum_{\substack{i\in\mathcal{I}}}\sum\limits_{t\in\mathcal{T}_F}\rho_{l_i, u_i, t}O^{t}_{r}P_{l_i}h\left(l_i,l\right)G\left(\theta^{s,l}_{i}\right).
%\Omega_{n}=&\sum\limits_{m\in\mathcal{M}}\sum_{\substack{i\in\mathcal{I}\left(b_n\right) \\ i \ne{I}_{n,m} \cup {I}_{s,n}}} {y}\left(n,m\right)O^{i}_{n,m}P_{i} h\left(i\right)\\
%&+\sum_{\substack{i\in\mathcal{I}\left(b_n\right) \\ i \ne{I}_{n,m} \cup {I}_{s,n}}} {y}\left(s,n\right)O^{i}_{s,n}P_{i} h\left(i\right),
\end{split}
\end{equation}
Thus, the data rates at a terrestrial link $I^B_{nu}$ and backhaul links $I^M_{nu}$ and $I^S_{n}$ will only be a function of the allocation scheme $\rho_{nut}$, $\delta_{nmt}$, and $\beta_{nt}$, respectively. 

We also consider QoS requirements at each user in the form of $\overline{c}_{u}\left(\boldsymbol{\rho}\right)=\frac{1}{T}\sum\limits_{t\in\mathcal{T}}\sum\limits_{n\in\mathcal{N}}c_{n,u,t}\left(\boldsymbol{\rho}\right)\ge\Pi_u$, with $\overline{c}_{u}\left(\boldsymbol{\rho}\right)$ being user $u$'s average data rate over the studied duration, $\Pi_u$ being the minimum rate threshold required by user $u$. In other words, the average data rate that user $u$ can achieve over the studied duration should never be less than $\Pi_u$. Meanwhile, each BS $n$'s average backhaul data rate should not be less than threshold $\Pi_n$, that is, $\overline c_{n}\left(\boldsymbol{\delta}, \boldsymbol{\beta}\right)=\frac{1}{T}\sum\limits_{t\in\mathcal{T}}\left(\sum\limits_{m\in\mathcal{M}}c_{n,m,t}\left(\boldsymbol{\delta}\right)+c_{n,S,t}\left(\boldsymbol{\beta}\right)\right)\ge\Pi_n$.

\vspace{-0.16cm}

\section{Problem Formulation and Proposed Resource Allocation Framework}
\subsection{Problem Formulation}
Given the system model of Section II, our goal is to find an effective resource allocation and user association scheme that maximizes the system's data rates, while satisfying the user data service request, under a limited available time resource. However, we note that the resource allocation and association scheme at one user or BS can affects the entire network performance. In particular, the optimal allocation and association scheme of any given user or BS could yield significant performance degradation at other communication links by introducing interference to these victim links. This directly leads to a decrease in the amount of data provided by these links which, consequently, cannot meet the data service demand from the users. Moreover, while optimizing its own data rate, a user could continuously occupy a BS without considering its time budget. 
%could also outspent the limited time resource at BSs, MBSs, and the satellite.
%the optimal allocation and association scheme at one user or BS could lead to excessive consumption of the time resource, as the resource budget at the BSs, MBS and satellite being neglected.
%With its optimal allocation and association scheme, an user or BS could also take a lot time resource from its associated BS, MBS or satellite to satisfy its own data service demand, without considering the time resource budget at this BS, MBS or satellite. %such that other communication links will not have enough resource to serve their users.
In order to solve the joint access and resource allocation problem among the various communication links in the resource-limited ISDN, we formulate a \emph{competitive market} \cite{arrow1974general} which can take into account the QoS needs of every user, BS, MBS, and, also, the satellite. In such a market, the BSs, MBSs and the satellite submit quotations (or prices) for each slot of service. %from their service provider, (i.e. BSs for the users, MBSs and the satellite for the BSs), which will, on the other hand, 
Meanwhile, the users and BSs find the resource allocation and association scheme that is most beneficial for them, that is, the scheme with highest data rate but lowest charged price. By using a competitive market solution, the intractable many-to-many matching problem, that is, the resource-limited joint access and allocation problem in the ISDN, can be effectively solved. The problem is solved in a distributed way to prevent excessive overhead and alleviate privacy concerns.

%This is quite different from existing approaches like dynamic game, or optimization, in which each communication link chooses incentive mechanism that maximizes its utility, all the devices in the market can make decisions, based on their own utilities. Precisely speaking, the users and BSs determine their quotations based on their preference on the service provided by different BS, MBS, and the satellite, and choose the scheme with the highest communication quality and lowest price. The BSs, MBSs and satellite, then, allocate their time resource to each user and BS, based on the quotation proposed by each user and BS. In such a case, the users/BSs' preference is fully considered in the ISDN.

%that maximizes its own utility without considering the resource budget at each service provider (i.e. the BSs at radio access links, MBSs and the satellite at backhaul links), data service demand, and communication performance of the users.

%resource allocation and association scheme is adopt based on the supply and demand of the data service, as none of the user, BS, MBS, or satellite in the market can dominate the optimization process and put its own payoff at a primary place. This is quite different from existing approaches, in which each communication link chooses incentive mechanism
%that maximizes its own utility without considering the resource budget at each service provider (i.e. the BSs at radio access links, MBSs and the satellite at backhaul links), data service demand, and communication performance of the users.

In this model, the primary objective of each user is to find a resource allocation and association scheme that satisfies its data service demand and maximize its data rate, while each BS seeks to find an resource allocation scheme that can provide sufficient data service under the limited time budget. Meanwhile,  each BS also seeks to find a resource allocation and association scheme that satisfies its backhaul service demand with maximized data rate. Each MBS, or the satellite, seeks to find a resource allocation scheme that can meet the backhaul service demand within their resource budget.
To capture the decision making processes of the users, BSs, MBSs, and the satellite,  we introduce a competitive market defined by the tuple $\left[\mathcal{P}, \boldsymbol{\rho}, \boldsymbol{\delta},\boldsymbol{\beta}, \boldsymbol{\theta}, \boldsymbol{\varphi}, \boldsymbol{\varepsilon}, \mathcal{W},\mathcal{V}\right]$ to capture the dependence between the supply and the demand of data service. Here,
\begin{itemize}
\item $\mathcal{P}$ is the set of users, BSs, MBSs and the satellite. 

\item $\boldsymbol{\rho}$, $ \boldsymbol{\delta}$, and $\boldsymbol{\beta}$ are the vectors of allocation schemes deployed at BSs, MBSs, and the satellite, respectively.

\item $\boldsymbol{\theta}=\left[\theta_{111},\cdots,\theta_{NUT}\right]$ is the request schemes deployed at users,  whereby $\theta_{nut}=1$ implies that user $u$ requests data service from BS $n$ at slot $t$, otherwise, if $\theta_{nut}=0$, then user $u$ does not request service from BS $n$ at slot $t$.
$\boldsymbol{\varphi}=\left[\varphi_{111},\cdots,\varphi_{NMT}\right]$, and $\boldsymbol{\varepsilon}=\left[\varepsilon_{11},\cdots,\varepsilon_{NT}\right]$ are request vectors at each BS,  such that $\varphi_{nmt}=1$ and $\varepsilon_{nt}=1$ represent that, at slot $t$, BS $n$ requests backhaul service from MBS $m$, and the satellite, respectively.

\item $\mathcal{W}$ is the set of payoffs of the \emph{buyers}, which are the users in the radio access links and BSs in the backhaul links.

\item $\mathcal{V}$ is the set of payoffs of the \emph{sellers}, which are the BSs in the radio access links, MBSs and the satellite in the backhaul links.
\end{itemize}

%Here, $\mathcal{P}$ is the set of users, BSs, MBSs and the satellite in the ISDN system. $\boldsymbol{\rho}$, $ \boldsymbol{\delta}$, and $\boldsymbol{\beta}$ are the vectors of allocation schemes deployed at each BSs, MBSs and the satellite, respectively. $\boldsymbol{\theta}=\left[\theta_{111},\cdots,\theta_{NUT}\right]$ is the request schemes deployed at users,  whereby $\theta_{nut}=1$ implies that user $u$ requests data service from BS $n$ at slot $t$, otherwise, if $\theta_{nut}=0$, then user $u$ does not request service from BS $n$ at slot $t$.
%$\boldsymbol{\varphi}=\left[\varphi_{111},\cdots,\varphi_{NMT}\right]$, and $\boldsymbol{\varepsilon}=\left[\varepsilon_{11},\cdots,\varepsilon_{NT}\right]$ are request vectors at each BSs,  such that $\varphi_{nmt}=1$ and $\varepsilon_{nt}=1$ represents that, at slot $t$, BS $n$ requests backhaul service from MBS $n$, and the satellite, respectively.  $\mathcal{W}$ is the set of payoffs of \emph{buyers}, which are the users in the radio access links and BSs in the backhaul links. $\mathcal{V}$ is the set of payoffs of \emph{seller}, which are the BSs in the radio access links, MBSs and the satellite in the backhaul links. 
Note that, in our competitive market,  user $u$ will pay $\sum\limits_{t\in\mathcal{T}}q_{u,n,t}{\rho}_{nut}$ to seller BS $n$, BS $n$ will pay $\sum\limits_{t\in\mathcal{T}}q_{n,m,t}{\delta}_{nmt}$ to MBS $m$ and pay $\sum\limits_{t\in\mathcal{T}}q_{n,S,t}\beta_{nt}$ to the satellite, as a return for the provided service. Here, $q_{i,j,t}$ is the cost of the time slot $t$ seller $j$ received from buyer $i$, noted as unit price. Note that, for each buyer, the unit prices of time slots provided by different sellers will vary, as the communication quality provided by different sellers varies. In particular, the buyer will have to pay a higher price to obtain a higher quality service as captured by the data rates in the ISDN system.

%one specific task has different prices at different service executor due to regional differences of the tasks' relative importance at different service executors. For example, providing a high data rate to a voice chat user is relatively less important when the BS is obligated to serve a bunch of multimedia virtual reality (VR) users. %As the unit price is defined to implicate the relative importance of each task at each seller, we have $\sum\limits_{u\in\mathcal{U}}{q_{u,n,t}}=1$, $\sum\limits_{n\in\mathcal{N}}{q_{n,m,t}}=1$ and $\sum\limits_{n\in\mathcal{N}}{q_{n,S,t}}=1$. 

Each buyer will use a utility function to measure its profit in the market, which is defined as the normalized difference between its communication quality (i.e. its data rate) and the price it pays to the sellers. Thus, the payoff function of user $u$ can will be: 
 \begin{equation}\label{eq:pu}
\setlength{\abovedisplayskip}{2 pt}
\setlength{\belowdisplayskip}{3 pt}
\begin{split}
%W_{u}\left(\boldsymbol{\theta}, \boldsymbol{\varphi} \right)=&\frac{1}{NT\Pi_u}\sum\limits_{n\in\mathcal{N}}\sum\limits_{t\in\mathcal{T}} {c}_{n,u,t}\left(\boldsymbol{\theta}, \boldsymbol{\varphi}\right)-\frac{1}{T\Pi_u}\sum\limits_{n\in\mathcal{N}}\sum\limits_{t\in\mathcal{T}}q_{u,n,t}c_{n,u,t}\left(\boldsymbol{\theta}, \boldsymbol{\varphi}\right),
%W_{u}\left(\boldsymbol{\theta}, \boldsymbol{\varphi} \right)=&\frac{1}{I_u\left(\boldsymbol{\theta}, \boldsymbol{\varphi}\right)\Pi_u}\sum\limits_{n\in\mathcal{N}}\overline {c}_{n,u}\left(\boldsymbol{\theta}, \boldsymbol{\varphi}\right)-\frac{1}{A_{n,u}\left(\boldsymbol{\theta}\right)\Pi_u}\sum\limits_{n\in\mathcal{N}}\sum\limits_{t\in\mathcal{T}}q_{u,n,t}c_{n,u,t}\left(\boldsymbol{\theta}, \boldsymbol{\varphi}\right),
%W_{u}\left(\boldsymbol{\theta}, \boldsymbol{\varphi}  \right)=&\frac{1}{I_u\left(\boldsymbol{\theta} \right)\Pi_u}\sum\limits_{n\in\mathcal{N}}\sum\limits_{t\in\mathcal{T}} {c}_{n,u,t}\left(\boldsymbol{\theta}\right)-\frac{1}{T}\sum\limits_{n\in\mathcal{N}}\sum\limits_{t\in\mathcal{T}}q_{n,u,t}{\rho}'_{1,nut},
W_{u}\left(\boldsymbol{\theta}  \right)=&\frac{1}{T\Pi_u}\sum\limits_{n\in\mathcal{N}}\sum\limits_{t\in\mathcal{T}} {c}_{n,u,t}\left(\boldsymbol{\theta} \right)-\frac{1}{T}\sum\limits_{n\in\mathcal{N}}\sum\limits_{t\in\mathcal{T}}q_{u,n,t}{\theta}_{nut}.
\end{split}
\end{equation} 
%where $I_u\left(\boldsymbol{\theta}, \boldsymbol{\varphi}\right)=\sum\limits_{n\in\mathcal{N}}\mathbbm{1}_{\overline {c}_{n,u}\left(\boldsymbol{\theta}, \boldsymbol{\varphi}\right)\neq0}$ represent the number of BSs that have provided service to user $u$ during the whole process. $A_{n,u}\left(\boldsymbol{\theta}\right)=\sum\limits_{t\in\mathcal{T}}{\theta}_{nut}$ is the amount of time slots user $u$ requests from BS $n$ with scheme $\boldsymbol{\theta}$. 
%where $\Pi_u$ is the QoS requirement at user $u$, or the the minimum rate thresholds required by user $u$. 

%The QoS requirements in the ISDN system is $c_{u,t}\left(\boldsymbol{\rho}, \boldsymbol{\delta}\right)=\sum\limits_{n\in\mathcal{N}}c_{n,u,t}\left(\boldsymbol{\rho}, \boldsymbol{\delta}\right)\ge\Pi_u$ and $c_{n,t}\left(\boldsymbol{\rho}, \boldsymbol{\delta}, \boldsymbol{\beta}\right)=c_{n,S,t}\left(\boldsymbol{\rho}, \boldsymbol{\delta}, \boldsymbol{\beta}\right)+\sum\limits_{m\in\mathcal{M}}c_{n,m,t}\left(\boldsymbol{\rho}, \boldsymbol{\delta}\right)\ge\Pi_n$,
%where $\Pi_u$ and $\Pi_n$ are the minimum rate thresholds required by user $u$ and BS $n$, respectively. In other words, any allocation scheme that is applied in the ISDN system should guarantee data rates of $\Pi_u$ and $\Pi_n$ at users and BSs.
The payoff function of BS $n$, as a buyer, is given by:
 \begin{equation}\label{eq:pn2}
\setlength{\abovedisplayskip}{2 pt}
\setlength{\belowdisplayskip}{3 pt}
\begin{split}
W_{n}\left(\boldsymbol{\varphi},\boldsymbol{\varepsilon} \right)&=\frac{1}{T\Pi_n}\sum\limits_{m\in\mathcal{M}}\sum\limits_{t\in\mathcal{T}} {c}_{n,m,t}\left(\boldsymbol{\varphi}\right)+\frac{1}{T\Pi_n}\sum\limits_{t\in\mathcal{T}} c_{n,S,t}\left(\boldsymbol{\varepsilon} \right)\\
&-\frac{1}{T}\sum\limits_{m\in\mathcal{M}}\sum\limits_{t\in\mathcal{T}}q_{n,m,t}\varphi_{nmt}-\frac{1}{T}\sum\limits_{t\in\mathcal{T}}q_{n,S,t}\varepsilon_{nt}.
\end{split}
\end{equation} 
On the other hand, seller $j$ will receive its payment under the cost $P_j A_j \tau$, which is the energy it consumed during the data transmission process. Here, $P_j$ is the transmit power of executor $j$. $A_j$ is the total number of time slots that seller $j$ has allocated during the whole process. %such that $B_j=\sum\limits_{i\in\mathcal{U}}B_{j,i}\left({\boldsymbol{\rho}}\right)$ when $j$ is a BS, $B_j=\sum\limits_{i\in\mathcal{N}}B_{i,j}\left( \boldsymbol{\delta}\right)$ when $j$ is a MBS, $B_j=\sum\limits_{i\in\mathcal{N}}B_{i,S}\left(\boldsymbol{\beta}\right)$ when $j$ is the satellite.  
%\ The payoff of a sellers is the normalized difference between its received payment from buyers and energy consumption. Thus, the payoff function of user $u$ can be formulated as: 
% \begin{equation}\label{eq:pu}
%%\setlength{\abovedisplayskip}{2 pt}
%%\setlength{\belowdisplayskip}{3 pt}
%\begin{split}
%W_{u}\left(\boldsymbol{\theta}, \boldsymbol{\varphi} \right)=&\frac{1}{I_u\left(\boldsymbol{\theta}, \boldsymbol{\varphi}\right)\Pi_u}\sum\limits_{n\in\mathcal{N}}\overline {c}_{n,u}\left(\boldsymbol{\theta}, \boldsymbol{\varphi}\right)-\frac{1}{A_{n,u}\left(\boldsymbol{\theta}, \boldsymbol{\varphi}\right)\Pi_u}\sum\limits_{n\in\mathcal{N}}q_{n,u}\sum\limits_{t\in\mathcal{T}}c_{n,u,t}\left(\boldsymbol{\theta}, \boldsymbol{\varphi}\right),
%%W_{u}\left(\boldsymbol{\rho}' \right)=&\frac{1}{I_u\left(\boldsymbol{\rho}'\right)\Pi_u}\sum\limits_{n\in\mathcal{N}}\overline {c}_{n,u}\left(\boldsymbol{\rho}'\right)-\frac{1}{T}\sum\limits_{n\in\mathcal{N}}\sum\limits_{t\in\mathcal{T}}q_{n,u,t}{\rho}'_{1,nut},
%\end{split}
%\end{equation} 
%where $I_u\left(\boldsymbol{\theta}, \boldsymbol{\varphi}\right)=\sum\limits_{n\in\mathcal{N}}\mathbbm{1}_{\overline {c}_{n,u}\left(\boldsymbol{\theta}, \boldsymbol{\varphi}\right)\neq0}$ represent the number of BSs that have provided service to user $u$ during the whole process. 
Thus, the payoff of BS $n$ as an seller is given by:
 \begin{equation}\label{eq:pn}
\setlength{\abovedisplayskip}{3 pt}
\setlength{\belowdisplayskip}{5 pt}
\begin{split}
%V_{n}\left(\boldsymbol{\rho}, \boldsymbol{\delta} \right)&=\frac{1}{T\Pi_u}\sum\limits_{u\in\mathcal{U}}\sum\limits_{t\in\mathcal{T}}q_{n,u,t}c_{n,u,t}\left(\boldsymbol{\rho}, \boldsymbol{\delta}\right)-\frac{1}{T}\sum\limits_{u\in\mathcal{U}}\sum\limits_{t\in\mathcal{T}}\rho_{nut}.
%V_{n}\left(\boldsymbol{\rho}, \boldsymbol{\delta} \right)&=\frac{1}{B_{n,u}\left(\boldsymbol{\rho}\right)\Pi_u}\sum\limits_{u\in\mathcal{U}}\sum\limits_{t\in\mathcal{T}}q_{n,u,t}c_{n,u,t}\left(\boldsymbol{\rho}, \boldsymbol{\delta}\right)-\frac{1}{T}\sum\limits_{u\in\mathcal{U}}B_{n,u}\left(\boldsymbol{\rho}\right).
V_{n}\left(\boldsymbol{\rho} \right)&=\frac{1}{T}\sum\limits_{u\in\mathcal{U}}\sum\limits_{t\in\mathcal{T}}q_{u,n,t}{\rho}_{nut}-\frac{1}{T}\sum\limits_{u\in\mathcal{U}}\sum\limits_{t\in\mathcal{T}}\rho_{nut}.%B_{u,n}\left(\boldsymbol{\rho}, \boldsymbol{\delta} \right).
\end{split}
\end{equation} 
Note that, the energy consumption of BS $n$, $P_n\sum\limits_{u\in\mathcal{U}}\sum\limits_{t\in\mathcal{T}}\rho_{nut}\tau$, is normalized by $P_n\tau T$, such that the normalized cost of BS $n$ is $\frac{1}{T}\sum\limits_{u\in\mathcal{U}}\sum\limits_{t\in\mathcal{T}}\rho_{nut}$. 
Similarly, the payoff function of seller $m\in\mathcal{M}$ will be:
 \begin{equation}\label{eq:pm}
\setlength{\abovedisplayskip}{4 pt}
\setlength{\belowdisplayskip}{4 pt}
\begin{split}
%V_{m}&\left(\boldsymbol{\rho}, \boldsymbol{\delta} \right)=\frac{1}{T\Pi_n}\sum\limits_{n\in\mathcal{N}}\sum\limits_{t\in\mathcal{T}}q_{n,m,t}c_{n,m,t}\left(\boldsymbol{\rho}, \boldsymbol{\delta} \right)-\frac{1}{T}\sum\limits_{n\in\mathcal{N}}\sum\limits_{t\in\mathcal{T}}\delta_{nmt}.
%V_{m}&\left(\boldsymbol{\rho}, \boldsymbol{\delta} \right)=\frac{1}{B_{n,m}\left(\boldsymbol{\delta}\right)\Pi_n}\sum\limits_{n\in\mathcal{N}}\sum\limits_{t\in\mathcal{T}}q_{n,m,t}c_{n,m,t}\left(\boldsymbol{\rho}, \boldsymbol{\delta} \right)-\frac{1}{T}\sum\limits_{n\in\mathcal{N}}B_{n,m}\left(\boldsymbol{\delta} \right).
V_{m}&\left(\boldsymbol{\delta} \right)=\frac{1}{T}\sum\limits_{n\in\mathcal{N}}\sum\limits_{t\in\mathcal{T}}q_{n,m,t}{\delta}_{nmt}-\frac{1}{T}\sum\limits_{n\in\mathcal{N}}\sum\limits_{t\in\mathcal{T}}\delta_{nmt}.%B_{n,m}\left(\boldsymbol{\rho} \right).
\end{split}
\end{equation} 
The payoff function of the satellite, as a seller, is:
 \begin{equation}\label{eq:ps}
\setlength{\abovedisplayskip}{4 pt}
\setlength{\belowdisplayskip}{4 pt}
\begin{split}
%V_{S}&\left(\boldsymbol{\rho}, \boldsymbol{\delta}, \boldsymbol{\beta} \right)=\frac{1}{T\Pi_n}\sum\limits_{n\in\mathcal{N}}\sum\limits_{t\in\mathcal{T}}q_{n,S,t}c_{n,S,t}\left(\boldsymbol{\rho}, \boldsymbol{\delta},\boldsymbol{\beta}\right)-\frac{1}{T}\sum\limits_{n\in\mathcal{N}}\sum\limits_{t\in\mathcal{T}}\beta_{nt}.
%V_{S}&\left(\boldsymbol{\rho}, \boldsymbol{\delta}, \boldsymbol{\beta} \right)=\frac{1}{B_{n,S}\left(\boldsymbol{\beta}\right)\Pi_n}\sum\limits_{n\in\mathcal{N}}\sum\limits_{t\in\mathcal{T}}q_{n,S,t}c_{n,S,t}\left(\boldsymbol{\rho}, \boldsymbol{\delta},\boldsymbol{\beta}\right)-\frac{1}{T}\sum\limits_{n\in\mathcal{N}}B_{n,S}\left(\boldsymbol{\beta} \right).
V_{S}&\left(\boldsymbol{\beta} \right)=\frac{1}{T}\sum\limits_{n\in\mathcal{N}}\sum\limits_{t\in\mathcal{T}}q_{n,S,t}{\beta}_{nt}-\frac{1}{T}\sum\limits_{n\in\mathcal{N}}\sum\limits_{t\in\mathcal{T}}\beta_{nt}.%B_{n,S}\left(\boldsymbol{\beta} \right).
\end{split}
\end{equation} 

In the formulated market, each one of the buyer or seller, i.e., an user, BS, MBS, or the satellite will seek to maximize its own payoff. 
To solve this market, we make use of the concept of \emph{Walrasian equilibrium}\cite{arrow1974general}. At this equilibrium, the supply of data and backhaul service from the sellers %BSs, MBSs, and the satellite 
exactly matches the demand of data and backhaul service at the buyers, with each participants' payoff being maximized at the same time. The concept of a Walrasian equilibrium is formally defined as:

\vspace{-0.3cm}
\begin{definition}
\emph{Given the price of the goods (i.e. the data service and the backhaul service), a \emph{Walrasian equilibrium} in the considered ISDN competitive market is given by a tuple of vectors $\left[\boldsymbol{\rho}^*, \boldsymbol{\delta}^*, \boldsymbol{\beta}^*, \boldsymbol{\theta}^{*}, \boldsymbol{\varphi}^*, \boldsymbol{\varepsilon}^*\right]$ \iffalse$\boldsymbol{x}^*$,  $\boldsymbol{y}^*$\fi that satisfy the following conditions:}

%\begin{itemize}
\vspace{-0.5cm}
    \emph{\begin{enumerate}
\item {For radio access links, $\boldsymbol{\theta}^{*}$ is the optimal solution that maximizes the users' payoff in (\ref{eq:pu}) and  $\boldsymbol{\rho}^*$ is the optimal solution that maximizes the payoff  (\ref{eq:pn}),  of BSs. }
\item {For backhaul links, $\boldsymbol{\varphi}^*$, $\boldsymbol{\varepsilon}^*$ are the solutions that optimize the BSs' payoff in (\ref{eq:pn2}), while $\boldsymbol{\delta}^*$ and $ \boldsymbol{\beta}^*$ optimize the MBSs' and the satellite's payoffs in (\ref{eq:pm}) and (\ref{eq:ps}), respectively. }%$ \boldsymbol{\beta}^*$ optimizes executors' payoff (21).}
%\item \emph{The time slots provider by each seller exactly match to time slots requested by the corresponding buyer. In other words, $\boldsymbol{\rho}'=\boldsymbol{\rho}$, $\boldsymbol{\beta}'=\boldsymbol{\beta}$.}
\item {The time slots provided by each seller are exactly matched to the number of time slots requested by the buyers. In other words, $\boldsymbol{\rho}^*=\boldsymbol{\theta}^{*}$, $\boldsymbol{\delta}^*=\boldsymbol{\varepsilon}^*$, $\boldsymbol{\beta}^*=\boldsymbol{\varepsilon}^*$.} %$c_{n,u,t}\left(\boldsymbol{\rho}'\right)=c_{n,u,t}\left(\boldsymbol{\rho}\right)$, $c_{n,m,t}\left(\boldsymbol{\rho}'\right)=c_{n,m,t}\left(\boldsymbol{\rho}\right)$ and $c_{n,S,t}\left(\boldsymbol{\beta}'\right)=c_{n,S,t}\left(\boldsymbol{\beta}\right)$, for all $u\in\mathcal{U}$, $n\in\mathcal{N}$, $m\in\mathcal{M}$ and $t\in\mathcal{T}$.}
\end{enumerate}}
%\end{itemize}

\end{definition}
\vspace{-0.2cm}
At the Walrasian equilibrium,  the communication links' optimal data rates are guaranteed and the {clearance of the market} is ensured by a full utilization of the time resource in the system, with which payoffs (\ref{eq:pu})-(\ref{eq:ps}) in the ISDN are optimized\cite{arrow1974general}. %Note that with the market clearance constraint, one can get $\theta_{nut}=\rho_{nut}$, $\varphi_{nmt}=\delta_{nmt}$ and $\varepsilon_{nt}=\beta_{nt}$, for all $u\in\mathcal{U}$, $n\in\mathcal{N}$, $m\in\mathcal{M}$ and $t\in\mathcal{T}$. This stem from a fact that any $c_{n,u,t}\left(\boldsymbol{\rho}, \boldsymbol{\delta}\right)\neq0$ indicates $\rho_{nut}=1$, otherwise $\rho_{nut}=0$. Similar logic applies to $c_{n,m,t}\left(\boldsymbol{\rho}, \boldsymbol{\delta}\right)$ and $c_{n,S,t}\left(\boldsymbol{\rho}, \boldsymbol{\delta},\boldsymbol{\beta}\right)$.
However, improper prices can lead to the mismatch between the time slots provided by each seller and the number of time slots requested by the corresponding buyer. For example, when the prices are low, the users and BSs will be willing to request as much time slots as possible, yet the BSs, MBS, and satellite could rather not to provide any service at all. As such, a set of optimal prices should yield schemes that fully optimize the individual payoffs (\ref{eq:pu})-(\ref{eq:ps}) in the system without violating any resource budget and resource demand constraints. In practical ISDN scenarios, each user, BS, MBS, or the satellite may not have the ability to analytically characterize such optimal prices by themselves.
To this end, a heavy ball based algorithm, which enables a numerical computation of such optimal prices, is proposed next, which as a result, allows computation of Walrasian equilibrium at which both the payoffs at each devices and the total payoff in the market is reached.
%\begin{proposition}\emph{
%Given the Walrasian equilibrium allocation schemes $\left[\boldsymbol{\rho}^*, \boldsymbol{\delta}^*, \boldsymbol{\beta}^*\right]$, the increase on the number of time slots provided by one BS $n$, with the help of satellite, is no less than:
% \begin{equation}\label{eq:pp1}
%\setlength{\abovedisplayskip}{2 pt}
%\setlength{\belowdisplayskip}{3 pt}
%\begin{split}
%\frac{B_S\log\left(1+\gamma_{s,n,t}\left(\boldsymbol{\rho}^*, \boldsymbol{\delta}^*\right)\right)-\max_{m\in\mathcal{M}}c_{n,m,t}\left(\boldsymbol{\delta}^* \right)}{\max_{u\in\mathcal{U},t\in\mathcal{T}}c_{n,u,t}\left(\boldsymbol{\rho}^* \right)}.
%\end{split}
%\end{equation} 
%}
%\end{proposition}
%\begin{proof}
%\emph{There's no right rectangle whose sides measure 3cm, 4cm, and 6cm.}
%\end{proof}
%\begin{lemma}
%\emph{Based in Proposition 1, given the Walrasian equilibrium schemes $\left[\boldsymbol{\rho}^*, \boldsymbol{\delta}^*, \boldsymbol{\beta}^*, \boldsymbol{\theta}^{*}, \boldsymbol{\varphi}^*, \boldsymbol{\varepsilon}^*\right]$, the increase on the number of time slots that can be provided by at the radio access links in a backhaul stringent ISDN, in which $N>MT$, is no less than :}
% \begin{equation}\label{eq:pp1}
%\setlength{\abovedisplayskip}{2 pt}
%\setlength{\belowdisplayskip}{3 pt}
%\begin{split}
%\min_{n\in\mathcal{N}}\frac{\left(N-MT\right)c_{nSt}}{\max_{u\in\mathcal{U}}c_{nut}}.
%\end{split}
%\end{equation} 
%\end{lemma}
%\begin{proof}
%\emph{There's no right rectangle whose sides measure 3cm, 4cm, and 6cm.}
%\end{proof}

\vspace{-0.46cm}

\subsection{Proposed Resource Allocation Algorithm}
%In this section, the scenario where a centralized cloud platform gathers information and manages the interactions between buyers and mobile users is considered.

%In this section, the scenario where a centralized cloud platform gathers information and manages the interactions between the communication links is studied. 

In this section, the problem of joint access and backhaul resource allocatin is solved in the formulated market which considers the payoffs of all the users, BSs, MBSs, and the satellite at the same time. The constraints, including resource budget, service demand, and QoS requirements, is considered in the market. %A Walrasian equilibrium is used to solve the market.
Here, the total payoff gained by all the users, BSs, MBSs and the satellite in the market is given as:
%the social welfare is a concept from microeconomics that measures the payoffs gained by all the users, BSs, MBSs and the satellite in the market, and is given as:

%finding a Walrasian equilibrium under optimal prices
 \begin{equation}\label{eq:obj}
\setlength{\abovedisplayskip}{-16 pt}
\setlength{\belowdisplayskip}{0 pt}
\begin{split}
J\left(\boldsymbol{\rho}, \boldsymbol{\delta},\boldsymbol{\beta}, \boldsymbol{\theta}, \boldsymbol{\varphi},\boldsymbol{\varepsilon}\right)=&\sum\limits_{u\in\mathcal{U}}W_{u}\left(\boldsymbol{\theta} \right)+\sum\limits_{n\in\mathcal{N}}\left(V_{n}\left(\boldsymbol{\rho} \right)+W_{n}\left( \boldsymbol{\varphi},\boldsymbol{\varepsilon} \right)\right)+\sum\limits_{m\in\mathcal{M}}V_{m}\left(\boldsymbol{\delta} \right)+V_{S}\left(\boldsymbol{\beta} \right). 
\end{split}
\end{equation} 
As such, the joint access and backhaul resource allocation problem is formulated as:
%Here, the social welfare is defined as the sum-rates of all communication links. 

%As such, our main goal is to find the allocation and access scheme to achieve optimal sum-rates across all communication links, under the consideration of diverse QoS requirements, service demand, and resource budget in the formulated ISDN competitive market. The clearance of goods (i.e. the data service provided by the BSs, backhaul service provided by MBSs and the satellite) is also considered so as to guarantee a Walrasian equilibrium solution.
%The corresponding social-welfare maximization problem is formulated as:
\addtocounter{equation}{0}
\setlength{\abovedisplayskip}{-5 pt}
\begin{equation}\label{eq:opt}
%\begin{split}
%\max_{\boldsymbol{\rho}, \boldsymbol{\beta}}\sum\limits_{u\in\mathcal{U}}\sum\limits_{n\in\mathcal{N}}\overline {c}_{n,u}\left(\boldsymbol{\rho}\right)+\sum\limits_{n\in\mathcal{N}}\sum\limits_{m\in\mathcal{M}}\overline {c}_{n,m}\left(\boldsymbol{\rho}\right)-\sum\limits_{n\in\mathcal{N}}P_n\sum\limits_{u\in\mathcal{U}}B_{u,n}\left(\boldsymbol{\rho} \right)\tau-\sum\limits_{m\in\mathcal{M}}P_m\sum\limits_{n\in\mathcal{N}}B_{n,m}\left(\boldsymbol{\rho} \right)\tau-P_S\sum\limits_{n\in\mathcal{N}}B_{n,S}\left(\boldsymbol{\beta} \right)\tau,
\max_{\boldsymbol{\rho}, \boldsymbol{\delta}, \boldsymbol{\beta},\boldsymbol{\theta}, \boldsymbol{\varphi},\boldsymbol{\varepsilon}}J\left(\boldsymbol{\rho}, \boldsymbol{\delta},\boldsymbol{\beta},\boldsymbol{\theta}, \boldsymbol{\varphi},\boldsymbol{\varepsilon} \right),
%\end{split}
\end{equation}
\vspace{-0.4cm}
\begin{align}\label{c1}
\setlength{\abovedisplayskip}{-5 pt}
\setlength{\belowdisplayskip}{0 pt}
%\scalebox{1}{$\;\;\;\;\;\;\;\;\;\;\; \nonumber (\ref{eq:cl})$} \\&
&\;\;\;\;\rm{s.\;t.}\scalebox{1}{$\;\;\; \sum\limits_{n\in\mathcal{N}}\sum\limits_{t\in\mathcal{T}}{c}_{n,u,t}\left(\boldsymbol{\theta}\right)\tau\ge C_u, u\in\mathcal{U},$}   \tag{\theequation a}\\
%&\scalebox{1}{$\;\;\;\;\;\;\;\;\;\;\;  \sum\limits_{n\in\mathcal{N}}\sum\limits_{t\in\mathcal{T}/t'_{u}}{c}_{n,u,t}\left(\boldsymbol{\theta}, \boldsymbol{\varphi}\right)\tau< C_u, u\in\mathcal{U}$}   \tag{\theequation b}\\
%&\scalebox{1}{$\;\;\;\;\;\;\;\;\;\;\; c_{u,t}\left(\boldsymbol{\theta}\right)\ge\Pi_u,  u\in \mathcal {U}$} \tag{\theequation b}\\
%&\scalebox{1}{$\;\;\;\;\;\;\;\;\;\;\; c_{n,t}\left(\boldsymbol{\theta}, \boldsymbol{\varphi},\boldsymbol{\varepsilon}\right)\ge\Pi_n, n\in\mathcal{N}$} \tag{\theequation c}\\
%&\scalebox{1}{$\;\;\;\;\;\;\;\;\;\;\;\sum\limits_{\substack{u\in\mathcal{U}}}c_{n,u,t}\left(\boldsymbol{\rho}, \boldsymbol{\delta}\right)\le c_{n,t}\left(\boldsymbol{\theta}, \boldsymbol{\varphi},\boldsymbol{\varepsilon}\right),  n\in\mathcal {N}, $} \tag{\theequation d}\\
&\scalebox{1}{$\;\;\;\;\;\;\;\;\;\;\; \sum\limits_{\substack{u\in\mathcal{U}}}\rho_{nut}+\sum\limits_{\substack{m\in\mathcal{M}}}\varphi_{nmt}+\varepsilon_{nt}\le1, t\in\mathcal{T},$}   \tag{\theequation b}\\
&\scalebox{1}{$\;\;\;\;\;\;\;\;\;\;\;\sum\limits_{\eta=1}^{t}\left(\sum\limits_{u\in\mathcal{U}}{c}_{n,u,\eta}\left(\boldsymbol{\rho}\right)-\sum\limits_{m\in\mathcal{M}}{c}_{n,m,\eta}\left( \boldsymbol{\varphi}\right)-{c}_{n,S,\eta}\left(\boldsymbol{\varepsilon}\right)\right)\le 0, n\in\mathcal{N},$}   \tag{\theequation c}\\
%&\scalebox{1}{$\;\;\;\;\;\;\;\;\;\;\;  \varphi_{nm1} =1,  n\in\mathcal{N}, m\in\mathcal{M},$}   \tag{\theequation f}\\
%&\scalebox{1}{$\;\;\;\;\;\;\;\;\;\;\; \sum\limits_{\substack{u\in\mathcal{U}}} \rho_{nut}\le\sum\limits_{\substack{m\in\mathcal{M}}} \varphi_{nmt},  n\in\mathcal{N}, t\in\mathcal{T},$}   \tag{\theequation g}\\
&\scalebox{1}{$\;\;\;\;\;\;\;\;\;\;\; 0\le\sum\limits_{\substack{u\in\mathcal{U}}} \sum\limits_{\substack{t\in\mathcal{T}}}\rho_{nut}\le T_n, n\in\mathcal{N}, $} \tag{\theequation d}\\
&\scalebox{1}{$\;\;\;\;\;\;\;\;\;\;\; 0\le\sum\limits_{\substack{n\in\mathcal{N}}}\sum\limits_{\substack{t\in\mathcal{T}}}\delta_{nmt}\le T, m\in\mathcal{M}, $} \tag{\theequation e}\\
&\scalebox{1}{$\;\;\;\;\;\;\;\;\;\;\; 0\le\sum\limits_{\substack{n\in\mathcal{N}}}\sum\limits_{\substack{t\in\mathcal{T}}}\beta_{nt}\le T,$}   \tag{\theequation f}\\
&\scalebox{1}{$\;\;\;\;\;\;\;\;\;\;\; 0\le\sum\limits_{\substack{n\in\mathcal{N}}}\theta_{nut}\le 1, u\in\mathcal{U}, t\in\mathcal{T},$} \tag{\theequation g}\\
&\scalebox{1}{$\;\;\;\;\;\;\;\;\;\;\; 0\le\sum\limits_{\substack{m\in\mathcal{M}}}\varphi_{nmt}+\varepsilon_{nt}\le 1, n\in\mathcal{N}, t\in\mathcal{T},$} \tag{\theequation h}\\
%&\scalebox{1}{$\;\;\;\;\;\;\;\;\;\;\; \sum\limits_{\substack{u\in\mathcal{U}}}\rho_{nut}= 1, n\in\mathcal{N}, t\in\mathcal{T},$}   \tag{\theequation m}\\
&\scalebox{1}{$\;\;\;\;\;\;\;\;\;\;\; \sum\limits_{\substack{n\in\mathcal{N}}}\delta_{nmt}\le1, m\in\mathcal{M}, t\in\mathcal{T}, $} \tag{\theequation i}\\
&\scalebox{1}{$\;\;\;\;\;\;\;\;\;\;\; \sum\limits_{\substack{n\in\mathcal{N}}}\beta_{nt}\le1,t\in\mathcal{T},$}   \tag{\theequation j}\\
&\scalebox{1}{$\;\;\;\;\;\;\;\;\;\;\; \rho_{nut}= 0, n\in\mathcal{N}, u\in\mathcal{U}, t\in\mathcal{T}, t >T_n$}   \tag{\theequation k}\\
&\scalebox{1}{$\;\;\;\;\;\;\;\;\;\;\; \varphi_{nmt}= 0, n\in\mathcal{N},m\in\mathcal{M}, t\in\mathcal{T}, t >T_n$} \tag{\theequation l}\\
&\scalebox{1}{$\;\;\;\;\;\;\;\;\;\;\; \varepsilon_{nt}= 0, n\in\mathcal{N}, t\in\mathcal{T}, t >T_n$} \tag{\theequation m}\\
&\scalebox{1}{$\;\;\;\;\;\;\;\;\;\;\; \overline c_{u}\left(\boldsymbol{\theta}\right)\ge\Pi_u,  u\in \mathcal {U}$} \tag{\theequation n}\\
&\scalebox{1}{$\;\;\;\;\;\;\;\;\;\;\; \overline c_{n}\left(\boldsymbol{\varphi},\boldsymbol{\varepsilon}\right)\ge\Pi_n,  n\in \mathcal {N}$} \tag{\theequation o}\\
&\scalebox{1}{$\;\;\;\;\;\;\;\;\;\;\;  {\rho}_{nut},  {\theta}_{nut}\in \left\{0,1\right\},  n\in\mathcal {N}, u\in \mathcal {U}, t\in \mathcal{T},$}   \tag{\theequation p}\\
&\scalebox{1}{$\;\;\;\;\;\;\;\;\;\;\; {\delta}_{nmt}, {\varphi}_{nmt}\in \left\{0,1\right\},   m\in \mathcal {M}, n\in\mathcal {N}, t\in \mathcal{T},$}   \tag{\theequation q}\\
%&\scalebox{1}{$\;\;\;\;\;\;\;\;\;\;\; {\beta}_{nt}\in \left\{0,1\right\},   n\in\mathcal {N}, t\in \mathcal{T}, $} \tag{\theequation p}\\
&\scalebox{1}{$\;\;\;\;\;\;\;\;\;\;\; {\beta}_{nt}, {\varepsilon}_{nt}\in \left\{0,1\right\},   n\in\mathcal {N}, t\in \mathcal{T}, $}   \tag{\theequation r}\\
&\scalebox{1}{$\;\;\;\;\;\;\;\;\;\;\; \boldsymbol{\rho}=\boldsymbol{\theta}, \boldsymbol{\delta}=\boldsymbol{\varepsilon},   \boldsymbol{\beta}=\boldsymbol{\varepsilon}, $}   \tag{\theequation s}
%&\scalebox{1}{${\boldsymbol{\rho}^*=\boldsymbol{\theta}^{*}, \boldsymbol{\delta}^*=\boldsymbol{\varepsilon}^*, \boldsymbol{\beta}^*=\boldsymbol{\varepsilon}^*.$}   \tag{\theequation y}\\
%&\scalebox{1}{${\rho}\left(n, u, k\right)\in \left\{0,1\right\},  n\in\mathcal {N}, u\in \mathcal {U}, k\in \mathcal{K},$}   \tag{\theequation i}  \\
%&\scalebox{1}{${\rho}\left(n,m, k\right)\in \left\{0,1\right\},   m\in \mathcal {M}, n\in\mathcal {N}, k\in \mathcal{K},$}   \tag{\theequation j} \\
%&\scalebox{1}{${\rho}\left(s,n, k\right)\in \left\{0,1\right\},   n\in\mathcal {N}, k\in \mathcal{K},$}   \tag{\theequation k} \\
% &\scalebox{1}{$c\left(i\right)\neq c\left(j\right), i\in\mathcal{B}\left(n\right), j\in\mathcal{E}\left(n\right),$}   \tag{\theequation l}\\
\end{align}
where the maximization problem in ({\ref{eq:opt}}) captures the total payoff maximization objective in the market, with ({\ref{eq:opt}}a) indicating that the data service demand of each user, $C_u$, must be satisfied. %Note that, here, $t'_{u}$ is the last time slot allocated to user $u$, meaning that, for all $n\in\mathcal{N}$ and $t\ge t'_{u}$, $\rho_{nut}=0$.  
Meanwhile, ({\ref{eq:opt}}b) indicates that one time slot cannot be simultaneously used for backhaul links and radio access links at a BS $n$, so as to avoid interference between the backhaul and radio access links. ({\ref{eq:opt}}c) indicates that the data service each BS can provide must not exceed the capacity of the backhaul of this BS.
({\ref{eq:opt}}d)-({\ref{eq:opt}}f) indicate that BSs, MBSs, and the satellite must operate under time limitations. ({\ref{eq:opt}}g) and ({\ref{eq:opt}}h) capture the fact that each user or BS can only be serviced at most once in each studied time slot. 
({\ref{eq:opt}}i) and ({\ref{eq:opt}}j) captures the fact that, each BS or MBS will serve at most one terrestrial user or BS at each time slot within its effective time. ({\ref{eq:opt}}k)-({\ref{eq:opt}}m) capture the fact that each BS can neither receive nor transmit data out of its effective time. ({\ref{eq:opt}}n) and ({\ref{eq:opt}}o) capture the QoS requirements in the ISDN. ({\ref{eq:opt}}p)-({\ref{eq:opt}}r) indicate that ({\ref{eq:opt}}) is a binary problem. ({\ref{eq:opt}}s) captures the market clearance constraints, with which the data service provided by each BS, MBS, and the satellite is consumed by the users and BSs.
Notice that, under the market clearance constraint ({\ref{eq:opt}}s), the payments from all buyers are returned to the sellers. Hence, the objective function can be simplified as:
 \begin{equation}\label{eq:obj}
\setlength{\abovedisplayskip}{5 pt}
\setlength{\belowdisplayskip}{3 pt}
\begin{split}
J\left(\boldsymbol{\rho}, \boldsymbol{\delta},\boldsymbol{\beta}, \boldsymbol{\theta}, \boldsymbol{\varphi},\boldsymbol{\varepsilon} \right)&=\frac{1}{T\Pi_u}\sum\limits_{u\in\mathcal{U}}\sum\limits_{n\in\mathcal{N}}\sum\limits_{t\in\mathcal{T}}{c}_{n,u,t}\left({\theta}_{nut}\right)-\frac{1}{T}\sum\limits_{n\in\mathcal{N}}\sum\limits_{u\in\mathcal{U}}\sum\limits_{t\in\mathcal{T}}\rho_{nut}\\
&+\sum\limits_{n\in\mathcal{N}}\left[\frac{1}{T\Pi_n}\sum\limits_{m\in\mathcal{M}}\sum\limits_{t\in\mathcal{T}} {c}_{n,m,t}\left({\varphi}_{nmt}\right)+\frac{1}{T\Pi_n}\sum\limits_{t\in\mathcal{T}}c_{n,S,t}\left({\varepsilon}_{nt} \right)\right]\\
&-\frac{1}{T}\sum\limits_{m\in\mathcal{M}}\sum\limits_{n\in\mathcal{N}}\sum\limits_{t\in\mathcal{T}} {\delta}_{nmt}-\frac{1}{T}\sum\limits_{n\in\mathcal{N}}\sum\limits_{t\in\mathcal{T}}{\beta}_{nt}.
%J\left(\boldsymbol{\rho}, \boldsymbol{\delta},\boldsymbol{\beta}, \boldsymbol{\theta}, \boldsymbol{\varphi},\boldsymbol{\varepsilon} \right)&=J_1\left(\boldsymbol{\rho}, \boldsymbol{\delta}, \boldsymbol{\theta}, \boldsymbol{\varphi} \right)+J_2\left(\boldsymbol{\rho}, \boldsymbol{\delta},\boldsymbol{\beta} \right),%\\
%&+\sum\limits_{u\in\mathcal{U}}\frac{1}{I_u\left(\boldsymbol{\rho}'\right)\Pi_u}\sum\limits_{n\in\mathcal{N}}\overline {c}_{n,u}\left(\boldsymbol{\rho}'\right))-\sum\limits_{n\in\mathcal{N}}\frac{1}{T}\sum\limits_{u\in\mathcal{U}}B_{u,n}\left(\boldsymbol{\rho} \right)\\
%&\sum\limits_{n\in\mathcal{N}}\left[\frac{1}{I_n\left(\boldsymbol{\rho}'\right)\Pi_n}\sum\limits_{m\in\mathcal{M}}\overline {c}_{n,m}\left(\boldsymbol{\rho}'\right)+\frac{1}{\Pi_n}\overline c_{n,S}\left(\boldsymbol{\rho}',\boldsymbol{\beta}' \right)\right]\\
%&-\sum\limits_{m\in\mathcal{M}}\frac{1}{T}\sum\limits_{n\in\mathcal{N}}B_{n,m}\left(\boldsymbol{\rho} \right)-\frac{1}{T}\sum\limits_{n\in\mathcal{N}}B_{n,S}\left(\boldsymbol{\beta} \right). 
\end{split}
\end{equation} 

 \subsection{Dual decomposition}
The integer non-linear programming (INLP) optimization problem in ({\ref{eq:opt}}) can be solved using a dual decomposition method. This stems from the fact that constraints ({\ref{eq:opt}}a)-({\ref{eq:opt}}r) are local constraints at different users, BSs, MBSs, and the satellite, which means that each of these constraints will only set boundaries to one user, BS, MBS or the satellite.  For example, constraints ({\ref{eq:opt}}a), ({\ref{eq:opt}}g), and ({\ref{eq:opt}}n) are local constraints for the terrestrial users, ({\ref{eq:opt}}b)-({\ref{eq:opt}}d), ({\ref{eq:opt}}h), ({\ref{eq:opt}}k)-({\ref{eq:opt}}m), and ({\ref{eq:opt}}o) are local constraints at the BSs, ({\ref{eq:opt}}e), ({\ref{eq:opt}}i) are local constraints at MBSs, and ({\ref{eq:opt}}f), ({\ref{eq:opt}}j)  are local constraints at the satellite. ({\ref{eq:opt}}p)-({\ref{eq:opt}}r) are local constraints at different user, BS, MBS, and the satellite, indicating that the allocation schemes are binaries.  Note that, hereinafter, constraints ({\ref{eq:opt}}a)-({\ref{eq:opt}}r) are called local constraints. Note that, solving ({\ref{eq:opt}}) with these local constraints introduces problems including increased amount of information exchanged in the ISDN and also the potential privacy and security crisis. In such a case, the Lagrangian function of the objective function in ({\ref{eq:opt}}) is formulated without considering ({\ref{eq:opt}}a)-({\ref{eq:opt}}r) as:
 \begin{equation}\label{eq:L1}
\setlength{\abovedisplayskip}{2 pt}
\setlength{\belowdisplayskip}{3 pt}
\begin{split}
%L_1=&\sum\limits_{u\in\mathcal{U}}\frac{1}{I_u\left(\boldsymbol{\theta}, \boldsymbol{\varphi}\right)\Pi_u}\sum\limits_{n\in\mathcal{N}}\overline {c}_{n,u}\left(\boldsymbol{\theta}, \boldsymbol{\varphi}\right))-\sum\limits_{n\in\mathcal{N}}\frac{1}{T}\sum\limits_{u\in\mathcal{U}}B_{n,u}\left(\boldsymbol{\rho}, \boldsymbol{\delta} \right)\\
%&+\sum\limits_{n\in\mathcal{N}}\sum\limits_{u\in\mathcal{U}}\sum\limits_{t\in\mathcal{T}}\lambda_{n,u,t}\left(c_{n,u,t}\left(\boldsymbol{\rho}, \boldsymbol{\delta}\right)-c_{n,u,t}\left(\boldsymbol{\theta}, \boldsymbol{\varphi}\right)\right). 
L=&\frac{1}{T\Pi_u}\sum\limits_{u\in\mathcal{U}}\sum\limits_{n\in\mathcal{N}}\sum\limits_{t\in\mathcal{T}} {c}_{n,u,t}\left({\theta}_{nut}\right)-\frac{1}{T}\sum\limits_{n\in\mathcal{N}}\sum\limits_{u\in\mathcal{U}}\sum\limits_{t\in\mathcal{T}}\rho_{nut}+\sum\limits_{n\in\mathcal{N}}\sum\limits_{u\in\mathcal{U}}\sum\limits_{t\in\mathcal{T}}\lambda_{n,u,t}\left({\rho}_{nut}-\theta_{nut}\right)\\
&+\sum\limits_{n\in\mathcal{N}}\left[\frac{1}{T\Pi_n}\sum\limits_{m\in\mathcal{M}}\sum\limits_{t\in\mathcal{T}} {c}_{n,m,t}\left({\varphi}_{nmt}\right)+\frac{1}{T\Pi_n}\sum\limits_{t\in\mathcal{T}} c_{n,S,t}\left({\varepsilon}_{nt}\right)\right]-\frac{1}{T}\sum\limits_{m\in\mathcal{M}}\sum\limits_{n\in\mathcal{N}}\sum\limits_{t\in\mathcal{T}}\delta_{nmt}\\&-\frac{1}{T}\sum\limits_{n\in\mathcal{N}}\sum\limits_{t\in\mathcal{T}}\beta_{nt}
+\sum\limits_{n\in\mathcal{N}}\sum\limits_{m\in\mathcal{M}}\sum\limits_{t\in\mathcal{T}}\varsigma _{n,m,t}\left(\delta_{nmt}-\varphi_{nmt}\right)+\sum\limits_{n\in\mathcal{N}}\sum\limits_{t\in\mathcal{T}}\xi _{n,t}\left(\beta_{nt}-\varepsilon_{nt}\right). 
\end{split}
\end{equation} 
% \begin{equation}\label{eq:L2}
%%\setlength{\abovedisplayskip}{2 pt}
%%\setlength{\belowdisplayskip}{3 pt}
%\begin{split}
%L_2&=\sum\limits_{n\in\mathcal{N}}\left[\frac{1}{MT\Pi_n}\sum\limits_{m\in\mathcal{M}}\sum\limits_{t\in\mathcal{T}} {c}_{n,m,t}\left({\varphi}_{nmt}\right)+\frac{1}{T\Pi_n}\sum\limits_{t\in\mathcal{T}} c_{n,S,t}\left({\varepsilon}_{nt}\right)\right]-\frac{1}{T}\sum\limits_{m\in\mathcal{M}}\sum\limits_{n\in\mathcal{N}}\sum\limits_{t\in\mathcal{T}}\delta_{nmt}\\&-\frac{1}{T}\sum\limits_{n\in\mathcal{N}}\sum\limits_{t\in\mathcal{T}}\beta_{nt}
%+\sum\limits_{n\in\mathcal{N}}\sum\limits_{m\in\mathcal{M}}\sum\limits_{t\in\mathcal{T}}\varsigma _{n,u,t}\left(\delta_{nut}-\varphi_{nmt}\right)+\sum\limits_{n\in\mathcal{N}}\sum\limits_{t\in\mathcal{T}}\xi _{n,t}\left(\beta_{nt}-\varepsilon_{nt}\right). 
%\end{split}
%\end{equation} 
Here, $\lambda_{n,u,t}$, $\varsigma _{n,u,t}$, and $\xi _{n,t}$ are the dual variables.

Meanwhile, since the data rate functions $c_{n,u,t}\left({\theta}_{nut}\right)$, $c_{n,m,t}\left({\varphi}_{nmt}\right)$, and $c_{n,S,t}\left({\varepsilon}_{nt}\right)$ are independent, one can obtain the augmented Lagrangian function transformed into a form $L=\sum\limits_{u\in\mathcal{U}}L_u+\sum\limits_{n\in\mathcal{N}}L_n+\sum\limits_{m\in\mathcal{M}}L_m+L_S$. Here, $L_u$ is the Lagrangian function at user $u$, such that user $u\in\mathcal{U}$ is obligated to solve a maximization problem in the form:
%\addtocounter{equation}{0}
%%\setlength{\abovedisplayskip}{-5 pt}
\begin{equation}\label{eq:optu}
\setlength{\abovedisplayskip}{2 pt}
\setlength{\belowdisplayskip}{0 pt}
\begin{split}
%\max_{\boldsymbol{\rho}, \boldsymbol{\beta}}\sum\limits_{u\in\mathcal{U}}\sum\limits_{n\in\mathcal{N}}\overline {c}_{n,u}\left(\boldsymbol{\rho}\right)+\sum\limits_{n\in\mathcal{N}}\sum\limits_{m\in\mathcal{M}}\overline {c}_{n,m}\left(\boldsymbol{\rho}\right)-\sum\limits_{n\in\mathcal{N}}P_n\sum\limits_{u\in\mathcal{U}}B_{u,n}\left(\boldsymbol{\rho} \right)\tau-\sum\limits_{m\in\mathcal{M}}P_m\sum\limits_{n\in\mathcal{N}}B_{n,m}\left(\boldsymbol{\rho} \right)\tau-P_S\sum\limits_{n\in\mathcal{N}}B_{n,S}\left(\boldsymbol{\beta} \right)\tau,
&D_u\left(\boldsymbol{\theta}_u,\boldsymbol{\lambda}_u\right)=\max_{\boldsymbol{\theta}_u\in{X}_{u}} L_u \left(\boldsymbol{\theta}_u,\boldsymbol{\lambda}_u\right)=\max_{\boldsymbol{\theta}_u}\left\{\frac{1}{T\Pi_u}\sum\limits_{n\in\mathcal{N}}\sum\limits_{t\in\mathcal{T}} {c}_{n,u,t}\left({\theta}_{nut}\right)-\sum\limits_{n\in\mathcal{N}}\sum\limits_{t\in\mathcal{T}}\lambda_{n,u,t}{\theta}_{nut}\right\},
\end{split}
\end{equation}
%\vspace{-0.0cm}
%\begin{align}\label{c1}
%%\setlength{\abovedisplayskip}{-5 pt}
%%\setlength{\belowdisplayskip}{0 pt}
%&\;\;\;\;\rm{s.\;t.}\scalebox{1}{$\;\;\;\;\;\;\;\;\;\;\; \nonumber ({\ref{eq:opt}}a)-({\ref{eq:opt}}c), ({\ref{eq:opt}}l) $}
%\end{align}
%$L_u=\frac{1}{NT\Pi_u}\sum\limits_{n\in\mathcal{N}}\sum\limits_{t\in\mathcal{T}}{c}_{n,u,t}\left({\theta}_{n,u,t}\right)-\sum\limits_{n\in\mathcal{N}}\sum\limits_{t\in\mathcal{T}}\lambda_{n,u,t}{\theta}_{nut}$ is the Lagrangian function of user $u$ in the market. 
where $\mathcal{X}_{u}=\left\{\boldsymbol{\theta}_u: ({\ref{eq:opt}}a), ({\ref{eq:opt}}g), ({\ref{eq:opt}}n), \theta_{nut}\in\left\{0,1\right\}, \forall n\in\mathcal{N}, t\in\mathcal{T}\right\}$ includes the local constraints at user $u$. Here, $\boldsymbol{\theta}_{u}=\left[{\theta}_{1u1},\cdots,{\theta}_{1uT},\cdots, {\theta}_{Nu1},\cdots,{\theta}_{NuT}\right]$ is the vector of data request schemes of user $u$, with ${\theta}_{nut}\in\left\{0,1\right\}$. Note that, ${\theta}_{nut}=1$ implies that user $u$ requests data service from BS $n$ at time slot $t$. $\boldsymbol{\lambda}_{u}=\left[{\lambda}_{1u1},\cdots,{\lambda}_{1uT},\cdots, {\lambda}_{Nu1},\cdots,{\lambda}_{NuT}\right]$ is the vector of dual variables related to user $u$. Also, note that $D_u\left(\boldsymbol{\theta}_u,\boldsymbol{\lambda}_u\right)$ is the dual function at user $u$.
$L_n$ is the Lagrangian function at BS $n$, such that the optimization problem at BS $n$ is defined as follows:
\begin{equation}\label{eq:optn}
\begin{split}
%\max_{\boldsymbol{\rho}, \boldsymbol{\beta}}\sum\limits_{u\in\mathcal{U}}\sum\limits_{n\in\mathcal{N}}\overline {c}_{n,u}\left(\boldsymbol{\rho}\right)+\sum\limits_{n\in\mathcal{N}}\sum\limits_{m\in\mathcal{M}}\overline {c}_{n,m}\left(\boldsymbol{\rho}\right)-\sum\limits_{n\in\mathcal{N}}P_n\sum\limits_{u\in\mathcal{U}}B_{u,n}\left(\boldsymbol{\rho} \right)\tau-\sum\limits_{m\in\mathcal{M}}P_m\sum\limits_{n\in\mathcal{N}}B_{n,m}\left(\boldsymbol{\rho} \right)\tau-P_S\sum\limits_{n\in\mathcal{N}}B_{n,S}\left(\boldsymbol{\beta} \right)\tau,
D_n&\left(\boldsymbol{\varphi}_n,\boldsymbol{\varepsilon}_n,\boldsymbol{\rho}_n, \boldsymbol{\lambda}_n,\boldsymbol{\xi}_n,\boldsymbol{\varsigma}_n\right)=\max_{\boldsymbol{\varphi}_n,\boldsymbol{\varepsilon}_n,\boldsymbol{\rho}_n, \in\mathcal{X}_{n}} L_n\left(\boldsymbol{\varphi}_n,\boldsymbol{\varepsilon}_n,\boldsymbol{\rho}_n, \boldsymbol{\lambda}_n,\boldsymbol{\xi}_n,\boldsymbol{\varsigma}_n\right) \\
&=\max_{\boldsymbol{\varphi}_n,\boldsymbol{\varepsilon}_n,\boldsymbol{\rho}_n, \in\mathcal{X}_{n}}\frac{1}{T\Pi_n}\sum\limits_{m\in\mathcal{M}}\sum\limits_{t\in\mathcal{T}} {c}_{n,m,t}\left({\varphi}_{n,m,t}\right)+\frac{1}{T\Pi_n}\sum\limits_{t\in\mathcal{T}} c_{n,S,t}\left({\varepsilon}_{nt} \right)\\
&-\sum\limits_{m\in\mathcal{M}}\sum\limits_{t\in\mathcal{T}}\varsigma_{n,m,t}{\varphi}_{nmt}-\sum\limits_{t\in\mathcal{T}}\xi_{n,S,t}{\varepsilon}_{nt}+\sum\limits_{u\in\mathcal{U}}\sum\limits_{t\in\mathcal{T}}\lambda_{n,u,t}{\rho}_{nut}-\frac{1}{T}\sum\limits_{u\in\mathcal{U}}\sum\limits_{t\in\mathcal{T}}{\rho}_{n,u,t},
\end{split}
\end{equation}
%The Lagrangian function at BS $n$ is $L_n=\frac{1}{MT\Pi_n}\sum\limits_{m\in\mathcal{M}}\sum\limits_{t\in\mathcal{T}} {c}_{n,m,t}\left({\varphi}_{n,m,t}\right)+\frac{1}{T\Pi_n}\sum\limits_{t\in\mathcal{T}} c_{n,S,t}\left({\varepsilon}_{nt} \right)-\sum\limits_{m\in\mathcal{M}}\sum\limits_{t\in\mathcal{T}}\varsigma_{n,m,t}{\varphi}_{nmt}-\sum\limits_{t\in\mathcal{T}}\xi_{n,S,t}{\varepsilon}_{nt}+\sum\limits_{u\in\mathcal{U}}\sum\limits_{t\in\mathcal{T}}\lambda_{n,u,t}{\rho}_{nut}-\frac{1}{T}\sum\limits_{u\in\mathcal{U}}\sum\limits_{t\in\mathcal{T}}{\rho}_{n,u,t}$. 
where $D_n\left(\boldsymbol{\varphi}_n,\boldsymbol{\varepsilon}_n,\boldsymbol{\rho}_n, \boldsymbol{\lambda}_n,\boldsymbol{\xi}_n,\boldsymbol{\varsigma}_n\right)$ is called the dual function at BS $n$.
%$\begin {\multline} {X}_{n}=\left{\right.   \boldsymbol{\varphi}_n,\boldsymbol{\varepsilon}_n,\boldsymbol{\rho}_n: ({\ref{eq:opt}}d)-({\ref{eq:opt}}h), ({\ref{eq:opt}}i), ({\ref{eq:opt}}m), ({\ref{eq:opt}}n), {\varphi}_{n,m,t}, {\varepsilon}_{nt}, {\rho}_{nut}\in\left\{0,1\right\},   \\       \left.  \forall m\in\mathcal{M},u\in\mathcal{U}, t\in\mathcal{T}       \right} \end {\multline} $
Also, note that $\mathcal{X}_{n}=\small{\left\{\boldsymbol{\varphi}_n,\boldsymbol{\varepsilon}_n,\boldsymbol{\rho}_n: ({\ref{eq:opt}}b)-({\ref{eq:opt}}d), ({\ref{eq:opt}}h), ({\ref{eq:opt}}k)-({\ref{eq:opt}}m), ({\ref{eq:opt}}o), {\varphi}_{n,m,t}, {\varepsilon}_{nt}, {\rho}_{nut}\in\left\{0,1\right\}, \forall m\in\mathcal{M},u\in\mathcal{U}, t\in\mathcal{T}\right\}}$ includes all the local constraints at BS $n$. Here, $\boldsymbol{\varphi}_{n}=\left[{\varphi}_{n11},\cdots,{\varphi}_{n1T},\cdots,{\varphi}_{nM1},\cdots,{\varphi}_{nMT}\right]$, and $\boldsymbol{\varepsilon}_{n}=\left[{\varepsilon}_{n1},\cdots,{\varepsilon}_{nT}\right]$ are vectors of the data request schemes at BS $n$, with ${\varphi}_{nmt}=1$ indicating that BS $n$ requests data service from BS $m$, ${\varepsilon}_{nt}=1$ indicating that BS $n$ requests data service from the satellite, at time slot $t$.
$\boldsymbol{\rho}_{n}=\left[{\rho}_{n11},\cdots,{\rho}_{n1T},\cdots, {\rho}_{nU1},\cdots,{\rho}_{nUT}\right]$ is the allocation scheme vector at BS $n$ in the studied duration. $\boldsymbol{\lambda}_{n}=\left[{\lambda}_{n11},\cdots,{\lambda}_{n1T},\cdots, {\lambda}_{nU1},\cdots,{\lambda}_{nUT}\right]$, $\boldsymbol{\varsigma}_{n}=\left[{\varsigma}_{n11},\cdots,{\varsigma}_{n1T},\cdots, {\varsigma}_{nU1},\cdots,{\varsigma}_{nUT}\right]$, and $\boldsymbol{\xi}_{n}=\left[{\xi}_{n1},\cdots,{\xi}_{nT}\right]$ are the vectors of dual variables related to BS $n$. { Note that, instead of calculating the data rate it can receive from the satellite at slot $t$ by tracking the location of the satellite, a BS $n$ can estimate the data rate based on Proposition \ref{pp1}. 
\begin{proposition}\label{pp1}\emph{
Given an allocation vector $\boldsymbol{\beta}$, the variation of backhaul data rate $c_{n,S,t}$ caused by the movement of the LEO satellite is: 
 \begin{equation}\label{eq:pp2}
\setlength{\abovedisplayskip}{2 pt}
\setlength{\belowdisplayskip}{3 pt}
\begin{split}
%V_{S}&\left(\boldsymbol{\rho}, \boldsymbol{\delta}, \boldsymbol{\beta} \right)=\frac{1}{T\Pi_n}\sum\limits_{n\in\mathcal{N}}\sum\limits_{t\in\mathcal{T}}q_{n,S,t}c_{n,S,t}\left(\boldsymbol{\rho}, \boldsymbol{\delta},\boldsymbol{\beta}\right)-\frac{1}{T}\sum\limits_{n\in\mathcal{N}}\sum\limits_{t\in\mathcal{T}}\beta_{nt}.
%V_{S}&\left(\boldsymbol{\rho}, \boldsymbol{\delta}, \boldsymbol{\beta} \right)=\frac{1}{B_{n,S}\left(\boldsymbol{\beta}\right)\Pi_n}\sum\limits_{n\in\mathcal{N}}\sum\limits_{t\in\mathcal{T}}q_{n,S,t}c_{n,S,t}\left(\boldsymbol{\rho}, \boldsymbol{\delta},\boldsymbol{\beta}\right)-\frac{1}{T}\sum\limits_{n\in\mathcal{N}}B_{n,S}\left(\boldsymbol{\beta} \right).
\tilde c_{n,S,t} \approx \frac{-\alpha \log_2 10 \beta_{nt}B_SP_SL_1Q_{MM}Q_{RM} 10^{-\frac{0.1\alpha'+0.1\alpha\log_{10}d_{n,S}\left(t\right)+0.1\chi}{10}}\left(\upsilon ^2t+\left(x_0-x_{2,n}\right)\upsilon\right)\tau^2}{10\Omega_{n,S,t}d^2_{n,S}\left(t\right)\left(1+ 10^{-\frac{0.1\alpha'+0.1\alpha\log_{10}d_{n,S}\left(t\right)+0.1\chi}{10}}\right)}.
\end{split}
\end{equation} }
\end{proposition}
\begin{proof}
See Appendix B.
\end{proof}
}
As such, using Proposition 1, the BSs can evaluate the satellite backhaul data rates at slot $t$ in the form of $c_{n,S,t-1}+\tilde c_{n,S,t}$, without introducing unnecessary overhead while tracking the location of the LEO satellite.
The Lagrangian function at MBS $m$ is $L_m$. The optimization problem at MBS $m$ is:
\begin{equation}\label{eq:optm}
\begin{split}
%\max_{\boldsymbol{\rho}, \boldsymbol{\beta}}\sum\limits_{u\in\mathcal{U}}\sum\limits_{n\in\mathcal{N}}\overline {c}_{n,u}\left(\boldsymbol{\rho}\right)+\sum\limits_{n\in\mathcal{N}}\sum\limits_{m\in\mathcal{M}}\overline {c}_{n,m}\left(\boldsymbol{\rho}\right)-\sum\limits_{n\in\mathcal{N}}P_n\sum\limits_{u\in\mathcal{U}}B_{u,n}\left(\boldsymbol{\rho} \right)\tau-\sum\limits_{m\in\mathcal{M}}P_m\sum\limits_{n\in\mathcal{N}}B_{n,m}\left(\boldsymbol{\rho} \right)\tau-P_S\sum\limits_{n\in\mathcal{N}}B_{n,S}\left(\boldsymbol{\beta} \right)\tau,
D_m\left(\boldsymbol{\delta}_m, \boldsymbol{\varsigma}_m\right)&=\max_{\boldsymbol{\delta}_m \in\mathcal{X}_{m}} L_m\left(\boldsymbol{\delta}_m, \boldsymbol{\varsigma}_m\right) =\max_{\boldsymbol{\delta}_m \in\mathcal{X}_{m}} \sum\limits_{n\in\mathcal{N}}\sum\limits_{t\in\mathcal{T}}\varsigma_{n,m,t}{\delta}_{nmt}-\frac{1}{T}\sum\limits_{n\in\mathcal{N}}\sum\limits_{t\in\mathcal{T}}{\delta}_{n,m,t},
\end{split}
\end{equation}
where $\mathcal{X}_{m}=\left\{\boldsymbol{\delta}_m: ({\ref{eq:opt}}e), ({\ref{eq:opt}}i), \delta_{nmt}\in\left\{0,1\right\}, \forall n\in\mathcal{N}, t\in\mathcal{T}\right\}$ includes all the local constraints at MBS $m$. $\boldsymbol{\delta}_{m}=\left[{\delta}_{1m1},\cdots,{\delta}_{1mT},\cdots, {\delta}_{Nm1},\cdots,{\delta}_{NmT}\right]$ is the allocation scheme vector at MBS $m$. $\boldsymbol{\varsigma}_{m}=\left[{\varsigma}_{1m1},\cdots,{\varsigma}_{1mT},\cdots, {\varsigma}_{Nm1},\cdots,{\rho}_{NmT}\right]$ is the vector of dual variables related to MBS $m$. Here, $D_m\left(\boldsymbol{\delta}_m, \boldsymbol{\varsigma}_m\right)$ is the dual function at MBS $m$.
%On the other hand, the Lagrangian function at MBS $m$ is formed as $L_m=\sum\limits_{n\in\mathcal{N}}\sum\limits_{t\in\mathcal{T}}\varsigma_{n,m,t}{\delta}_{nmt}-\frac{1}{T}\sum\limits_{n\in\mathcal{N}}\sum\limits_{t\in\mathcal{T}}{\delta}_{n,m,t}$, 
The Lagrangian function at the satellite is $L_S$, which is maximized as follows:
\begin{equation}\label{eq:opts}
\setlength{\abovedisplayskip}{1 pt}
\setlength{\belowdisplayskip}{1 pt}
\begin{split}
%\max_{\boldsymbol{\rho}, \boldsymbol{\beta}}\sum\limits_{u\in\mathcal{U}}\sum\limits_{n\in\mathcal{N}}\overline {c}_{n,u}\left(\boldsymbol{\rho}\right)+\sum\limits_{n\in\mathcal{N}}\sum\limits_{m\in\mathcal{M}}\overline {c}_{n,m}\left(\boldsymbol{\rho}\right)-\sum\limits_{n\in\mathcal{N}}P_n\sum\limits_{u\in\mathcal{U}}B_{u,n}\left(\boldsymbol{\rho} \right)\tau-\sum\limits_{m\in\mathcal{M}}P_m\sum\limits_{n\in\mathcal{N}}B_{n,m}\left(\boldsymbol{\rho} \right)\tau-P_S\sum\limits_{n\in\mathcal{N}}B_{n,S}\left(\boldsymbol{\beta} \right)\tau,
D_S\left(\boldsymbol{\beta}, \boldsymbol{\xi}\right)=\max_{\boldsymbol{\beta} \in\mathcal{X}_{S}} L_S\left(\boldsymbol{\beta}, \boldsymbol{\xi}\right) =&\max_{\boldsymbol{\beta} \in\mathcal{X}_{S}}\sum\limits_{n\in\mathcal{N}}\sum\limits_{t\in\mathcal{T}}\xi_{n,S,t} {\beta}_{nt}-\frac{1}{T}\sum\limits_{n\in\mathcal{N}}\sum\limits_{t\in\mathcal{T}}{\beta}_{nt},
\end{split}
\end{equation}
where $\mathcal{X}_{S}=\left\{\boldsymbol{\beta}: ({\ref{eq:opt}}f), ({\ref{eq:opt}}j), \beta_{nt}\in\left\{0,1\right\}, \forall n\in\mathcal{N}, t\in\mathcal{T}\right\}$ includes all the local constraints at the satellite. $\boldsymbol{\xi}=\left[{\xi}_{11},\cdots,{\xi}_{1T},\cdots, {\xi}_{N1},\cdots,{\xi}_{NT},\right]$ is the vector of dual variables related to the satellite. Here, $D_S\left(\boldsymbol{\beta}, \boldsymbol{\xi}\right)$ is the dual function at the satellite.
%$ =\sum\limits_{n\in\mathcal{N}}\sum\limits_{t\in\mathcal{T}}\xi_{n,S,t} {\beta}_{nt}-\frac{1}{T}\sum\limits_{n\in\mathcal{N}}\sum\limits_{t\in\mathcal{T}}{\beta}_{nt}$.   
Clearly, these Lagrangian functions are actually the payoff functions at each of the user, BS, MBS, and the satellite, with dual variables $\lambda_{n,u,t}$, $\varsigma_{n,m,t}$, and $\xi_{n,S,t}$ functioning as normalized price $\frac{1}{T}p_{n,u,t}$, $\frac{1}{T}p_{n,m,t}$, and $\frac{1}{T}p_{n,S,t}$. The vectors of dual variables are noted as $\boldsymbol{\lambda}=\left[\lambda_{1,1,1},\cdots,\lambda_{N,U,T}\right]$, $\boldsymbol{\varsigma}=\left[\varsigma_{1,1,1},\cdots,\varsigma_{N,M,T}\right]$, and $\boldsymbol{\xi}=\left[\xi_{1,S,1},\cdots,\xi_{N,S,T}\right]$, respectively.

 \begin{lemma}\label{lemma1}
\emph{There exist at least one Walrasian equilibrium in the ISDN system.}% with any $\boldsymbol{\lambda}$, $\boldsymbol{\varsigma}$, and $\boldsymbol{\xi}$. }

\end{lemma} 
\begin{proof} See Appendix C
 \end{proof}
%at Lagrangian functions at each user, BS, MBS and the satellite, dual variables $\lambda_{n,u,t}$, $\varsigma_{n,m,t}$ and $\xi_{n,S,t}$ functions as price $p_{n,u,t}$, $p_{n,m,t}$ and $p_{n,S,t}$. 
 
In summary, given the local constraints, ({\ref{eq:opt}}a)-({\ref{eq:opt}}r), the feasible region of the INLP problem in (\ref{eq:opt}) can be separated into several regions, i.e. $\mathcal{X}_{u}$, $\mathcal{X}_{n}$, $\mathcal{X}_{m}$, and $\mathcal{X}_{S}$, each of which serves as the domain of the optimization problem at one user, BS, MBS or the satellite. Precisely, the optimization problem at one user, BS, MBS or the satellite is the maximizing the value of its Lagrangian function, or its payoff function. The resulting maximal value of payoff at one user, BS, MBS or the satellite, with given dual variables, is defined as the dual function at this user, BS, MBS or the satellite. The dual of the INLP problem in (\ref{eq:opt}),  $D=\sum\limits_{u\in\mathcal{U}}D_u+\sum\limits_{n\in\mathcal{N}}D_n+\sum\limits_{m\in\mathcal{M}}D_m+D_S$, is formulated as the sum of maximal payoffs at all the users, BSs, MBSs, and the satellite in the system with given dual variables $\boldsymbol{\lambda}$, $\boldsymbol{\varsigma}$, and $\boldsymbol{\xi}$. The dual of (\ref{eq:opt}) is then minimized at each user, BS, MBS, and the satellite over dual variables $\boldsymbol{\lambda}$, $\boldsymbol{\varsigma}$, and $\boldsymbol{\xi}$, with the heavy ball method \cite{boyd2004convex}. With the convergence of the heavy ball algorithm, the market clearance constraints is satisfied, which gives rise to a Walrasian equilibrium of the ISDN system, as the maximal values for the payoffs of each user, BS, MBS, and the satellite are reached.

\vspace{-0.3cm}
  \subsection{Heavy ball method}
  The heavy ball method \cite{boyd2004convex} is an efficient sub-gradient method that seeks to find the optimal dual variables that minimizes the dual function, $D$, by iteratively updating the dual variables $\boldsymbol{\lambda}$, $\boldsymbol{\varsigma}$, and $\boldsymbol{\xi}$ in the opposite direction to the gradient $\nabla L(\boldsymbol{\lambda})$, $\nabla L(\boldsymbol{\varsigma})$, and $\nabla L(\boldsymbol{\xi})$, respectively.
  The dual function $D$ is, firstly, computed with initial dual variables $\boldsymbol{\lambda}^{\left(0\right)}$, $\boldsymbol{\varsigma}^{\left(0\right)}$ and $\boldsymbol{\xi}^{\left(0\right)}$. However, the allocation schemes that maximize payoff functions with such initial dual variables, or normalized prices, may not satisfy the market clearance constraints, that is, the data and backhaul service supply may not match the data and backhaul service demand. The mismatch between supply and demand is referred as to violation of the market clearance constrains, which is:
  \begin{equation}\label{eq:vl1}
\begin{split}
&\boldsymbol{s}^{\left(k\right)}_1=\boldsymbol{\theta}^{\left(k\right)}-\boldsymbol{\rho}^{\left(k\right)}, \;\;\;\;\;\boldsymbol{s}^{\left(k\right)}_2=\boldsymbol{\varphi}^{\left(k\right)}-\boldsymbol{\delta}^{\left(k\right)},\;\;\;\;\; \boldsymbol{s}^{\left(k\right)}_3=\boldsymbol{\varepsilon}^{\left(k\right)}-\boldsymbol{\beta}^{\left(k\right)},
\end{split}
\end{equation}
%  \begin{equation}\label{eq:vl2}
%\begin{split}
%&\boldsymbol{s}^{\left(k\right)}_2=\boldsymbol{\varphi}^{\left(k\right)}-\boldsymbol{\delta}^{\left(k\right)},
%\end{split}
%\end{equation}
%  \begin{equation}\label{eq:vl3}
%\begin{split}
%&\boldsymbol{s}^{\left(k\right)}_3=\boldsymbol{\varepsilon}^{\left(k\right)}-\boldsymbol{\beta}^{\left(k\right)},
%\end{split}
%\end{equation}
where $\boldsymbol{s}^{\left(k\right)}_1$, $\boldsymbol{s}^{\left(k\right)}_2$, and $\boldsymbol{s}^{\left(k\right)}_3$, are vectors of market clearance violations at radio access links, terrestrial backhaul links, and satellite backhaul links at iteration $k$ of the heavy ball dual variable updating process, respectively.  $\boldsymbol{\rho}^{\left(k\right)}$, $\boldsymbol{\delta}^{\left(k\right)}$, and $\boldsymbol{\beta}^{\left(k\right)}$ are the optimal allocation vectors that maximize Lagrangian functions $L_n$, $L_m$, and $L_S$ at each BS, MBS, and the satellite in the ISDN, at iteration $k$.
 $\boldsymbol{\theta}^{\left(k\right)}$, $\boldsymbol{\varsigma}^{\left(k\right)}$, and $\boldsymbol{\xi}^{\left(k\right)}$ are the optimal request vectors that maximize Lagrangian functions $L_u$ and $L_n$, for all the users and BSs in the system, at iteration $k$. As these market clearance violations are considered to reach $0$ at the Walrasian equilibrium of the ISDN, dual variables are updated iteratively at the BSs, MBSs, and the satellite for the decreased violation values.
Precisely, in the studied ISDN, each user and BS report their optimal service request scheme  $\boldsymbol{\theta}^{\left(k\right)}$, $\boldsymbol{\varsigma}^{\left(k\right)}$, and $\boldsymbol{\xi}^{\left(k\right)}$ to their associated BS, MBS, and satellite, at each iteration $k$.
The BSs will be responsible for updating prices $\boldsymbol{\lambda}$ based on the market clearance violation $\boldsymbol{s}^{\left(k\right)}_1$ derived from (\ref{eq:vl1}), such that $\boldsymbol{\lambda}^{\left(k+1\right)}=\boldsymbol{\lambda}^{\left(k\right)}+{\pi}^{\left(k\right)}\boldsymbol{\mu}^{\left(k\right)}_1$ %\frac{\boldsymbol{s}^{\left(k\right)}_1}{\left\| {\boldsymbol{s}^{\left(k\right)}_1} \right\|}+\psi^{\left(k\right)}\left(\boldsymbol{\lambda}^{\left(k\right)}-\boldsymbol{\lambda}^{\left(k-1\right)}\right)$
, where $\boldsymbol{\mu}^{\left(k\right)}_1=\frac{\boldsymbol{s}^{\left(k\right)}_1}{\left\| {\boldsymbol{s}^{\left(k\right)}_1} \right\|}+\nu^{\left(k\right)}_1 \boldsymbol{\mu}^{\left(k-1\right)}_1$ with $\nu^{\left(k\right)}_1=\max\left\{0, -1.5\frac{\boldsymbol{s}^{\left(k\right)}_1\left(\boldsymbol{\mu}^{\left(k-1\right)}_1\right)^T}{\left\| {\boldsymbol{s}^{\left(k\right)}_1} \right\|{\left\| {\boldsymbol{\mu}^{\left(k-1\right)}_1} \right\|}} \right\}$.
 MBSs are responsible for updating prices  $\boldsymbol{\varsigma}$, based on the market clearance violation $\boldsymbol{s}^{\left(k\right)}_2$ derived from (\ref{eq:vl1}), with $\boldsymbol{\varsigma}^{\left(k+1\right)}=\boldsymbol{\varsigma}^{\left(k\right)}+{\pi}^{\left(k\right)}\boldsymbol{\mu}^{\left(k\right)}_2$, %\frac{\boldsymbol{s}^{\left(k\right)}_2}{\left\| {\boldsymbol{s}^{\left(k\right)}_2} \right\|}+\psi^{\left(k\right)}\left(\boldsymbol{\varsigma}^{\left(k\right)}-\boldsymbol{\varsigma}^{\left(k-1\right)}\right)$, 
 where $\boldsymbol{\mu}^{\left(k\right)}_2=\frac{\boldsymbol{s}^{\left(k\right)}_2}{\left\| {\boldsymbol{s}^{\left(k\right)}_2} \right\|}+\nu^{\left(k\right)}_2 \boldsymbol{\mu}^{\left(k-1\right)}_2$ with $\nu^{\left(k\right)}_2=\max\left\{0,-1.5\frac{{{\bf{s}}_2^{\left( k \right)}{{\left( {\mu _2^{\left( {k - 1} \right)}} \right)}^T}}}{{\left\| {{\bf{s}}_2^{\left( k \right)}} \right\|\left\| {\mu _2^{\left( {k - 1} \right)}} \right\|}}  \right\}$.  
 The satellite, at the mean time, is obligated to update prices $\boldsymbol{\xi}$ based on $\boldsymbol{\xi}^{\left(k+1\right)}=\boldsymbol{\xi}^{\left(k\right)}+{\pi}^{\left(k\right)}\boldsymbol{\mu}^{\left(k\right)}_3$, %\frac{\boldsymbol{s}^{\left(k\right)}_3}{\left\| {\boldsymbol{s}^{\left(k\right)}_3} \right\|}+\psi^{\left(k\right)}\left(\boldsymbol{\xi}^{\left(k\right)}-\boldsymbol{\xi}^{\left(k-1\right)}\right)$, 
 where $\boldsymbol{\mu}^{\left(k\right)}_3=\frac{\boldsymbol{s}^{\left(k\right)}_3}{\left\| {\boldsymbol{s}^{\left(k\right)}_3} \right\|}+\nu^{\left(k\right)}_3 \boldsymbol{\mu}^{\left(k-1\right)}_3$ with $\nu^{\left(k\right)}_3=\max\left\{0, -1.5\frac{\boldsymbol{s}^{\left(k\right)}_3\left(\boldsymbol{\mu}^{\left(k-1\right)}_3\right)^T}{\left\| {\boldsymbol{s}^{\left(k\right)}_3} \right\|{\left\| {\boldsymbol{\mu}^{\left(k-1\right)}_3} \right\|}} \right\}$. 
Here, $\boldsymbol{\lambda}^{\left(k\right)}$, $\boldsymbol{\varsigma}^{\left(k\right)}$ and $\boldsymbol{\xi}^{\left(k\right)}$ are vectors of the dual variables at the $k$-th iteration, $\pi^{\left(k\right)}$ is the updating step size in the price updating process. %, $\psi^{\left(k\right)}$ is the weight of the memory terms in the updating process.
Then, the terrestrial users, BSs, MBSs, and the satellite are informed with the updated prices $\boldsymbol{\lambda}^{\left(k\right)}$, $\boldsymbol{\varsigma}^{\left(k\right)}$ and $\boldsymbol{\xi}^{\left(k\right)}$ and will find the optimal solutions of their Lagrangian function maximization problems,  $\boldsymbol{\theta}^{\left(k+1\right)}$, $\boldsymbol{\varsigma}^{\left(k+1\right)}$, $\boldsymbol{\xi}^{\left(k+1\right)}$, $\boldsymbol{\rho}^{\left(k\right)}$, $\boldsymbol{\delta}^{\left(k\right)}$, and $\boldsymbol{\beta}^{\left(k\right)}$ with the updated dual variables. If these new optimal solutions still do not satisfy the market clearance constraints, the dual variables will keep being updated.
\begin{algorithm}[t]\footnotesize
\caption{Proposed algorithm for ISDN total payoff maximization problem. }   
\label{alg:Framwork}   
\setlength{\abovecaptionskip}{-15pt} 
\setlength{\belowcaptionskip}{-15pt}
\begin{algorithmic} [1] %è¿ä¸ª1 è¡¨ç¤ºæ¯ä¸è¡é½æ¾ç¤ºæ°å­  
\REQUIRE The working time limits of each BS, data service request of each user, QoS requirements at each user and BS. \\ 
\vspace{2pt}  
\ENSURE Initialize dual variables $\boldsymbol{\lambda}^{\left(0\right)}$, $\boldsymbol{\varsigma}^{\left(0\right)}$ and $\boldsymbol{\xi}^{\left(0\right)}$, which are the initial normalized prices.\\%, such that $\Pr \left( {{a_n}} \right)=\Pr \left( {{b_m}} \right)=\frac{1}{2}$ satisfies equation (12).  \\  
\vspace{2pt}  
\STATE Set step $k=1$.
\vspace{2pt}  
\FOR {User $u=1:U$} 
\vspace{2pt}  
\STATE Solve optimization problem (\ref{eq:optu}) with dual variable $\boldsymbol{\lambda}^{\left(k-1\right)}$ and find the optimal solution $\boldsymbol{\theta}^{\left(k-1\right)}$.
\vspace{2pt}  
\STATE Send scheme $\boldsymbol{\theta}^{\left(k-1\right)}$ to BSs.
\vspace{2pt}  
\ENDFOR
\vspace{2pt}  
\FOR {BS $n=1:N$}
\vspace{2pt}  
\STATE Solve optimization problem (\ref{eq:optn}) with dual variables$\boldsymbol{\lambda}^{\left(k-1\right)}$, $\boldsymbol{\varsigma}^{\left(k-1\right)}$ and $\boldsymbol{\xi}^{\left(k-1\right)}$ and find the optimal solution $\boldsymbol{\rho}^{\left(k-1\right)}$, $\boldsymbol{\varphi}^{\left(k-1\right)}$ and $\boldsymbol{\varepsilon}^{\left(k-1\right)}$.
\vspace{2pt}  
\STATE Compute market clearance violation $\boldsymbol{s}^{\left(k\right)}_1$, update price $\boldsymbol{\lambda}^{\left(k\right)}$. 
\vspace{2pt}  
\STATE Send schemes $\boldsymbol{\varphi}^{\left(k-1\right)}$ and $\boldsymbol{\varepsilon}^{\left(k-1\right)}$ to MBSs and the satellite, respectively.
\vspace{2pt} 
\ENDFOR 
\vspace{2pt}  
\FOR {MBS $m=1:M$} 
\vspace{2pt}  
\STATE Solve optimization problem (\ref{eq:optm}) with dual variable $\boldsymbol{\varsigma}^{\left(k-1\right)}$ and find the optimal solution $\boldsymbol{\delta}^{\left(k-1\right)}$.  
\vspace{2pt}  
\STATE Compute market clearance violation $\boldsymbol{s}^{\left(k\right)}_2$, update price $\boldsymbol{\varsigma}^{\left(k\right)}$. 
\vspace{2pt}  
\ENDFOR  
\FOR {The satellite} 
\vspace{2pt}  
\STATE Solve optimization problem (\ref{eq:opts}) with dual variable $\boldsymbol{\xi}^{\left(k-1\right)}$ and find the optimal solution $\boldsymbol{\beta}^{\left(k-1\right)}$. 
\vspace{2pt}  
\STATE Compute market clearance violation $\boldsymbol{s}^{\left(k\right)}_3$, update price $\boldsymbol{\xi}^{\left(k\right)}$.    
\vspace{2pt}  
\ENDFOR  
%\vspace{2pt}  
%\STATE Update price $\lambda^{k+1}_{n,u,t}=\lambda^{k}_{n,u,t}+\pi_{nut} \left(\theta^k_{nut}-\rho^k_{nut}\right)$, for all $u\in\mathcal{U}$, $n\in\mathcal{N}$ and $t\in\mathcal{T}$. 
%\vspace{2pt}  
%\STATE Update price $\varsigma^{k+1}_{n,m,t}=\varsigma^{k}_{n,m,t}+\pi_{nmt}\left(\varphi^k_{nmt}-\delta^k_{nmt}\right)$, for all $n\in\mathcal{N}$, $m\in\mathcal{M}$ and $t\in\mathcal{T}$. 
%\vspace{2pt}  
%\STATE Update price $\xi^{k+1}_{n,S,t}=\xi^{k}_{n,S,t}+\pi_{n,S,t}\left(\varepsilon^k_{nt}-\beta^k_{nt}\right)$, for all $n\in\mathcal{N}$ and $t\in\mathcal{T}$. 
\vspace{2pt}  
\STATE Set step $k=k+1$. 
\REPEAT
\STATE Step 2 until Step 19. %considering at the t-th iteration the links indicated in At as unassigned links.
\UNTIL{Convergence}
%\vspace{2pt}
%\REPEAT
%\STATE Step 1 until Step 10 considering at the t-th iteration the links indicated in At as unassigned links. 
%\UNTILL{Convergence}
%\vspace{2pt}  
%\STATE Estimate the optimal satellite backhaul allocation vector $\boldsymbol{\rho}^*_S$ based on ({\ref{eq:15}}). 
\vspace{2pt}  
\RETURN{ Optimal allocation vectors  $\boldsymbol{\theta}^*=\boldsymbol{\rho}^*$, $\boldsymbol{\varphi}^*=\boldsymbol{\delta}^*$ and $\boldsymbol{\varepsilon}^*=\boldsymbol{\beta}^*$}. 
\end{algorithmic}
\end{algorithm} 
%\begin{algorithmic}
%\REPEAT
%\STATE carry out some processing
%\UNTIL{some condition is met}
%\end{algorithmic}
\vspace{-0.3cm}
\subsection{Complexity and convergence}
At each iteration, the complexity of the heavy ball algorithm is $\mathcal{O}\left(UN+N\left(M+1\right)\right)$. Also, at each iteration of the heavy ball algorithm, the BSs, MBSs, and the satellite broadcast the updated dual variables, the users and the BSs report its service request schemes to the BSs, MBSs and the satellite. %These dual variables and the service request schemes contain only some amount of considerably small real numbers. 
The number of information expected to be exchanged in the distributed heavy ball algorithm is $K^*\left(U+2N+M+1\right)$, with $K^*$, a considerably small number, being the index of the iteration when the heavy ball algorithm converges. Meanwhile, the number of information exchanged in a centralized algorithm would be $\left(U+M+1\right)N$. 
However, it is worth noting that, each piece of information being exchanged in the heavy ball algorithm only includes a vector of dual variables or service request schemes. In contrast, every piece of information being exchanged in the centralized algorithm contains not only vectors of dual variables and service request/allocation schemes, but also vectors of QoS requirements, service demands, resource budgets, and other information. In other words, the total amount of information that should be exchanged within the heavy ball algorithm is considerably smaller. Moreover, the distributed iterative algorithm alleviates the privacy and security concerns in the ISDN, as the payoff optimization problems at solved locally at each network device, and, thus, unnecessary information exchanges are avoided. As stated next, a Walrasian equilibrium is guaranteed at the convergence of the heavy ball algorithm.

%At each iteration, the complexity of the distributed algorithm is O(NBNU).
%As for the exchanged information, at each iteration each BS broadcasts its j which is a relatively small
%real number, and each user reports its association request to only one BS which it wants to connect to.
%The amount of information to be exchanged in the distributed algorithm is k(NB + NU), where k is
%the number of iterations, while in the centralized method it is proportional to (NB  NU
\begin{lemma}\label{lemma2}
%\emph{The soldier achieves a higher expected material payoff at the PE compared to the expected  material payoff at the NE of the corresponding conventional dynamic game, in which the psychological aspect of the players' decision making processes is not considered.} 
\emph{The proposed dual decomposition-heavy ball based algorithm will converge to the Walrasian equilibrium, at which the optimal total payoff is reached in the system, with diminishing step sizes. }

\end{lemma}

\begin{proof} 
See Appendix D.
 \end{proof}

\section{Simulation Results and Analysis}

For our simulations, we consider a scenario with $U=60$ mobile users, $N=10$ BSs (including $5$ SBSs and $5$ DBSs), $M=1$ macro base station, and $1$ LEO satellite constellation. Note that, the height of the DBSs is assumed to be $200$ meters. %The channel gain between the soldier and each
%IoBT device follows a Rayleigh distribution with unit variance.
Other parameters used in the simulations are listed in Table I. The heavy ball based Walrasian equilibrium results are compared to the centralized optimization based results, the sub-gradient based results, and a random allocation scheme results considering the market clearance. All statistical results are averaged over a large number of independent runs.

%We evaluate the performance of the Hungarian based solution using simulations. We consider an ISDN system with $1500$ uniformly distributed IoBT devices, $1$ macro base station and $N1$ drone base stations. %The channel gain between the soldier and IoBT device follows a Rayleigh distribution with unit variance. %The BER $r\left(\gamma_l\right)$ of the soldier is calculated using the LTE system toolbox in MATLAB \cite {LTESYSTEM}.  %The IoBT network is set up based on ... . 
%The parameters are listed in Table  \uppercase\expandafter{\romannumeral1}. The results are compared to a random connection scheme. %baseline dynamic Bayesian game, {which is dynamic imperfect information game with only the material payoffs (with no psychology) of the soldier and attacker being considered.} %All statistical results are averaged over a large number of independent runs. Hereinafter, we use the term ``sum-rate'' to refer to the total downlink and uplink rates.

\begin{table}
 \vspace{-0.45cm}
  \newcommand{\tabincell}[2]{\begin{tabular}{@{}#1@{}}#2\end{tabular}}
\renewcommand\arraystretch{1}
 \caption{
    \vspace*{-0.2em}Simulation Parameters\cite{840210}}\vspace*{-1em}
\centering  
\begin{tabular}{|c|c|c|c|}% Â±Ã­ÃÅ¸Å¾Ã·ÃÃÃÂªÃÃÂ¶ÃÃÃ«Â·ÅÃÅÂ£Â¬left-l,right-r,center-c
\hline
\textbf{Parameter} & \textbf{Value} & \textbf{Parameter} & \textbf{Value} \\
\hline
$P_n $ & 20 dBm & $P_m$ & 43 dBm\\
\hline
$P_S$ & 9.23 dBW & $B_N $ & 56 MHz\\
\hline
%$B_S$ & 62.4 MHz & $B^N $ & 56 MHz\\
%\hline
$\sigma^2$ & -104 dB & $\Omega_c$ & 10.5354 dB\\
\hline
%$G^t_n\left(0\right)$ & 38 dBi & $G^t_s\left(0\right)$ & 42.1 dBi\\
%\hline

%$ \sigma ^2$ & -95 dBm & $ \\
%\hline
%$B_v$ & 1 GHz & &$h_{\min}$ & 100 m \\
%\hline
%$L$ & 1 Mbit & $P_{\max}$ & 20 W&$\delta_{S_i,n}$ & 5 Mbit/s\\
%\hline
%$\mu_\textrm{LoS}$,$\mu_\textrm{NLoS}$& 2, 2.4 &$\Omega_x, {S}$ &4, 12&$\chi  $ & 15\\
%\hline
%  & &$\zeta_1,\zeta_2$&0.5,0.5\\ 
%\hline
% $\chi _{\sigma_\textrm{NLoS}}$ & 5.27 & $\beta, \eta $ & 2, 100&$X,Y$& 11.9, 0.13\\ 
%\hline
\end{tabular}
 \vspace{-0.5cm}
\end{table}

%\begin{figure}[!t]
\setlength{\abovecaptionskip}{0pt} 
\setlength{\belowcaptionskip}{-15pt} 
%  \begin{center}
%   \vspace{0cm}
%    \includegraphics[width=7.4cm]{mismatchcvg.pdf}
%    \vspace{-0.2cm}
%    \caption{\label{figure2} Convergence of the heavy ball alogorithm.}
%  \end{center}
%  \vspace{-0.7cm} %è¿ä¸å¥å°±æ¯ç¼©è¡è·ç
%\end{figure}
\begin{figure}
\centering
\subfigure[]{
\label{fig:subfig:a} %% label for first subfigure
\includegraphics[width=8cm]{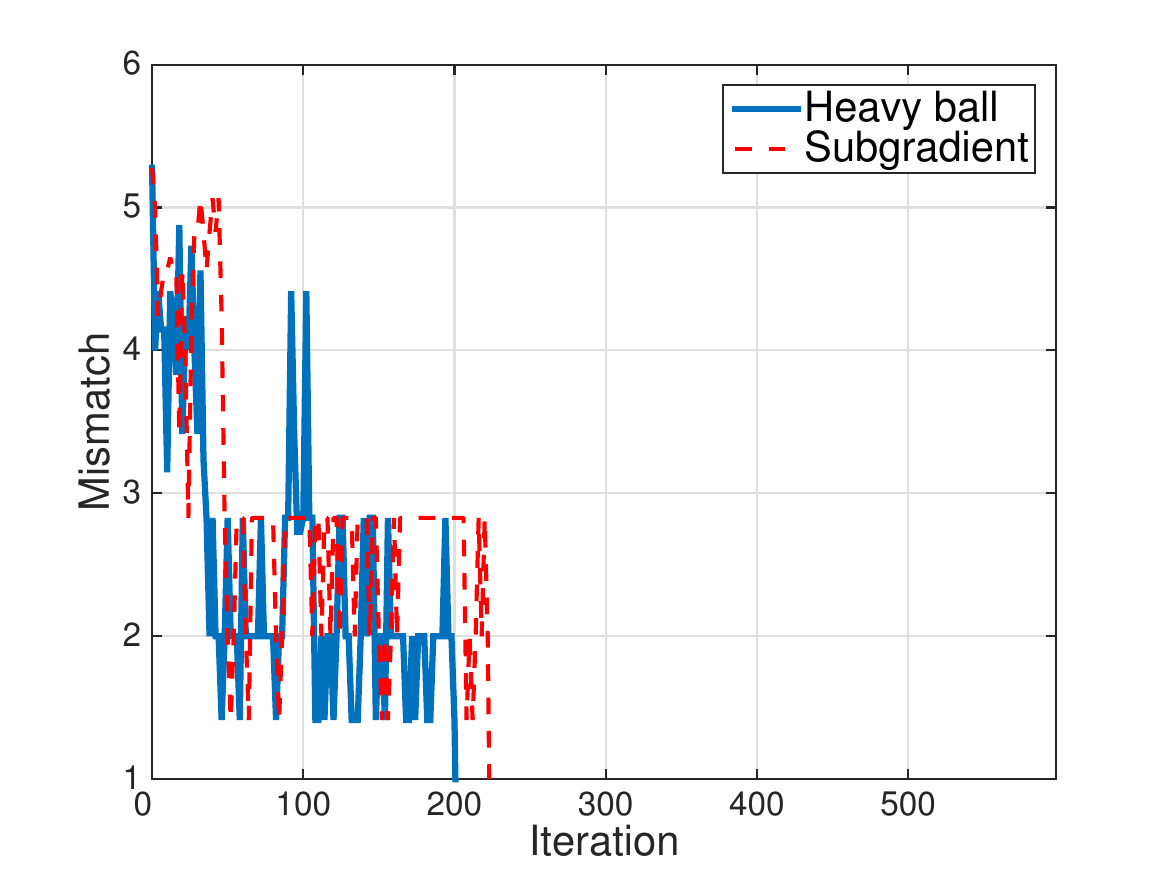}}
\subfigure[]{
\label{fig:subfig:b} %% label for second subfigure
\includegraphics[width=8cm]{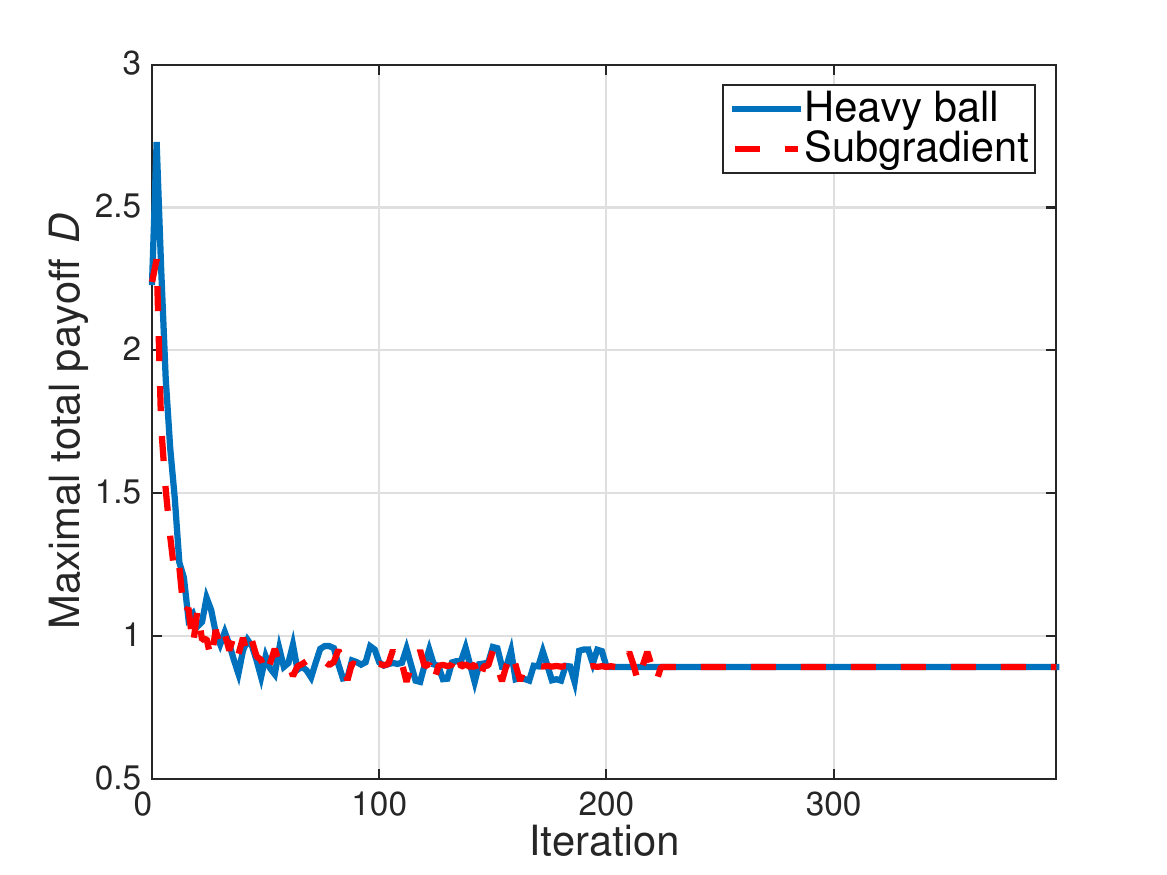}}
\caption{\footnotesize{Convergence of the heavy ball alogorithm.}}
\label{fig:subfig} %% label for entire figure
\vspace{-0.5cm}
\end{figure}
Fig. \ref{fig:subfig} shows the convergence of the proposed heavy ball algorithm. In the results shown in Fig. \ref{fig:subfig}, the heavy ball algorithm requires approximately 200 iterations to reach convergence, which is $10\%$ less than the number of iterations required for convergence of the sub-gradient algorithm. This stems from the fact that, the updating direction of the heavy ball algorithm has a smaller angle towards the optimal result than the sub-gradient algorithm. From Fig. \ref{fig:subfig:a}, we can see that, the mismatch between the supply and demand of data and backhaul service reaches $0$ upon the convergence of the heavy ball algorithm and sub-gradient algorithm where the market clearance constraint is satisfied. Meanwhile, Fig. \ref{fig:subfig:b} shows that the minimization of $D$ is reached upon the convergence of both algorithms where the total payoff maximization problem (\ref{eq:opt}) is solved. In such a case, a Walrasian equilibrium scheme is reached at the convergence of the heavy ball algorithm and, also, the sub-gradient algorithm. 
\begin{figure}[t]
  \begin{center}
   \vspace{0cm}
    \includegraphics[width=9.2cm]{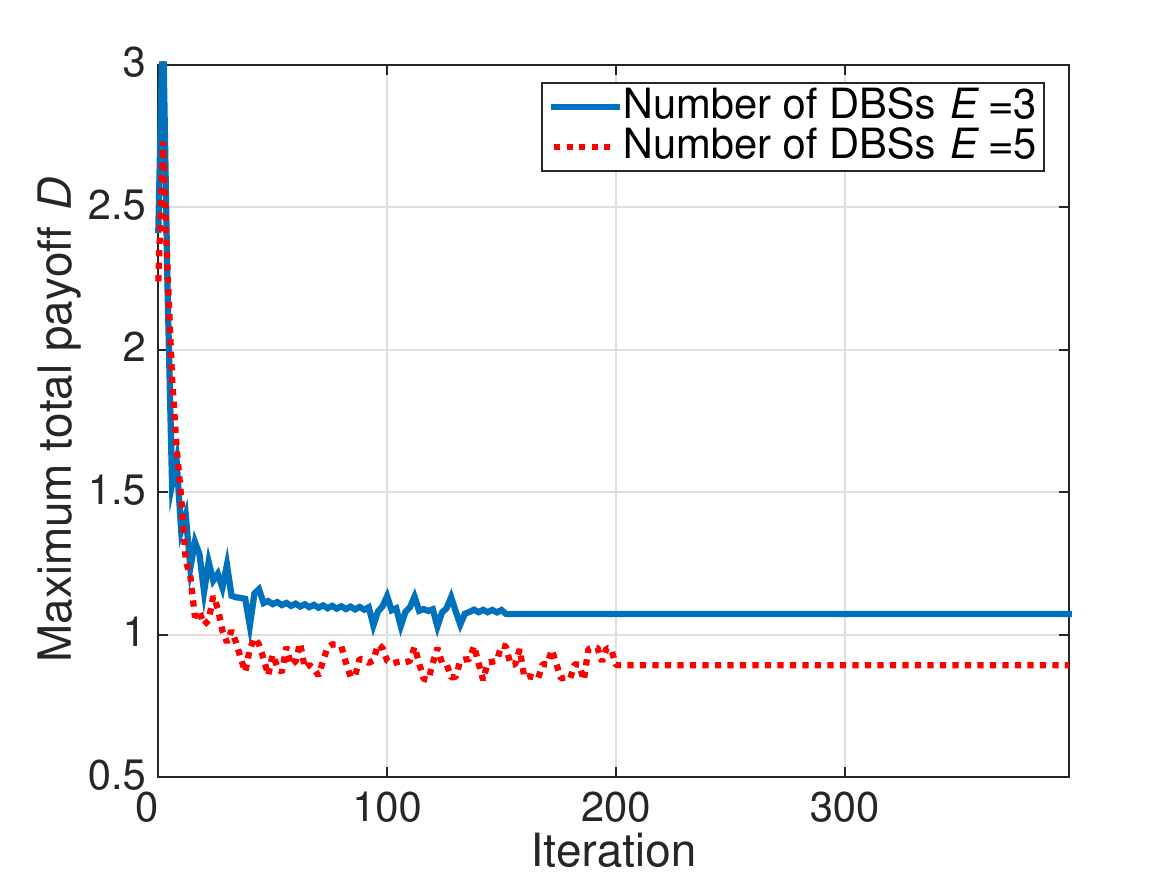}
    \vspace{-0.2cm}
    \caption{\label{complex} Convergence of heavy ball algorithm as the number of DBSs varies. }
  \end{center}\vspace{-0.7cm}
\end{figure}

Fig. \ref{complex} shows the convergence of the proposed heavy ball algorithm as the number of DBSs varies.  In Fig. \ref{complex}, we can see that as the number of DBS increases, the number of iterations needed for convergence increases. This stems from the fact that, as the number of DBSs increases, the number of communication links increases which requires more iterations to find the optimal user association and resource allocation schemes. Fig. \ref{complex} also shows that the maximum total payoff decreases as the number of DBSs increases. This is because the data rates at each communication link decrease, with the newly added DBSs introducing significant interference to the ISDN.

In Fig. \ref{figurep}, we show how, at a Walrasian equilibrium of the studied ISDN, the price of the data service at radio access links and backhaul service at backhaul links varies with the normalized data rate. As normalized data rate increases from $0.1$ to $0.5$, the prices for both a data service and backhaul service will increase. This is because the users and BSs would submit higher bids for high-rate data and backhaul service.  However, the price only slightly increases as normalized data rate increases from $0.3$ to $0.5$. This stems from the fact that, with the increase in its normalized data rate, a link gets the highest ``quote'' (or price) of its data or backhaul service for guaranteed adoption, henceforth, the users and BSs will stop proposing higher quotes for this link so as to maintain higher payoffs. 

\begin{figure}[t]
  \begin{center}
   \vspace{0cm}
    \includegraphics[width=9.2cm]{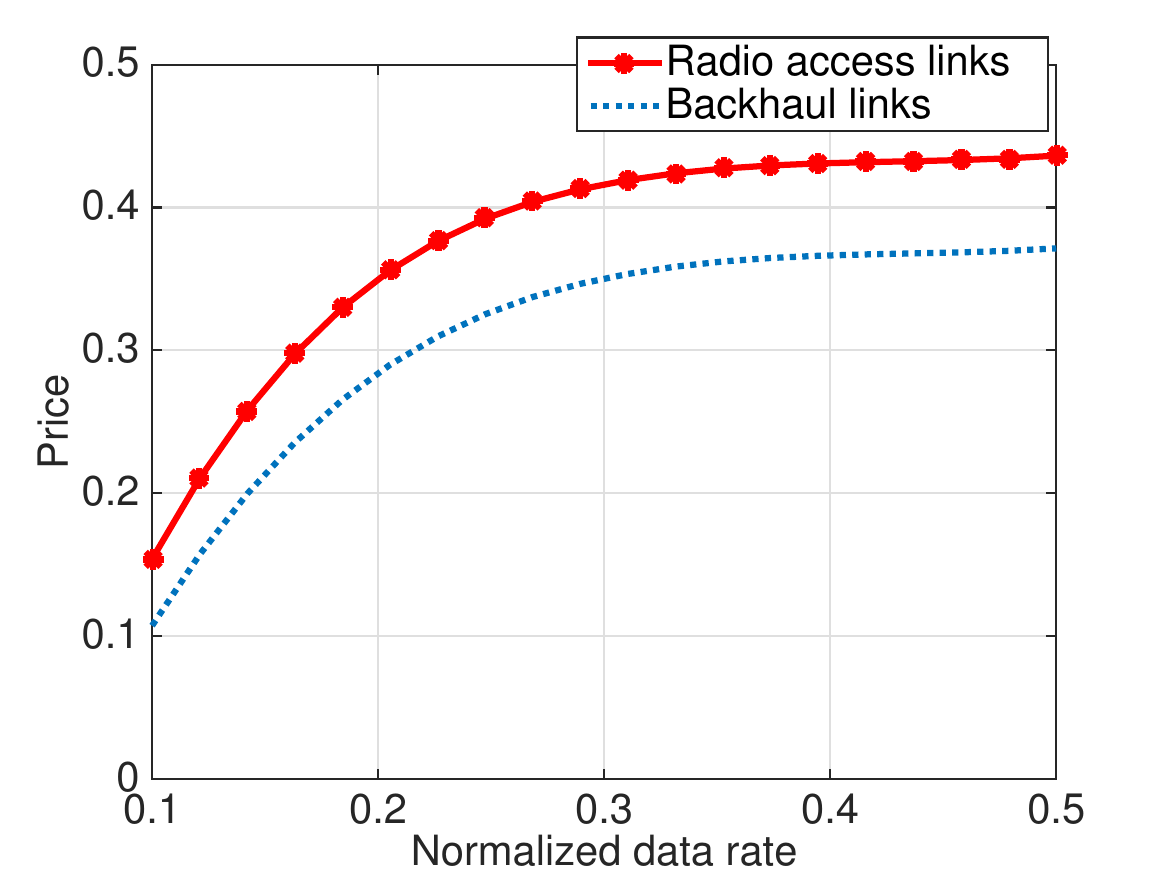}
    \vspace{-0.2cm}
    \caption{\label{figurep} Price as the normalized data rate varies. }
  \end{center}\vspace{-0.7cm}
\end{figure}

%\begin{figure}
%%\setlength{\belowcaptionskip}{-4pt}
%%\setlength{\abovecaptionskip}{-4pt}
%\centering
%%\renewcommand{\captionlabelfont}{\footnotesize}
%\subfigure[]{
%%\setlength{\belowcaptionskip}{-7pt}
%%\setlength{\abovecaptionskip}{-7pt} 
%\label{fig:subfig:a} %% label for first subfigure
%\includegraphics[width=6cm]{CDF.pdf}}
%\subfigure[]{
%%\setlength{\belowcaptionskip}{-7pt}
%%\setlength{\abovecaptionskip}{-7pt} 
%\label{fig:subfig:b} %% label for second subfigure
%\includegraphics[width=6cm]{CDF2.pdf}}
%\caption{\footnotesize{Data rate CDF vs.different algorithms.}}
%\label{fig:subfig} %% label for entire figure
%%\setlength{\belowcaptionskip}{-15pt}
%\vspace{-0.5cm}
%\end{figure}

%
%\begin{figure}[!t]
%\setlength{\abovecaptionskip}{-10pt} 
%\setlength{\belowcaptionskip}{-20pt}
%  \begin{center}
%   \vspace{0cm}
%    \includegraphics[width=7.4cm]{Fig4.pdf}
%    \vspace{-0.2cm}
%    \caption{\label{figure5} The soldier and attacker's strategy. ($v=3, X=10$).}
%  \end{center}\vspace{-0.7cm}
%\end{figure}
%\begin{figure}[t]
%%\setlength{\abovecaptionskip}{-0.05pt}
%\setlength{\belowcaptionskip}{-8pt}
%  \centering
%  \includegraphics[width=8.5 cm]{Fig3}
%  \renewcommand{\captionlabelfont}{\footnotesize}
%  \caption{\footnotesize{The soldier's material payoff as the weight $\omega$ of the attacker's frustration in the soldier's utility varies. ($X=20, v=3, \theta_1=\theta_2=0.5$).}}
%  \label{Fig. 1}
%  \centering
%\end{figure}
\begin{figure}
\centering
\subfigure[]{
\label{subfig2:a} %% label for first subfigure
\includegraphics[width=8cm]{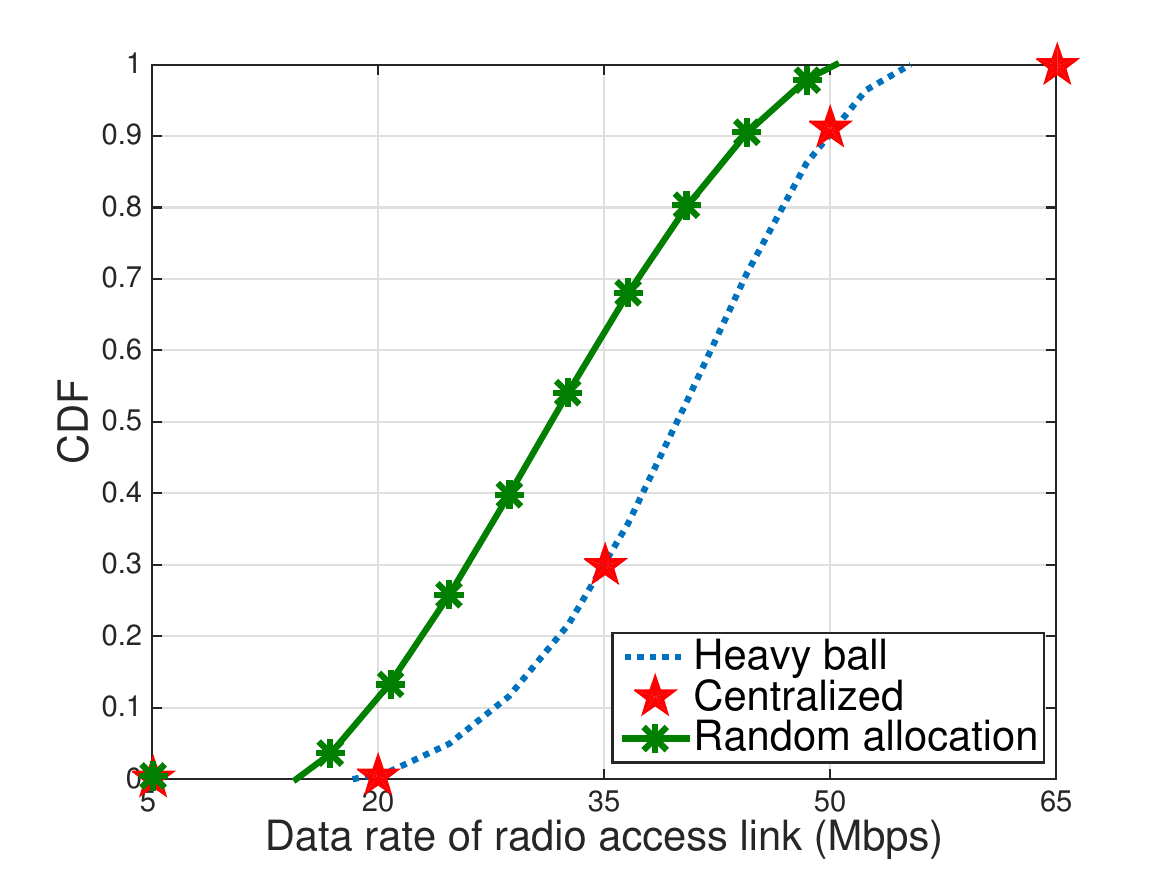}}
\subfigure[]{
\label{subfig2:b} %% label for second subfigure
\includegraphics[width=8cm]{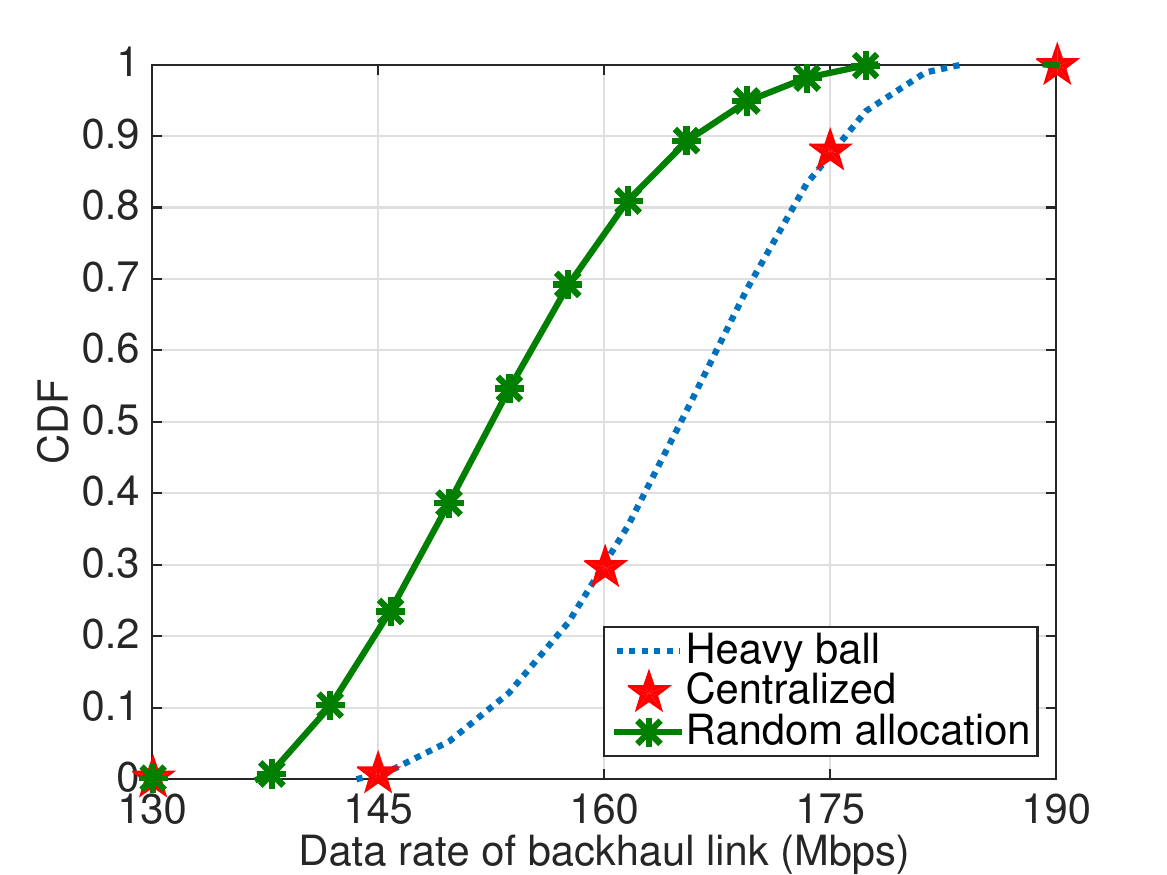}}
\caption{\footnotesize{Cumulative distribution function (CDF) in data rates.}}
\label{subfig2} %% label for entire figure
\vspace{-0.5cm}
\end{figure}

Fig. \ref{subfig2} shows the cumulative distribution function (CDF) in data rates resulting from all the considered schemes. In Fig. \ref{subfig2:a}, we can see that the proposed approaches achieves  $2.1$ times gains in the number of radio access links that operate over a $40$ Mbps rate, compared to the random allocation scheme. The proposed approach yields over $3.4$ times gains in the number of backhaul links with over $1.6$ Gbps data rate, compared to the random allocation scheme as shown in Fig. \ref{subfig2:b}. This means that more users and BSs receive higher data rates with the proposed solution. The main reason is that the proposed algorithm encourages communication links to associate with the optimal users and BSs, while fully exploiting the time resources. Meanwhile, Fig. \ref{subfig2} also shows that the proposed algorithm yields no loss on data rates, compared with the centralized optimization method.

%%In Fig. 3, we show how the average data rate of the radio access links varies with the number of DBSs. As the number of DBSs increases from $27$ to $35$, both the proposed solution and random connection solution yield an increased average data rate as the available spectrum resource in the system increases. However, the average data rate decreases as the number of DBSs increases from $35$ to $43$. This stems from the fact that as the number of DBSs increases, the distance between the users and BSs decreases, such that the interference of the users increases.
%%

\begin{figure}[t]

  \begin{center}
   \vspace{0cm}
    \includegraphics[width=9.2 cm]{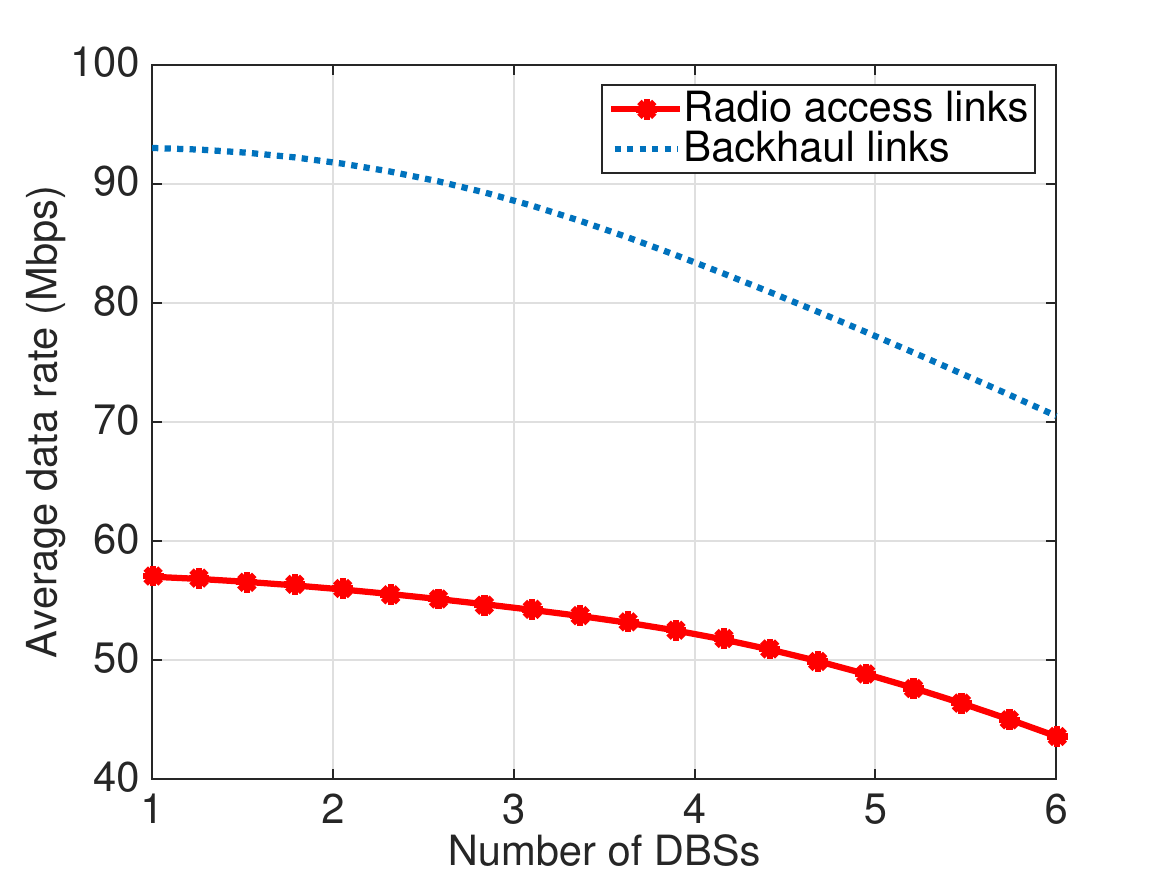}
    \vspace{-0.2cm}
    \caption{\label{figuredv} Average data rate as the number of DBSs varies. }
  \end{center}\vspace{-0.7cm}
\end{figure}

In Fig. \ref{figuredv}, we show how the average data rate at a Walrasian equilibrium varies with the number of DBSs. As the number of DBSs increases from $1$ to $6$, the average data rates at both radio access links and backhaul links decreases. This stems from the fact the newly added DBSs introduce significant interference to the system. Fig. \ref{figuredv} also shows that, the decrease in the average data at backhaul links is higher than the one at radio access links. This is because the interference introduced by a DBS over the backhaul links is larger than the one it introduces to the radio access links.%,  as the DBS locates way closer to other BSs than the terrestrial users.
%distance between a DBS and another BS is way much smaller than the one between the DBSs and terrestrial users. 
 %locate closer to other DBSs and SBSs.

\begin{figure}
  \begin{center}
   \vspace{0cm}
    \includegraphics[width=9.2cm]{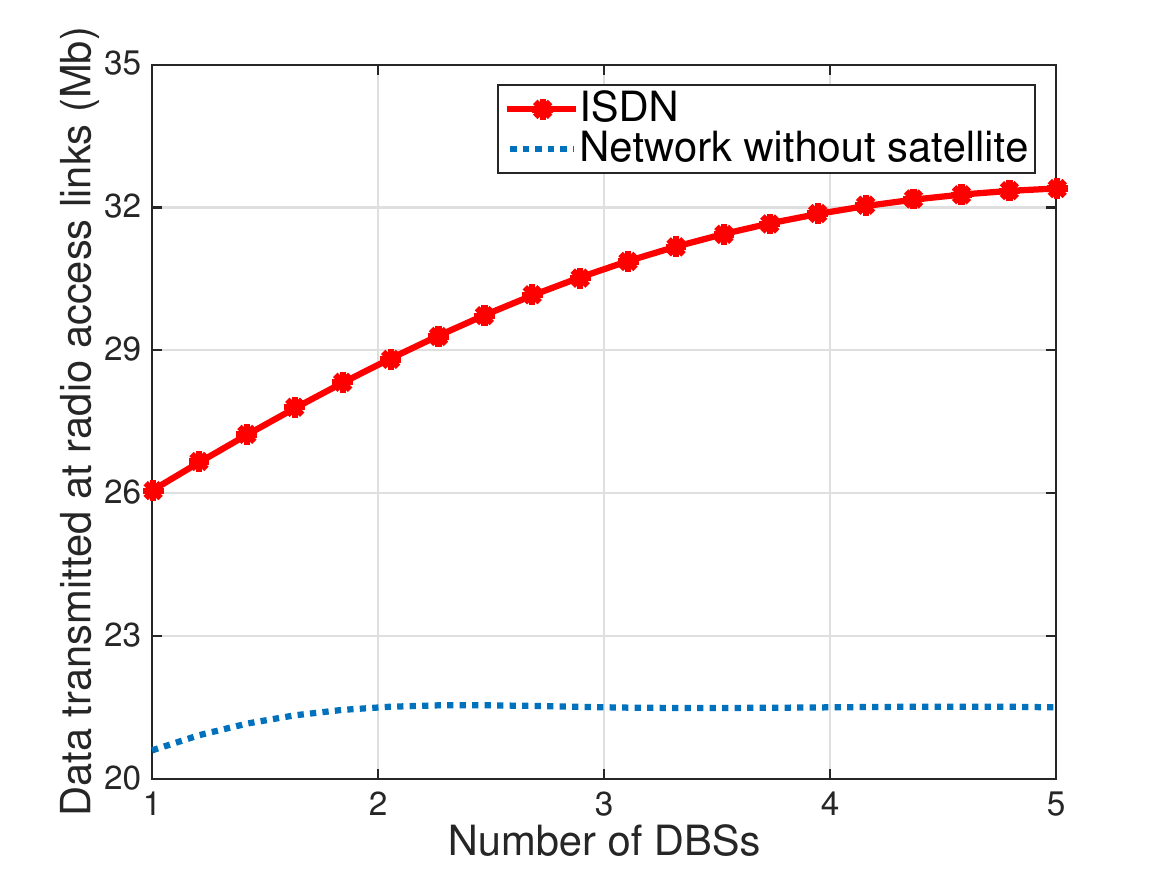}
    \vspace{-0.2cm}
    \caption{\label{figuret} Transmitted data at radio access links as the number of DBSs varies. }
  \end{center}\vspace{-0.7cm}
\end{figure}

Fig. \ref{figuret} shows how the amount of data service provided at radio access links vary as the number of DBSs increases. Fig.  \ref{figuret} shows that, as the number of DBS increases, the amount of data service provided at radio access links increases, with the increase of available time resources in the ISDN.  From Fig.  \ref{figuret}, we can also see that, as the number of DBSs increases, the amount of data service provided at radio access links increases at a lower speed. This is because a larger number of DBSs yields lower data rates. % and minimal increase on the available time resource at radio access links. This stems from the fact that, with the increase in the number of DBSs, the backhaul time resource in the system becomes more and more stringent. In such a case, some of the BSs should wait for backhaul service, and their available time resource for data service decreases directly.
%only after being served at the resource-limited backhaul links can a DBS provide data service to terrestrial users, the later a DBS gets served, the less amount of time resource it can provide to the users.
%incoming DBS would line up behind the deployed BSs to get backhaul service, only after which it can provides data service requested by the users. Thus, the amount of time resource can be provided by a incoming DBS decreases as the length of backhaul service waiting line piles up, or in other words, the last entered DBS can provide the least time resource to the radio access links. In such a case, the last entered DBS also provides the least data service to the radio access links, the increment speed of the amount of data service provided at radio access links decreases with the increase of the number of DBS in the system.
Fig. \ref{figuret} also shows that the amount of data service provided at the radio access links stops increasing as the number DBSs increases from $3$ to $6$ when the satellite is not deployed in the system. This is because, as the number of DBSs increases, the only deployed MBS is not capable of serving all the BSs, which leads to some users not being served.

\begin{figure}[t]
  \begin{center}
   \vspace{0cm}
    \includegraphics[width=9.2 cm]{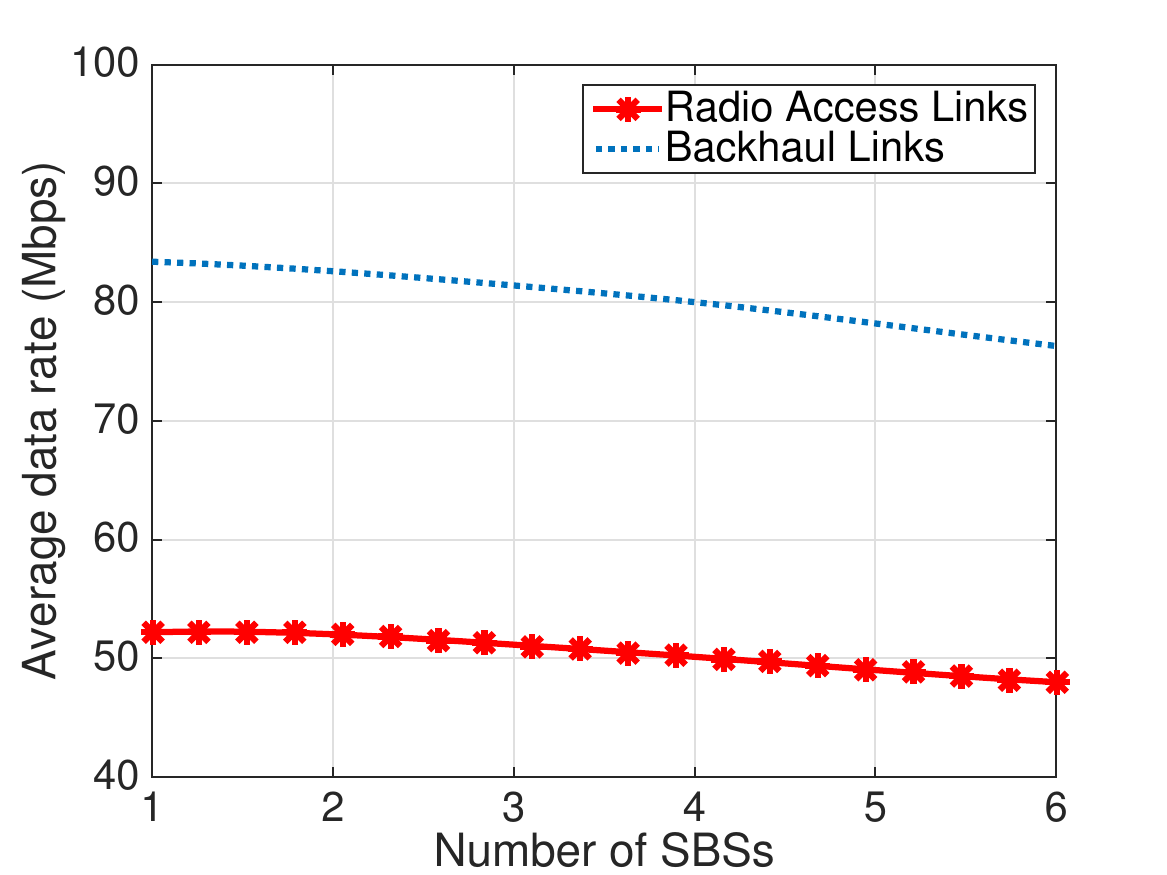}
    \vspace{-0.2cm}
    \caption{\label{figuresbs} Average data rate as the number of SBSs varies. }
  \end{center}\vspace{-0.7cm}
\end{figure}
%the Walrasian equilibrium solution yields a decrease on average data rate on both radio access links and backhaul links, with the increased interference in the system. Fig. \ref{figuredv}  shows that, as the DBSs can always provide LoS signals to the users and BSs in the ISDN, the variation of DBSs' number can significantly influence the communication performance of the communication links in ISDN.
Fig. \ref{figuresbs} shows how the average data rate at a Walrasian equilibrium varies with the number of SBSs. In Fig. \ref{figuresbs}, we can see that, as the number of SBSs increases from $1$ to $6$, the average data rates at both radio access links and backhaul links decreases. This stems from the fact the newly added SBSs introduce interference to the system. However, the decrease in data rates caused by adding more SBSs is much smaller than the one caused from increasing the number of DBSs. This is due to the fact that DBSs can provide LoS links to every other devices in the ISDN and, hence,  they can introduce a DBS can introduce a significantly larger interference over the backhaul and radio links.

%as the backhaul time resource becomes stringent after satellite being removed from the system, some incoming DBS may never get backhaul service and could not provide any data service to the radio access links.

\vspace{-0.12cm}
\section{Conclusion}
In this paper, we have introduced a novel joint resource allocation and user association problem in an ISDN system in which satellite backhaul links share time resources with terrestrial backhaul links and terrestrial radio access links. We have formulated this problem using a competitive market in which the data rate of the communication links is optimized distributively with a heavy ball algorithm. Simulation results have shown that, by considering the Walrasian equilibrium in the market, the proposed algorithm can yield over two-fold gain in terms of  the number of radio access links with over $40$ Mbps rates, and over three-fold gain in terms of the number of backhaul links with over $1.6$ Gbps rates. The heavy ball algorithm also yields a $10\%$ improvement in the convergence speed, compared to the sub-gradient algorithm.
 \
%Due to the reliance of the players' actions on their beliefs, we have used the psychological forward induction to solve this game. The psychological game enables the soldier to determine its actions based on its estimation on the attacker's behavior and belief. Simulation results have shown that, by explicitly intending to frustrate its opponent, the soldier can get a better payoff than the case in which a conventional Bayesian game is used.
\vspace{0.2cm}
 \section*{Appendix}
 {
\subsection{Proof of Proposition \ref{pp1}}\label{Bp:b}  
\begin{proof} 
{Based on (\ref{eq:channel}), (\ref{eq:cnst}), and (\ref{eq:reinterference3}), the data rate of link $I^S_{n}$ at slot $t$ can be given as:
\begin{equation}\label{eq:pv1}
\setlength{\abovedisplayskip}{5 pt}
\setlength{\belowdisplayskip}{5 pt}
\begin{split}
c_{n,S,t}=\beta_{nt}B_S\log\left(1+\frac{P_S {L_{1}10^{\frac{-\left(\alpha'+\alpha\log_{10}d_{n,S}\left(t\right)+\chi\right)}{10}}Q_{MM}Q_{RM}}}{\Omega_{n,S,t}+\Omega_c+\sigma^2}\right),
\end{split}
\end{equation} 
where $d_{n,S}\left(t\right)=\sqrt {\left(x_0+\upsilon t \tau-x_{2,n}\right)^2+\left(y_0-y_{2,n}\right)^2+\left(f_0-f_{2,n}\right)^2}$ is the distance between the satellite and BS $n$ at slot $t$, with $\frac{\partial d_{n,S}\left(t\right)}{\partial t}=\frac{\left(x_0-x_{2,n}\right)\upsilon\tau+\upsilon^2\tau t}{d_{n,S}\left(t\right)}$. Also, we define function $g\left(t\right)$ in a form of:
 \begin{equation}\label{eq:pv2}
\setlength{\abovedisplayskip}{5 pt}
\setlength{\belowdisplayskip}{5 pt}
\begin{split}
g\left(t\right)=10^{\frac{-\left(\alpha'+\alpha\log_{10}d_{n,S}\left(t\right)+\chi\right)}{10}},
\end{split}
\end{equation} 
such that:
 \begin{equation}\label{eq:pv3}
\setlength{\abovedisplayskip}{5 pt}
\setlength{\belowdisplayskip}{5 pt}
\begin{split}
\frac{\partial g\left(t\right)}{\partial t}=-\frac{\alpha\left(x_0\upsilon\tau+\upsilon^2\tau t\right)}{10d^2_{n,S}\left(t\right)}10^{\frac{-\left(\alpha'+\alpha\log_{10}d_{n,S}\left(t\right)+\chi\right)}{10}}=-\frac{\alpha\left(x_0\upsilon\tau+\upsilon^2\tau t\right)}{10d^2_{n,S}\left(t\right)}g\left(t\right).
\end{split}
\end{equation} 
Denote $P'_S=\frac{P_SL_1Q_{MM}Q_{RM}}{\Omega_{n,S,t}+\Omega_c+\sigma^2}$, such that $c_{n,S,t}\left(t\right)=\beta^*_{nt}B_S\log_2\left(1+P'_S g\left(t\right)\right)$. In such a case:
 \begin{equation}\label{eq:pv3}
\setlength{\abovedisplayskip}{5 pt}
\setlength{\belowdisplayskip}{5 pt}
\begin{split}
\frac{\partial c_{n,S,t}\left(t\right)}{\partial t}=\frac{\beta^*_{nt}B_S}{\ln 2}\frac{P'_S g'\left(t\right)}{1+P'_S g\left(t\right)}=-\frac{\alpha\beta^*_{nt}B_SP'_S}{10\ln 2}\frac{\left(\left(x_0-x_{2,n}\right)\upsilon\tau+\upsilon^2\tau t\right) g\left(t\right)}{\left(1+P'_S g\left(t\right)\right)d^2_{n,S}\left(t\right)}.
\end{split}
\end{equation} 
As each studied time slot has a rather short time duration $\tau$, we can approximate the variation of $c_{n,S,t}$ in one time slot as $\frac{\partial c_{n,S,t}\left(t\right)}{\partial t}\tau$. This completes the proof.
}
\end{proof}}
\subsection{Proof of Lemma \ref{lemma1}}
\begin{proof} We first, relax the integer constraints ({\ref{eq:opt}}p)-({\ref{eq:opt}}r) into ${\rho}_{nut},  {\theta}_{nut}, {\delta}_{nmt}, {\varphi}_{nmt},{\beta}_{nt}, {\varepsilon}_{nt}\in \left[0,1\right]$, and, thus, the maximization problem ({\ref{eq:opt}}) becomes a convex problem. Given the bounded, closed feasible region defined by ({\ref{eq:opt}}a)-({\ref{eq:opt}}o), ({\ref{eq:opt}}s), and ${\rho}_{nut},  {\theta}_{nut}, {\delta}_{nmt}, {\varphi}_{nmt},{\beta}_{nt}, {\varepsilon}_{nt}\in \left[0,1\right]$, the relaxed continuous convex problem must have one optimal solution. Also, considering that the feasible region defined by ({\ref{eq:opt}}a)-({\ref{eq:opt}}s) is not empty, for each optimal solution to the relaxed convex problem, there must exist at least one integer optimal solution to the INLP problem in (\ref{eq:opt}) based on \cite{cook1986sensitivity}. This completes the proof.
 \end{proof}
 \subsection{Proof of Lemma \ref{lemma2}}
 \begin{proof} 
We first observe that, with the non-summable diminishing updating  step size $\pi^{\left(k\right)}$, the heavy ball algorithm is guaranteed to converge to optimal solutions of the studied dual minimization problem based on \cite{camerini1975improving}. With the convergence of the heavy ball algorithm at iteration $K^*$, market clearance violations tend to be $0$, meaning that  $\boldsymbol{\theta}^{\left(K^*\right)}=\boldsymbol{\rho}^{\left(K^*\right)}$, $\boldsymbol{\varphi}^{\left(K^*\right)}=\boldsymbol{\delta}^{\left(K^*\right)}$ and $\boldsymbol{\varepsilon}^{\left(K^*\right)}=\boldsymbol{\beta}^{\left(K^*\right)}$, with $\boldsymbol{\theta}^{\left(K^*\right)}$, $\boldsymbol{\rho}^{\left(K^*\right)}$, $\boldsymbol{\varphi}^{\left(K^*\right)}$, $\boldsymbol{\delta}^{\left(K^*\right)}$, $\boldsymbol{\varepsilon}^{\left(K^*\right)}$, and $\boldsymbol{\beta}^{\left(K^*\right)}$ being, respectively, the optimal solutions of payoff maximization problems (\ref{eq:optu})-(\ref{eq:opts}) at iteration $K^*$. Thus,  at the convergence of the heavy ball algorithm, a Walrasian equilibrium in the competitive market is reached. 

Also, as the constraints are all linear equalities and inequalities, and the feasible region defined by these constrains have at least one interior point, the strongly duality between the primal total payoff maximization problem (\ref{eq:opt}) and its dual minimization problem holds based on the Slater's condition  \cite{boyd2004convex}. Thus, as the heavy ball algorithm converges, the optimal solutions of the total payoff maximization problem (\ref{eq:opt}) must be reached with the heavy ball algorithm. This completes the proof.
 \end{proof}

\vspace{-0.12cm}

\bibliographystyle{IEEEbib}
\def\baselinestretch{1.42}
\bibliography{references}
\end{document}